\documentclass[11pt, reqno, a4paper]{amsart}
\usepackage[T1]{fontenc}
\usepackage[table,xcdraw]{xcolor}
\usepackage{braket}
\usepackage{graphicx}
\usepackage{tikz}
\usepackage{centernot}
\usepackage{amssymb,amsthm}
\usepackage[margin=1in]{geometry}
\usepackage[colorlinks = true, linkcolor = blue, urlcolor  = blue, citecolor = red]{hyperref}
\usepackage{subfig}
\usepackage{mathtools}
\renewcommand{\epsilon}{\varepsilon}
\renewcommand{\phi}{\varphi}
\newcommand{\C}{\mathbb{C}}

\DeclareMathOperator{\Tr}{Tr}

\newcommand{\ketbra}[2]{|#1\rangle\langle#2|} 
\newtheorem{theorem}{Theorem}[section]
\newtheorem{definition}[theorem]{Definition}
\newtheorem{proposition}[theorem]{Proposition}
\newtheorem{corollary}[theorem]{Corollary}

\newtheorem{lemma}[theorem]{Lemma}
\newtheorem{remark}[theorem]{Remark}

\DeclareMathAlphabet\mathbfcal{OMS}{cmsy}{b}{n}
\let\originallesssim\lesssim
\let\originalgtrsim\gtrsim

\DeclareRobustCommand{\lesssim}{%
  \mathrel{\mathpalette\lowersim\originallesssim}%
}
\DeclareRobustCommand{\gtrsim}{%
  \mathrel{\mathpalette\lowersim\originalgtrsim}%
}

\makeatletter
\newcommand{\lowersim}[2]{%
  \sbox\z@{$#1<$}%
  \raisebox{-\dimexpr\height-\ht\z@}{$\m@th#1#2$}%
}
\makeatother

\begin{document}
\title[Efficient Entanglement Manipulation Cookbook]{Efficient Entanglement Manipulation Cookbook: From Mixing to a Computational Second Law}

\author{Ilia Ryzov, Faedi Loulidi, David Elkouss}
\email{ilia.ryzov@oist.jp}
\email{faedi.loulidi@oist.jp}
\email{david.elkouss@oist.jp}
\address{Okinawa Institute of Science and Technology, Okinawa, Japan}
\date{\today}

\begin{abstract}
 Quantum entanglement is an essential resource for implementing practical quantum devices. However, for the resource to be useful in large-scale applications, it must be accessible to parties with limited computational resources. 
Computational entanglement measures quantify the usefulness of entanglement in the presence of limited computational resources in the practical setting of non-asymptotic manipulations.
In this paper, we build a framework of non-asymptotic efficient entanglement manipulations by systematically analyzing a wide range of properties of two recently introduced computational entanglement measures: the computational one-shot distillable entanglement and cost. 
To do so, we adapt definitions of some common properties of scalar entanglement measures to the case in which only lower or upper function bounds are available.
Using the developed framework, we then proceed to investigate the mathematical behavior of the newly introduced entanglement measures under probabilistic state mixing and tensor-product operations.    
Furthermore, we analyze how the values of both measures change under local unitary and LOCC transformations. For this, we derive an explicit quantitative connection between the geometry of projective-net-generated families of states and their computational one-shot distillable entanglement, and use this connection to establish an unconditional separation between computational and information-theoretic one-shot distillable entanglement.
Moreover, we investigate the robustness of efficiently accessible entanglement against noise and derive the first quantitative continuity bound on computational distillable entanglement. Finally, as an application of the developed ``dictionary'' of properties, we establish a fundamental relation between computational entanglement measures that we call the ``Second Law of Efficient Entanglement Manipulation''.
Our analysis covers both the one-shot scenario and the uniform setting, with properties established for the former naturally extending to the latter.
\end{abstract}
\maketitle
\tableofcontents
\newpage
\section{Introduction}
Quantum theory has demonstrated its practical significance by enabling a new generation of information-processing technologies~\cite{PK11,PA20}, including devices capable of performing tasks that are inaccessible to their classical counterparts~\cite{LB25}. At the same time, current developments in quantum technologies are increasingly directed toward their upscaling. This trend can be observed in the construction of quantum processors with growing numbers of qubits, recent advances in quantum error correction, and the development of modular and networked quantum architectures~\cite{G25,MNB25,MDN25}. Although these developments concern different physical platforms and applications, they share the problem of controlling quantum resources as the size and complexity of the underlying devices increase. In particular, entanglement constitutes a central resource for quantum computation and error correction within individual processors, as well as for communication, distributed computation, cryptography, and sensing in quantum networks~\cite{ZZH25,GB20,K26,LX23,XYM22,B24}.

At the same time, for entanglement to be useful as a practical resource, it must be generated, transformed, verified, and consumed by computationally constrained devices. In quantum processors, these tasks arise in the implementation of entangling operations, encoded transformations, error-correction procedures, and resource-state preparation. Similarly, in quantum networks, link-layer and end-to-end protocols formulate entanglement generation and delivery as a network service~\cite{DSC19,KDW20,PDT22}, while entanglement routing may additionally involve entanglement swapping, purification, scheduling, and resource-allocation procedures~\cite{LLLC21,CERW20,PKT19}. Together, these developments frame the complexity of entanglement manipulation as part of the broader problem of scaling quantum technologies. They suggest that, to characterize entanglement as a resource in large quantum systems, one must account not only for the amount of entanglement contained in a state, but also for whether it can be accessed using the computational resources available to the parties.

Traditionally, quantum entanglement has been studied through an information-theoretic lens \cite{horodecki2009quantum}, which idealizes parties as having unbounded local computational power: local operations and classical communication (LOCC) protocols may invoke arbitrarily large coherent unitaries and unrestricted classical post-processing. Therefore, information-theoretic entanglement measures do not, in general, account for whether the corresponding entanglement manipulations can be implemented efficiently in a large-scale quantum system. This motivated the development of a \emph{new theory of entanglement} that explicitly accounts for the computational limitations of the parties.

To bridge this gap, Arnon-Friedman et al.~\cite{arnon2023computational} proposed a \emph{computational resource theory of entanglement} restricting LOCC protocols to polynomial-size quantum circuits. They studied the behavior of entanglement in the proposed computationally restricted setting by introducing new types of measures that quantify entanglement via efficient operational tasks -- computational one-shot distillable entanglement and cost -- which serve as extensions of the familiar information-theoretic quantifiers~\cite{BD10,BD11}. Subsequent work by Leone et al.~\cite{leone2025entanglement} developed a complementary $k$-copy framework for computational entanglement manipulation of bipartite pure states. Allowing approximate LOCC protocols with nonzero failure probability, they characterized worst-case distillation and dilution rates in state-agnostic and state-aware settings and identified the min-entropy as the quantity governing efficient pure-state distillation. 

Recently, computational frameworks have begun to attract more attention within the broader quantum information community. Subsequent works explored not only the effects of complexity restrictions on measurement theory~\cite{YHK25}, entropic quantities~\cite{AHR25,MR26}, and nonlocality~\cite{GBG23}, but also the applications of computational entanglement measures in cryptography~\cite{GE24,RGLE26}. Despite these advances, the question of how entanglement behaves under computational restrictions in the non-asymptotic setting, which was posed in the original work~\cite{arnon2023computational}, remains largely unexplored. Given the central role of entanglement in scalable quantum technologies, understanding the implications of real-world efficiency limitations on entanglement theory is a crucial step toward implementing large-scale quantum devices. In this work, we address this question and study the behavior of computational entanglement of mixed states under general manipulations in a non-asymptotic setting. In particular, we focus on the following broad open question:

\begin{center}\em
    How does the amount of efficiently accessible entanglement behave under general transformations in the one-shot setting? 
\end{center}

We answer this question by providing a systematic analysis of the behavior of the original one-shot computational measures, introduced in~\cite{arnon2023computational}, under \emph{state composition}, \emph{transformations}, and \emph{finite perturbations}, for general families of mixed states. For this, we first extend familiar properties of scalar entanglement measures to measures defined through lower/upper bounds, which are utilized in the computational framework. With appropriate extensions developed, we proceed to do the following:
\begin{itemize}
    \item Derive the behavior of computational entanglement measures under state mixing and tensor-product operations.
    \item Describe quantitatively how noise in state preparation affects the value of computational one-shot distillable entanglement. We refer to this observation as the continuity bound. 
    \item Derive a quantitative connection between the geometrical structure of projective-net-generated families and their computational one-shot distillable entanglement, and use it to establish an unconditional separation previously considered in~\cite[Lemma~4.1]{arnon2023computational}. 
    \item Establish the behavior of computational entanglement measures under LOCC and local unitary transformations.
    \item Provide a fundamental relation between the values of computational one-shot distillable entanglement and cost, which we refer to as the ``Second Law of Efficient Entanglement Manipulation.'' This extends the result of~\cite{W21}.
\end{itemize}

Our results have strong implications for both the newly established computational entanglement framework and practical applications that utilize entanglement as a resource. On the practical side, to the best of our knowledge, our work provides the first systematic description of how efficiently accessible entanglement behaves under common operational transformations relevant to  large-scale quantum devices. Furthermore, our work contributes to the fundamental development of the newly established computational entanglement framework~\cite{arnon2023computational}. First, our results illuminate how exactly computational restrictions affect the amount of efficiently accessible entanglement. Precisely, we show in a quantitative way that, for the projective-net-generated families studied here, their \emph{geometric separation controls an explicit upper bound} on efficiently distillable entanglement. This quantitatively establishes existential separations between information-theoretic and computational entanglement frameworks in the non-asymptotic setting without relying on any computational assumptions. A related unconditional separation was previously considered in~\cite[Lemma~4.1]{arnon2023computational} through a counting argument. Our results complement this argument and introduce an explicit dependence of the separation on the geometry of the underlying projective-net-generated family. Second, the established ``second law'' \emph{quantitatively connects computational one-shot distillation and dilution bounds} in the finite-error regime, while the continuity bound \emph{quantifies the robustness of one-shot efficiently accessible entanglement} under noise. Finally, the proof technique for the continuity bound of computational one-shot distillable entanglement may be of independent interest, as it establishes an explicit way to achieve the necessary precision for efficient entanglement distillation in the presence of noise that corrupts the input family of states.

The remainder of this work is organized as follows. In Section \ref{sec:main results}, we summarize our main results. In Section \ref{sec: intro information theory}, we introduce the necessary preliminaries used in this work.  In Section \ref{sec: comp entanglement}, we recall definitions of the computational entanglement measures introduced in~\cite{arnon2023computational}. We then extend some common properties of scalar entanglement measures to their functional counterparts. In Section \ref{sec: basic operations}, we establish the behavior of the computational entanglement measures under some basic operations such as mixing and taking the tensor product of states.  In Section \ref{sec: complexity vs entanglement}, we obtain a quantitative dependence of computational one-shot distillable entanglement on the geometrical structure of the distilled family of states. In Section \ref{sec: invariance local unitaries}, we analyze local unitary invariance and LOCC monotonicity of computational one-shot entanglement measures. In Section \ref{sec: val-cont}, we investigate how computational one-shot distillable entanglement changes when the input family is perturbed by noise, which we refer to as the continuity bound. We then combine the results of this work to derive a ``Second Law of Efficient Entanglement Manipulation'' that quantitatively relates computational one-shot distillable entanglement to cost.

\section{Main Results} \label{sec:main results}
In this work, we investigate how the computational entanglement of quantum states behaves under common manipulations in the non-asymptotic setting. For this, we analyze the mathematical properties of the \emph{computational one-shot distillable entanglement} and the \emph{computational one-shot entanglement cost}~\cite{arnon2023computational}, which quantify, under computational constraints, the number of EPR pairs one can extract from a given state and the number of EPR pairs needed to prepare it, respectively. We begin by presenting properties that entanglement measures may satisfy. We then recall informal definitions of the computational entanglement measures and summarize the main results.

Let $\mathcal{D}(\mathcal{H}_{AB})$ denote the set of quantum states shared between Alice and Bob, where $\mathcal{H}_{AB} := \mathcal{H}_A \otimes \mathcal{H}_B$ with $\mathcal H_A,\mathcal H_B$ being Hilbert spaces of Alice and Bob, respectively. Let $E$ be an arbitrary entanglement measure:
\begin{align*}
E : \mathcal{D}(\mathcal{H}_{AB}) &\to \mathbb{R}^+ \\
\rho_{AB} &\mapsto E(\rho_{AB}) \geq 0.
\end{align*}
Entanglement measures are typically evaluated under the paradigm of \emph{local operations and classical communication} (LOCC). An LOCC protocol consists of an arbitrary sequence of local quantum operations performed independently by each party and followed by  classical communication. 

An entanglement measure $E$ may satisfy some of the following properties:
\begin{itemize}
    \item \emph{Convexity/Concavity}: $E$ is convex (concave) if for any finite set of bipartite states $\{\rho^x_{AB}\}_{x\in X}$ and for any probability distribution $p$  over $X$, the following holds:   
    \begin{equation*}
        E\left(\sum_{x\in X}p(x)\rho^x_{AB}\right)\underset{\displaystyle (\geq)}{\leq} \sum_{x\in X}p(x)E(\rho^x_{AB}).
    \end{equation*}

    \item \emph{Superadditivity/Subadditivity}: $E$ is superadditive (subadditive) if for any states \(\rho_{AB}\) and \(\sigma_{AB}\) shared between Alice and Bob the following holds:
    \begin{equation*}
        E(\rho_{AB}\otimes\sigma_{AB})\underset{\displaystyle (\leq)}{\geq} E(\rho_{AB})+E(\sigma_{AB}).
    \end{equation*}
    \item \emph{LOCC monotonicity}: $E$ is an LOCC monotone if for any LOCC protocol $\Gamma$ and state $\rho_{AB}$,
    \begin{equation*}
    E(\rho_{AB}) \geq E(\Gamma(\rho_{AB})).
    \end{equation*}
    
\item \emph{Invariance under local unitaries}: $E$ is invariant under local unitaries if for any shared state $\rho_{AB}$ and for any unitary $U_{AB}:=U_A\otimes U_B$ we have 
\begin{equation*}
E(U_{AB}\,\rho_{AB}\,U^{\dagger}_{AB})=E(\rho_{AB})
\end{equation*}
\item \emph{Continuity Bound}: $E$ exhibits a continuity bound if for any $\rho_{AB}$ and $\sigma_{AB}$ the following holds:
\begin{equation*}
    |E(\rho_{AB}) - E(\sigma_{AB})| \leq f(T,\cdot),
\end{equation*}
where $T$ is the distance between states.
\end{itemize}

While these properties are well studied in the context of entanglement measures, many open problems remain (see~\cite[Table 1]{brandao2011faithful}).

Computational entanglement theory~\cite{arnon2023computational} introduces entanglement measures under computational assumptions. Specifically, it restricts LOCC operations to those that can be implemented by circuits with a polynomial number of gates. 

\begin{definition}[Informal of Definition~\ref{def: definition efficient LOCC}] A family of LOCC channels $\{\hat{\Gamma}^\lambda\}$ is efficient if every LOCC channel in the family admits a circuit representation and the total number of gates, measurements and ancilla preparations in these circuit representations grows at most polynomially with $\lambda$.    
\end{definition}

Figure~\ref{fig:circuit LOCC} illustrates such circuit representations. The definitions of computational distillable entanglement and entanglement cost introduced in~\cite{arnon2023computational} (see also Definitions~\ref{def: comp distil vidick} and~\ref{def: comp entang cost vidick}) give bounds on the number of EPR pairs that can be extracted from, or needed to reconstruct, a given state.

\begin{definition}[Informal of Definition~\ref{def: comp distil vidick} and~\ref{def: comp entang cost vidick}]
Let $\{\rho^\lambda_{AB}\}$ be a family of bipartite quantum states. 
A function \(m(\lambda)\) is a valid lower bound on the computational distillable entanglement \(\hat{E}_D^\epsilon\) for the family \(\{\rho^\lambda_{AB}\}\) if there exists an efficient family of LOCC channels that extract at least \(m(\lambda)\) EPR pairs from each \(\rho^\lambda_{AB}\) with error at most \(\epsilon(\lambda)\).

A function \(n(\lambda)\) is a valid upper bound on the computational entanglement cost \(\hat{E}_C^\epsilon\) for the family \(\{\rho^\lambda_{AB}\}\) if there exists an efficient family of LOCC channels that prepare each \(\rho^\lambda_{AB}\) from at most \(n(\lambda)\) EPR pairs with error at most \(\epsilon(\lambda)\).
\end{definition}

The main goal of this work is to analyze different mathematical properties of both the computational distillable entanglement $\hat E^{\epsilon}_D$ and the computational entanglement cost $\hat E^{\epsilon}_C$. However, the conventional desired properties of entanglement measures (such as continuity bounds, convexity, and so on) are defined assuming that entanglement measures are scalar functions. The computational one-shot distillable entanglement and cost, however, are given by valid function lower and upper bounds, respectively. Therefore, to analyze different mathematical properties, we need to extend classical properties to accommodate such measures. For this, we define \emph{lower and upper bound} counterparts of some selected properties of entanglement measures. We refer to Section~\ref{sec: comp entanglement} for more details.

Our first results cover basic properties that establish the behavior of computational entanglement under standard state manipulations. First, we show that, in both the one-shot and uniform settings, $\hat E^{\epsilon}_D$ is \emph{lower-bound convex}, and $\hat E^{\epsilon}_C$ is \emph{upper-bound concave}. Second, we establish the \emph{additive behavior} of the computational distillable entanglement and cost with respect to the tensor product.
\begin{proposition}[Informal of Proposition \ref{prop: convexity one shot dist} and \ref{prop: concavity one shot cost}] Let $\{\rho_x^{\lambda}\}_{x\in X}$ be quantum states shared between Alice and Bob and let $\sigma^\lambda$ be a convex mixture of states in $\{\rho_x^{\lambda}\}_{x\in X}$. Then the following relations hold:
        \begin{align*}
        \hat{E}^{\epsilon}_D(\{\sigma^\lambda\}) \geq m(\lambda) &\implies \forall x\in X,\,\hat{E}^{\epsilon^D_x}_D(\{\rho_x^{\lambda}\})\geq m(\lambda),\\
        \hat{E}^{\epsilon}_C(\{\sigma^\lambda\}) \leq n(\lambda) &\implies  \forall x\in X,\,\hat{E}^{\epsilon^C_x}_C(\{\rho_x^{\lambda}\})\leq n(\lambda),
    \end{align*}
where $\sigma^\lambda = \sum_{x\in X}p(x)\rho_x^{\lambda}$ and $\varepsilon^D_x(\lambda),\varepsilon^C_x(\lambda)$ denote the individual errors in distilling from and diluting to each $\rho_x^\lambda$, respectively. The latter statement holds assuming that the error is achieved on the corresponding diluting map. 
    \end{proposition}


\begin{proposition}[Informal of Proposition \ref{prop: supperaditive one shot dist} and \ref{prop: subbativity of one shot cost}] Let $\{\rho_1^{\lambda}\}$ and $\{\rho_2^{\lambda}\}$ be the families of bipartite states shared between Alice and Bob, such that one can distill at least $m_1(\lambda)$ and $m_2(\lambda)$ EPR pairs, and need $n_1(\lambda),n_2(\lambda)$  EPR pairs to recover the given families of bipartite states, respectively. Then the following relations hold: 
    \begin{equation*}
   \hat{E}_D^{\epsilon}(\{\rho_1^{\lambda}\otimes\rho_2^{\lambda}\})\geq m_1(\lambda) + m_2(\lambda) \quad\text{and}\quad  \hat{E}_C^{\epsilon}(\{\rho_1^{\lambda}\otimes\rho_2^{\lambda}\})\leq n_1(\lambda) + n_2(\lambda),
    \end{equation*}
    where $\varepsilon(\lambda)\leq \varepsilon_1(\lambda) + \varepsilon_2(\lambda)$, and $\varepsilon_1(\lambda),\varepsilon_2(\lambda)$ are errors of distillation (dilution) of $\{\rho_1^{\lambda}\}$ and $\{\rho_2^{\lambda}\}$, respectively.
    \end{proposition}
    We note that these statements, as well as the following ones, are trivially extended to the uniform case, in which states are indexed by a set of keys $K$, with the exception of lower-bound convexity, for which the per-key extension can be shown. We refer to Subsection \ref{sec: convex concave} and Subsection \ref{sec: subb and super distillation and cost} for more details.

Our next set of results establishes a quantitative unconditional separation between computational and information-theoretic one-shot distillable entanglement in the deterministic setting. A related maximal separation was previously sketched in~\cite[Lemma~4.1]{arnon2023computational} using a counting argument. Here, we develop a different quantitative approach based on projective-unitary nets and vector packings, which relates the magnitude of the separation to the geometrical structure of the underlying family of states. That is, we compute an upper bound on computational one-shot distillable entanglement of a special family of states $S^\lambda$, which is generated by a geometrical structure known as an $\eta$-net on two different groups -- the projective unitary and the projective special orthogonal group. Our approach shows that, for vanishing distillation error and specific choices of $\eta$ parameter, the fraction of the information-theoretically available entanglement that can be efficiently distilled tends to zero as $\lambda\to+\infty$. It also establishes an explicit connection between the geometrical structure of the underlying family of states and its efficiently accessible entanglement. Here, we highlight the result for the projective special orthogonal group (Theorem~\ref{lemma: bound on E_D so max}), which gives the strongest quantitative separation obtained in this work. We refer to Section~\ref{sec: complexity vs entanglement} for more details on the construction of the family of states $S^{\lambda}$ and the proof of the aforementioned results. 
\begin{theorem}[Informal of Theorem \ref{lemma: bound on E_D so max}]\label{lemma: inf 2}
   Let $m:\mathbb N\to\mathbb N$ be a non-decreasing polynomially bounded function with $m\in\omega(\log_2\lambda)$, and let $k:\mathbb N\to(1,\infty)$ satisfy $m(\lambda)(1-1/k(\lambda))\in\omega(\log_2\lambda)$. Define $\eta(\lambda):=1-2^{-m(\lambda)/k(\lambda)}$. Let $S^\lambda_{\mathrm{PSO}}$ be the family of bipartite pure states generated by applying a maximal $\sqrt{2\eta(\lambda)}$-separated set on the projective special orthogonal group to $m(\lambda)$ EPR pairs, and let $d(\lambda)$ be a valid lower bound on its computational one-shot distillable entanglement with error $\varepsilon(\lambda)$. Then, for all sufficiently large $\lambda$, the following bound holds:
    \begin{equation*}
        d(\lambda) \leq \frac{1}{1-\sqrt{\varepsilon(\lambda)}}\left(\left[1-\frac{1}{k(\lambda)}\right]m(\lambda) + g_2\left(\sqrt{\varepsilon(\lambda)}\right) + \nu(\lambda)\right),
    \end{equation*}
        where  $g_2(\varepsilon)\in[0,2]$ is defined in Proposition~\ref{prop: distillation and squashed}, $\varepsilon:\mathbb N\rightarrow(0,1)$ is the error of distillation, and $\nu(\lambda)$ is a negligible function. 
\end{theorem}

 Another desirable property that we investigate for computational entanglement measures is the invariance under local unitary channels. By combining any of the separation results of Section~\ref{sec: complexity vs entanglement} with the cofinal diagonalization argument of Proposition~\ref{prop: seq of S_lambda}, we show that computational one-shot distillable entanglement need not be invariant under arbitrary families of local unitaries when no efficiency restriction is imposed. 
As we describe in Section \ref{sec: invariance local unitaries}, both computational entanglement cost and distillable entanglement are invariant under \emph{efficiently implementable families of local unitaries}, whereas this \emph{invariance can fail} for families that are not efficiently implementable.
\begin{theorem}[Informal of Proposition \ref{th:comp one shot cost inva uni} and  \ref{th: local dist unita invar epsilon diff zero}]
    Let $\{\rho^{\lambda}_{AB}\}$ be the family of bipartite states shared between Alice and Bob, and let  $\{U^{\lambda}_{AB}\}$ be an efficient family of  local unitaries. Then the following holds: 
    \begin{align*}
                 \hat{E}^{\epsilon}_D(\{\rho^{\lambda}_{AB}\})\geq m(\lambda)&\iff\hat{E}^{\epsilon}_D(\{U^{\lambda}_{AB}\rho^{\lambda}_{AB}U^{\lambda\dagger}_{AB}\})\geq m(\lambda),
\\
                \hat E^{\epsilon}_C(\{\rho^{\lambda}_{AB}\})\leq n(\lambda) &\iff\hat E^{\epsilon}_C(\{U^{\lambda}_{AB}\,\rho^{\lambda}_{AB}\,U^{\lambda\dagger}_{AB}\})\leq n(\lambda).
    \end{align*}
However, the measures are not always invariant under arbitrary, not necessarily efficiently implementable, families of local unitaries.
\end{theorem}

As a corollary, we observe that computational distillable entanglement and cost are not LOCC monotone in general. However, LOCC monotonicity is recovered under the action of efficient families of LOCC channels. We utilize this property in the following when we establish the ``Second Law of Efficient Entanglement Manipulation'' (see Section \ref{sec: val-cont}).
\begin{theorem}[Informal of Proposition \ref{thm: locc not monotone comp dist ent} and \ref{thm: locc not monotone comp ent cost}]
    Let the family of bipartite states $\{\rho^{\lambda}_{AB}\}$ be shared between Alice and Bob, and let $\{\Lambda^{\lambda}\}$ be an efficient family of LOCC channels. Then the following holds: 
    \begin{align*}
        \hat E_D^\varepsilon(\{\Lambda^\lambda(\rho^\lambda_{AB})\}) \geq m(\lambda) &\implies \hat E_D^\varepsilon(\{\rho^\lambda_{AB}\}) \geq m(\lambda)\\
        \hat E_C^\varepsilon(\{\rho^\lambda_{AB}\}) \leq n(\lambda)&\implies \hat E_C^\varepsilon(\{\Lambda^\lambda(\rho^\lambda_{AB})\}) \leq n(\lambda).
    \end{align*}
    However, the computational entanglement measures are not always LOCC monotone for arbitrary families of LOCC channels.
\end{theorem}
 
The next topic that we explore in this work is the continuity property of computational one-shot distillable entanglement. Operationally, this property describes how the amount of computational entanglement changes when the original family of states is perturbed (see discussion after Theorem \ref{thm: val-cont E_d}). To answer this, we establish a quantitative relation that connects the trace distance between states and lower bounds on their computational distillable entanglement.
\begin{theorem}[Informal of Theorem  \ref{thm: val-cont E_d}]\label{thm: inf cont}
    Let the families of bipartite states $\{\rho^\lambda_{AB}\}$ and $\{\sigma^\lambda_{AB}\}$ be shared between Alice and Bob, such that $\frac{1}{2}\|\rho_{AB}^\lambda-\sigma^\lambda_{AB}\|_1\leq \delta(\lambda)$ for all $\lambda$. Assume that $\hat E^\varepsilon_D(\{\rho^\lambda_{AB}\})\geq m(\lambda)$. Then, if $\delta(\lambda)+\varepsilon(\lambda)<1-2^{-2m(\lambda)}$ for all $\lambda$, there exists an efficient family of distillation protocols, which fail with probability at most $P_{\emph{fail}}(\lambda):=\frac{2^{2m(\lambda)}}{2^{2m(\lambda)}-1}(\varepsilon(\lambda)+\delta(\lambda))(1-2^{-s_{\delta,\varepsilon,m}(\lambda)})$, such that the following holds:
 \begin{equation*}
        \hat E_D^\varepsilon(\{\sigma_{AB}^\lambda\}) \geq m(\lambda) - \min\{m(\lambda),s_{\delta,\varepsilon,m}(\lambda)\}
        ,
    \end{equation*}
    where the function $s_{\delta,\varepsilon,m}(\lambda)$ is given for all values of $\lambda$ as
    \begin{equation*}
        s_{\delta,\varepsilon,m}(\lambda) :=\left\lceil \log_2\left(1 + \frac{\delta(\lambda)+2^{-2m(\lambda)}\varepsilon(\lambda)}{\varepsilon(\lambda)[1-\varepsilon(\lambda)-\delta(\lambda)-2^{-2m(\lambda)}]}\right)\right\rceil.
    \end{equation*}
\end{theorem}
To establish this relation, we construct a family of distillation protocols for an altered family of states that achieve the given error of distillation by sacrificing a number of distilled EPR pairs. In the discussion following  Theorem  \ref{thm: val-cont E_d}, we show that, in the regime of small errors ($\delta$), the performance of the proposed family of distillation protocols is on par with the original family of distillation channels. 

Finally, we show an application of the established properties by deriving a fundamental relation between the values of computational one-shot distillable entanglement and cost, which we refer to as the ``Second Law of Efficient Entanglement Manipulation''. This result shows that in the regime of a small relative error ($\sqrt{\varepsilon_1}/\varepsilon_2$), values of computational one-shot measures can be close to each other. This result extends the result of \cite{W21} to the case in which the computational and information-theoretic bounds do not coincide (see discussion following Theorem \ref{th: second law}). 

\begin{theorem}[Informal of Theorem \ref{th: second law}]
    Let the family of bipartite states $\{\rho^\lambda_{AB}\}$ be shared between Alice and Bob. Let $\hat E^{\varepsilon_2}_D(\{\rho^\lambda_{AB}\})\geq d(\lambda)$ and $\hat E_C^{\varepsilon_1}(\{\rho^\lambda_{AB}\})\leq c(\lambda)$. Then, if $\sqrt{\varepsilon_1(\lambda)}+\varepsilon_2(\lambda)<1-2^{-2d(\lambda)}$ for all $\lambda$, the following holds for all $\lambda$:
    \begin{equation*}
       c(\lambda)  \geq d(\lambda)- \min\{s_{\sqrt{\varepsilon_1},\varepsilon_2,d}(\lambda),d(\lambda)\}-\log_2\left(\frac{1}{1-\varepsilon_2(\lambda)}\right) - \log_2\left(\frac{1}{1-P_{\emph{fail}}(\lambda)}\right),
    \end{equation*}
    where functions $P_{\emph{fail}}:\mathbb N\rightarrow [0,1)$ and $s_{\sqrt{\varepsilon_1},\varepsilon_2,d}:\mathbb N\rightarrow\mathbb N$ are given by Theorem~\ref{thm: inf cont}.
\end{theorem}

\section{Preliminaries}\label{sec: intro information theory}
In this section, we recall relevant definitions and results. Subsection \ref{subsec:States, LOCC channels and fidelity} concerns notation and quantum information, while Subsection \ref{subsec: distillable entanglement and entanglement cost} reviews definitions of entanglement theory. We refer to the textbooks \cite{nielsen2001quantum,watrous2018theory,aubrun2017alice,khatri2024principles} for more details.

\subsection{Quantum Information Theory}\label{subsec:States, LOCC channels and fidelity} 
In this subsection, we fix the notation and recall the elements of quantum information theory needed. 

First, we recall the Landau notation; we use the Big-Omega notation $f(n) \in \Omega(g(n))$ if $f(n)$ is bounded below by $g(n)$ asymptotically. More precisely, this means that there exist  $c>0,n_0\in\mathbb N$ such that $f(n)\geq c\cdot g(n)$ for all $n\geq n_0$. Similarly, we denote $f(n) \in\omega(g(n))$ if $f(n)$ dominates $g(n)$ asymptotically: for any $k>0$, there exists $n_0\in\mathbb N$ such that $f(n)>k \cdot g(n)$ for any $n\geq n_0$. We denote by $\operatorname{Poly}(\lambda)$ the set of real-valued polynomials in the variable $\lambda$. We also use the convention $[N]:=\{1,\dots,N\}$. Additionally, we use the notion of a negligible function, which is defined as follows:
\begin{definition}
    \label{def: negligible function}
A function $f:\mathbb N\rightarrow\mathbb R$ is said to be \emph{negligible} if for any positive polynomial $p\in\mathrm{Poly}(\lambda)$ the following holds
\begin{equation*}
 \exists\lambda_0\in\mathbb N \quad\text{such that}\quad \forall\lambda\geq\lambda_0 \quad |f(\lambda)| < \frac{1}{p(\lambda)}.
\end{equation*}
We denote the class of such functions by $\operatorname{negl}(\lambda)$.
\end{definition}

In the following, we assume that the dimension of Hilbert spaces is finite: $\mathcal H\cong\C^n$. Denote by $\mathcal{M}_n(\C)$ the space of $n\times n $ matrices. One can introduce metrics on this space using (Schatten) $p$-norms $\|\cdot\|_p$  defined for $1\leq p<\infty$: 
\begin{equation*}
\forall A\in\mathcal{M}_n(\C)\to\|A\|_p:=\left[\Tr\left(|A|^p\right)\right]^{\frac{1}{p}}\quad\text{where}\quad|A|:=\sqrt{A^{\dagger}A}.
\end{equation*}
We denote by $\mathcal L(\mathcal H)\cong\mathcal{M}_n(\C)$ the set of linear operators on $\mathcal H$. In quantum information theory, a quantum state is described by a density matrix $\rho\in\mathcal D(\mathcal H)$ where 
\begin{equation*}
\mathcal{D}(\mathcal{H}):=\{\rho\in\mathcal{M}_n(\C): \,\rho\geq 0,\,\Tr\rho=1\}.
\end{equation*}
We denote by $\mathbb U(n)$ and $\operatorname{SO}(n)$ the group of unitary and special orthogonal matrices in $\C^n$, respectively.

A quantum state $\rho$ is pure if $\rho=\ketbra{\psi}{\psi}$, where $\ket{\psi}\in\mathcal H$. We use the notation $\psi$ to denote pure states of the form $\ketbra{\psi}{\psi}$. Moreover, we  denote by $\mathcal D(\mathcal H_A\otimes\mathcal H_B)$ the set of \emph{bipartite quantum states}. We  use the shorthand notation $\mathcal H_{AB}\equiv\mathcal H_A\otimes\mathcal H_B$. A state $\rho_{AB}\in\mathcal D(\mathcal{H}_{AB})$ is \emph{separable} if it can be expressed as a convex combination of product states $\{\sigma^x_A\otimes\tilde{\sigma}^x_B\}_{x\in X}$:
\begin{equation*}
    \rho_{AB}=\sum_{x\in X}p(x)\sigma^x_A\otimes\tilde{\sigma}^x_B.
\end{equation*}
If a state is not separable, we say it is \emph{entangled}.
Recall that a \emph{partial trace} is a map from a bipartite system to one of the subsystems. We  denote by $\Tr_B$ the partial trace on the subsystem $B$.

Another relevant notion that we make use of in this work is the \emph{LOCC channel}. LOCC channels represent transformations of a state in the scenario where two physically separated parties, Alice and Bob, each perform some \emph{local operations} and exchange  \emph{classical information} obtained from local measurements of their part of the resulting state.
\begin{definition}
    An \emph{LOCC channel} is a map that is implemented by a finite sequence of applications of the corresponding one-way LOCC channels. In turn, a \emph{one-way LOCC channel from Alice to Bob} $\Gamma^{A\rightarrow B}:\mathcal D(\mathcal H_A\otimes\mathcal H_B)\rightarrow\mathcal{D}(\mathcal H_{\bar{A}}\otimes\mathcal H_{\bar{B}})$ is a map of the form 
\begin{equation*}
    \Gamma^{A\rightarrow B}:=\sum_{x\in X}\mathcal{E}_{A}^x\otimes\mathcal{N}_{B}^x, 
\end{equation*}
from the input spaces $\mathcal H_A,\mathcal H_B$ to the output spaces $\mathcal H_{\bar{A}},\mathcal H_{\bar B}$, for a finite $|X|$, where the maps $\{\mathcal{E}^x_A\}_x$, $\{\mathcal N_B^x\}_x$ are completely positive, with each $\mathcal N_B^x$ and $\sum_{x\in X}\mathcal E_A^x$ being trace preserving. The corresponding one-way LOCC channel from Bob to Alice is defined similarly as
\begin{equation*}
    \Gamma^{B\rightarrow A}:=\sum_{x\in X}\mathcal{N}_{A}^x\otimes\mathcal{E}_{B}^x.
\end{equation*}
Then, an LOCC channel $\Gamma: \mathcal D(\mathcal H_A\otimes\mathcal H_B)\rightarrow\mathcal{D}(\mathcal H_{\bar{A}}\otimes\mathcal H_{\bar{B}})$ can be written in the following form  
\begin{equation*}
    \Gamma_{AB\rightarrow A^\prime B^\prime}:=\sum_{y\in Y}\mathcal{T}_{A}^y\otimes\mathcal{R}_{B}^y,
\end{equation*}
 for a finite $|Y|$, where the maps $\left\{\mathcal{T}_A^y\right\}_y,\left\{\mathcal{R}_B^y\right\}_y$ are completely positive and $\Gamma_{AB\rightarrow A^\prime B^\prime}$ is trace preserving.
\end{definition}
From now on, to simplify the notation, we  denote by $\Gamma$ any LOCC channel instead of $\Gamma_{AB\to A'B'}$. 

Another important quantity is the \emph{fidelity} \cite{khatri2024principles}. Intuitively, the fidelity measures the similarity between two given quantum states. 
\begin{definition}[Fidelity]
    Let the quantum states $\rho$ and $\sigma$ be in $\mathcal{D}(\mathcal H)$. Then the fidelity is defined as: 
    \begin{equation*}
    F(\rho,\sigma):= \left\|\sqrt{\rho}\sqrt{\sigma}\right\|_1^2.
    \end{equation*}
\end{definition}
It is easy to check that if one of the quantum states is pure, the fidelity reduces to: 
\begin{equation*}\label{remark: fidelity of expectation}
            F(\rho,\ketbra{\psi}{\psi})=\bra{\psi}\rho\ket{\psi}. 
\end{equation*}
We recall some of the properties of fidelity.  
\begin{proposition}[Concavity of Fidelity]\label{prop: concav fidelity}
Let $\{\rho_x\}_{x\in X}$ be a set of states with $X$ being a finite register, and let $\sigma$ be a quantum state in $\mathcal D(\mathcal H)$. Then the fidelity is a concave function in either of its arguments: 
\begin{equation*}
F\Big(\sum_{x\in X}\,p(x)\,\rho_x,\sigma\Big)\geq \sum_{x\in X}\,p(x)\,F(\rho_x,\sigma).
\end{equation*}
\end{proposition}
\begin{proposition}[Multiplicativity of Fidelity]\label{prop: fidelity factorisation}
    Let $\{\rho_i\}_{i\in\{1,2\}}$ and $\{\sigma_i\}_{i\in\{1,2\}}$ be quantum states in $\mathcal D(\mathcal H_i)$ with $i\in\{1,2\}$. Then the fidelity factorizes as: 
    \begin{equation*}
F(\rho_1\otimes\rho_2,\sigma_1\otimes\sigma_2)=F(\rho_1,\sigma_1)F(\rho_2,\sigma_2).
    \end{equation*}
\end{proposition}
\begin{proposition}[Unitary Invariance of Fidelity]\label{prop: unitary inv fidelity}
    Let $\rho$ and $\sigma$ be quantum states in $\mathcal D(\mathcal{H})$. Let $U$ be a unitary operator on $\mathcal H$. Then the fidelity is unitarily invariant
    \begin{equation*}
        F(U\rho U^\dagger, U\sigma U^\dagger)=F(\rho,\sigma).
    \end{equation*}
\end{proposition}
Another property of the fidelity is the relation between the fidelity and the Schatten $1$-norm $\|\cdot\|_1$, which is also known as the Fuchs \& van de Graaf inequalities \cite{FvdG99}:
\begin{proposition}[Fuchs \& van de Graaf inequalities~\cite{FvdG99}]\label{prop: fidelity and one norm}
Let $\rho$ and $\sigma$ be quantum states in $\mathcal D(\mathcal{H})$. Then the following inequalities hold
\begin{equation*}
            1-\sqrt{F(\rho,\sigma)}\leq\frac{1}{2}\|\rho-\sigma\|_1\leq
        \sqrt{1-F(\rho,\sigma)}.
\end{equation*}
\end{proposition}
Concluding the list of properties of fidelity, we recall the data-processing inequality and another relation between fidelity and trace distance:
\begin{proposition}[Data Processing Inequality for Fidelity]\label{prop: data-processing fidelity}
    Let $\rho$ and $\sigma$ be in $\mathcal{D}(\mathcal H_A)$, and let $\tilde\Phi: \mathcal{D}(\mathcal H_A) \rightarrow \mathcal D(\mathcal H_B)$ be a quantum channel. Then the following holds 
    \begin{equation*}
        F(\rho, \sigma) \leq F(\tilde \Phi(\rho), \tilde \Phi (\sigma)).
    \end{equation*}
\end{proposition}
\begin{proposition}[Expected Fidelity and Trace Distance]\label{prop: close states fid}
    Let $\rho,\sigma\in\mathcal{D}(\mathcal H)$ and $\psi:=\ketbra{\psi}{\psi}\in\mathcal{D}(\mathcal H)$. Then the following relation holds
    \begin{equation*}
        F(\rho,\psi)\geq F(\sigma,\psi) - \frac{1}{2}\|\rho-\sigma\|_1.
    \end{equation*}
\end{proposition}

Finally, we recall an important property of the Schatten 2-norm -- invariance under the action of unitaries.
\begin{proposition}[Invariance of 2-norm~\cite{H12}]\label{prop: inv 2 norm}
    Let $A\in\mathcal M_n(\mathbb C)$ and $U\in\mathbb U(n)$. Then the following relation holds
    \begin{equation*}
        \|A\|_2 = \|UA\|_2 = \|AU\|_2.
    \end{equation*}
\end{proposition}

\subsection{Entanglement Measures}\label{subsec: distillable entanglement and entanglement cost}
In this subsection, we recall different entanglement measures that are relevant for the rest of this work.

The first two entanglement measures we recall are the \emph{distillable entanglement} and \emph{entanglement cost}, denoted respectively by $E_D$ and $E_C$. The distillable entanglement and entanglement cost represent, respectively, the asymptotic rate at which one can extract EPR pairs from an arbitrary bipartite state and the asymptotic rate at which one can fully reconstruct a given bipartite state from EPR pairs. We recall that an EPR pair state is given by $|\varphi\rangle = \frac{1}{\sqrt{2}}(|00\rangle + |11\rangle)$. We also denote the corresponding density matrix of an EPR pair as $\Phi = \ketbra{\varphi}{\varphi}$.

We start with the distillable entanglement~\cite{bennett1996mixed, bennett1996purification,BD10}. We first introduce the \emph{one-shot distillable entanglement} denoted by $E^{\epsilon}_D$ and then the \emph{asymptotic distillable entanglement} denoted by $E_D$.  
\begin{definition}[One-Shot Distillable Entanglement~\cite{BD10}]\label{def: inf-th dist ent}
    Let $\mathcal{H}_A,\mathcal H_B$ be the input spaces of Alice and Bob, respectively, and let $\mathcal{H}_{\bar A}=\mathcal H_{\bar B}=(\C^2)^{\otimes m}$. Let $\Gamma:\mathcal D(\mathcal H_A\otimes \mathcal H_B)\to \mathcal{D} (\mathcal H_{\bar A}\otimes \mathcal H_{\bar B})$ be an LOCC channel. The \emph{one-shot distillable} entanglement of a bipartite state $\rho_{AB}\in \mathcal D(\mathcal H_A\otimes \mathcal H_B)$ is given by: 
    \begin{equation*}
    \forall \epsilon\in[0,1),\quad E^{\epsilon}_D(\rho_{AB}):=\sup_{m,\Gamma}\{\,m\,|\,p_{\text{err}}(\Gamma,\rho_{AB})\leq \epsilon\},\end{equation*}
    where
    \begin{equation*}
    p_{\emph{err}}(\Gamma,\rho_{AB})=1-\bra{\phi^{\otimes m}}\Gamma(\rho_{AB})\ket{\phi^{\otimes m}}.
    \end{equation*}
\end{definition}
From the definition above, the one-shot distillable entanglement $E^{\epsilon}_D$ represents the number of EPR pairs $\ket{\phi}$ that can be extracted from a given state $\rho_{AB}$ with a given precision $\epsilon$. 
\begin{definition}[Distillable Entanglement~\cite{W21}]
    Let $\rho_{AB}$ be a bipartite state in $\mathcal D(\mathcal{H}_A\otimes \mathcal H_B)$. Then the \emph{asymptotic distillable entanglement} is given by: 
    \begin{equation*}
    E_D(\rho_{AB}):=\inf_{\epsilon\in(0,1)}\liminf_{t\to\infty}\frac{1}{t}E^{\epsilon}_D(\rho_{AB}^{\otimes t}).
    \end{equation*}
\end{definition}
The distillable entanglement is LOCC monotone~\cite{bennett1996mixed} and superadditive under tensor product~\cite{Bandyopadhyay2005}. At the same time, it is known that distillable entanglement is not convex~\cite{Shor2001}, if one assumes the validity of the necessary conjectures.

Similar to the distillable entanglement, we first introduce the \emph{one-shot entanglement cost} and then the \emph{asymptotic entanglement cost} denoted respectively by $E_C^{\epsilon}$ and $E_{C}$. 
\begin{definition}[One-Shot Entanglement Cost~\cite{BD11}]\label{def: inf-th cost ent}
    Let $\mathcal{H}_A,\mathcal H_B$ be the output spaces of Alice and Bob, respectively, and let $\mathcal{H}_{\bar A}=\mathcal H_{\bar B}=(\C^2)^{\otimes n}$. Let  $\Gamma:  \mathcal D( \mathcal H_{\bar A}\otimes \mathcal H_{\bar B})\to\mathcal D(\mathcal H_A\otimes \mathcal H_B)$ be an LOCC channel. The \emph{one-shot entanglement cost} of a bipartite state $\rho_{AB}\in\mathcal D(\mathcal H_A\otimes \mathcal H_B)$ is given by: 
\begin{equation*}
\forall \epsilon \in[0,1),\quad E^{\epsilon}_C(\rho_{AB}):=\inf_{n,\Gamma}\{\,n\,|\,p_{\text{err}}(\Gamma,\rho_{AB})\leq \epsilon\},
\end{equation*}
where 
\begin{equation*}
p_{\emph{err}}(\Gamma, \rho_{AB}):=1-F(\Gamma(\Phi^{\otimes n}),\rho_{AB}).
\end{equation*}
\end{definition}
The one-shot entanglement cost $E^{\epsilon}_C$ represents the ``minimal'' number of EPR pairs needed to recover a bipartite state with a given precision $\epsilon$.
\begin{definition}[Entanglement Cost~\cite{W21}]
    Let $\rho_{AB}$ be a bipartite quantum state. Then the asymptotic entanglement cost is given by: 
    \begin{equation*}
E_C(\rho_{AB}):=\sup_{\epsilon\in(0,1)}\limsup_{t\to\infty}\frac{1}{t}E^{\epsilon}_C(\rho_{AB}^{\otimes t}).
    \end{equation*}
\end{definition}
The entanglement cost is LOCC monotone due to the connection to the regularized entanglement of formation~\cite{hayden2001asymptotic} and the LOCC monotonicity of the latter~\cite{bennett1996mixed}. Entanglement cost is known to be convex~\cite{10.1063/1.1495917}.

Another major entanglement measure is \emph{squashed entanglement}. We recall its definition and relations to distillable entanglement and entanglement cost~\cite{christandl2004squashed}. For this, we first recall the notion of conditional mutual information.  
\begin{definition}[Conditional Quantum Mutual Information]\label{def:cond mutual information}
    Let $\rho_{ABE}\in\mathcal{D}(\mathcal{H}_{ABE})$ be a tripartite state. The  \emph{conditional mutual information} of the state $\rho_{ABE}$ is defined as follows:
    \begin{equation*}
        I(A;B|E)_{\rho}:= \operatorname{H}_{AE}(\rho) + \operatorname{H}_{BE}(\rho) - \operatorname{H}_{ABE}(\rho) - \operatorname{H}_E(\rho) ,
        \end{equation*}
        where $\operatorname{H}(\rho)$ is the von Neumann entropy.
        \end{definition}
\begin{definition}[Squashed Entanglement~\cite{christandl2004squashed}]\label{def:squashed entanglement}
    Let $\rho_{AB}\in\mathcal{D}(\mathcal{H}_{AB})$ be a bipartite state. The \emph{squashed entanglement} of the state $\rho_{AB}$ is defined as follows:
    \begin{equation*}
        E_{sq}(\rho_{AB}):=\inf\left\{\frac{1}{2}I(A;B|E)_{\rho_{ABE}}\,|\,\rho_{ABE}\right\},
    \end{equation*}
    where $\rho_{ABE}$ is an extension satisfying $\operatorname{Tr}_E(\rho_{ABE}) = \rho_{AB}$.
\end{definition}
The squashed entanglement measure is LOCC monotone, convex and additive under tensor products~\cite{christandl2004squashed}. 

In the following, we recall the relation of the squashed entanglement with the entanglement cost and distillable entanglement. First of all, it provides the following upper bound on one-shot distillable entanglement.
\begin{proposition}\cite[Theorem 13.9]{khatri2024principles}\label{prop: distillation and squashed}
    Let $\rho_{AB}$ be a bipartite state in $\mathcal D(\mathcal H_{AB})$ and let $\varepsilon\in[0,1)$. Then the following holds: 
\begin{equation*}
    E^{\varepsilon}_D(\rho_{AB})\leq\frac{1}{1-\sqrt{\varepsilon}}\left(E_{sq}(\rho_{AB}) + g_2\left(\sqrt{\varepsilon}\right)\right),
\end{equation*}
     where $g_2(\delta):=(\delta+1)\operatorname{log}_2(\delta+1) - \delta\operatorname{log}_2\delta$.
\end{proposition}

Moreover, the squashed entanglement is upper bounded by the entanglement cost and lower bounded by the distillable entanglement \cite{christandl2004squashed}.

\begin{proposition}
    Let $\rho_{AB}\in\mathcal{D}(\mathcal{H}_{AB})$. Then the following relations hold:
    \begin{equation*}
        E_D(\rho_{AB})\leq E_{\text{sq}}(\rho_{AB})\leq E_C(\rho_{AB}).
    \end{equation*}
\end{proposition}
\subsection{Unitary Designs and Twirling Channels}\label{subsec: twirling}

In this subsection, we recall the main concepts from the theory of unitary designs. We pay special attention to 2-designs and their connection to the Clifford group. 

We begin by recalling the definition of unitary $t$-designs. This concept represents any distribution over a unitary group that can replicate up to the $t$-th moment of the Haar random distribution \cite{DCEL09}. 
\begin{definition}[Unitary $t$-design~\cite{M24}]\label{def: t-design}
Let $t,d\in\mathbb N$. Then a probability distribution $\nu$ over a set of unitary operators $S\subseteq\mathbb U(d)$ is called a \emph{unitary $t$-design} if, for every linear operator $O\in\mathcal L(({\mathbb C^d})^{\otimes t})$, the following relation holds
\begin{equation*}
    \underset{V\sim \nu}{\mathbb E}\left[V^{\otimes t} OV^{\dagger \otimes t}\right] =  \underset{U\sim \mu_H}{\mathbb E}\left[U^{\otimes t} OU^{\dagger \otimes t}\right],
\end{equation*}
where $\mu_H$ denotes the Haar measure on the unitary group $\mathbb U(d)$ and $\underset{V\sim \nu}{\mathbb E}$ refers to averaging over $\mathbb U(d)$ according to the distribution $\nu$ in the following way
\begin{equation*}
    \underset{V\sim\nu}{\mathbb E}\left[V^{\otimes t} OV^{\dagger \otimes t}\right] := \int_{\mathbb U(d)} \left[V^{\otimes t} OV^{\dagger \otimes t}\right]d\nu (V).
\end{equation*}
\end{definition}
We now narrow our focus to the case of $t=2$. In what follows, we recall some of the properties of 2-designs most relevant to this work. First, we state the result establishing the Clifford group as a 2-design \cite{DCEL09}.
\begin{proposition}[Theorem 1~\cite{DCEL09}]\label{prop: cliff is 2}
    The uniform distribution over the Clifford group on $n$ qubits forms a unitary 2-design.
\end{proposition}
This result becomes especially relevant for this work when combined with the following observations about the Clifford group.
\begin{proposition}[Section IV, Corollary 1~\cite{AG04}]\label{prop: eff cliff 1}
    Any $n$-qubit Clifford circuit has an equivalent circuit description with $O(n^2/\operatorname{log}(n))$ gates from the set $\{\operatorname{H},\operatorname{CNOT},\operatorname{S}\}$.
\end{proposition}
\begin{proposition}[\cite{KS14}]\label{prop: eff cliff 2}
    There exists an algorithm that samples an element from the $n$-qubit Clifford group using $O(n^3)$ operations.
\end{proposition}
When combined, Proposition \ref{prop: cliff is 2}, Proposition \ref{prop: eff cliff 1}, and Proposition \ref{prop: eff cliff 2} imply that one can perform 2-designs \emph{efficiently}, meaning that the complexity of a quantum circuit implementing a 2-design grows at most polynomially with the number of qubits. 

We now focus our attention further on the computational entanglement setting defined in Definition \ref{def: comp distil vidick} and Definition \ref{def: comp entang cost vidick}. In this work, we mainly deal with families of states $\{\rho^\lambda_{AB}\}_{\lambda\in\mathbb N}$, where, for each value of $\lambda\in\mathbb N$, $\rho^\lambda_{AB}\in\mathcal{D}((\mathbb C^2)^{\otimes n(\lambda)}\otimes(\mathbb C^2)^{\otimes n(\lambda)})$. Therefore, for each element of the family, Alice's and Bob's registers consist of $n(\lambda)$ qubits, where $n:\mathbb N\rightarrow\mathbb N$ is some polynomially bounded function. 

With this, we can now define a family of twirling channels of a specific form that will be of interest to us in the following.
\begin{definition}[Family of Twirling Channels]\label{def: fam of twirling}
    Let $\lambda\in\mathbb N$, $n:\mathbb N\rightarrow\mathbb N$ be a polynomially bounded function. We say that $\{\operatorname{T}^\lambda\}_{\lambda\in\mathbb N}$ is a \emph{family of bilateral twirling channels} acting on states in $(\mathbb C^2)^{\otimes n(\lambda)}\otimes(\mathbb C^2)^{\otimes n(\lambda)}$, if the following holds for all states of the form $\rho^\lambda_{AB}\in\mathcal{D}((\mathbb C^2)^{\otimes n(\lambda)}\otimes(\mathbb C^2)^{\otimes n(\lambda)})$
    \begin{equation*}
        \forall \lambda\in\mathbb N,\quad \operatorname{T}^\lambda(\rho^\lambda_{AB}):=\underset{U\sim\nu^\lambda}{\mathbb E}\left[(U_A\otimes U_B^*)\rho^\lambda_{AB}(U_A^\dagger\otimes{U_B^{*\dagger}})\right],
    \end{equation*}
    where $\nu^\lambda$ is a distribution over $\mathbb U(2^{n(\lambda)})$ that forms a unitary 2-design.
\end{definition}
We note that, for each value of $\lambda$, $\operatorname{T}^\lambda$ is an LOCC channel -- it requires only a shared source of classical randomness and the application of appropriate unitaries on Alice's and Bob's sides. Moreover, we note that $\{\operatorname{T}^\lambda\}_{\lambda\in\mathbb N}$ is an efficient family of  LOCC channels. We make this precise in the following lemma.
\begin{lemma}\label{lemma: twirling is efficient}
    Let $\lambda\in\mathbb N$ and $n:\mathbb N\rightarrow\mathbb N$ be a polynomially bounded function. Then the family of twirling channels $\{\operatorname{T}^\lambda\}_{\lambda\in\mathbb N}$ acting on $n(\lambda)$-qubit bipartite states, as described in Definition \ref{def: fam of twirling}, is an efficient family of LOCC channels.
\end{lemma}
\begin{proof}
   First, $\{\operatorname{T}^\lambda\}_{\lambda\in\mathbb N}$ is a family of LOCC channels, since implementing each $\operatorname{T}^\lambda$  requires a shared source of randomness between parties and the application of the corresponding unitaries on subsystems of Alice and Bob. 

   Second, $\{\operatorname{T}^\lambda\}_{\lambda\in\mathbb N}$ is an efficient family of LOCC channels. To see this, we recall that the uniform distribution over the Clifford group forms a 2-design (see Proposition~\ref{prop: cliff is 2}). Then, we use the results of Proposition~\ref{prop: eff cliff 1} to obtain that the circuit description of each $\operatorname{T}^\lambda$ has $O(n^2(\lambda)/\operatorname{log}(n(\lambda)))$ gates. Therefore, since, by assumption of the lemma,  $n:\mathbb N\rightarrow\mathbb N$ is a polynomially bounded function, the complexity of the circuit descriptions of $\{\operatorname{T}^\lambda\}_{\lambda\in\mathbb N}$ grows at most polynomially. 
   
   Moreover, while it is not accounted for in the definition of the size of an LOCC channel (see Definition~\ref{def: circuit size}), it seems important to mention that sampling an $n$-qubit Clifford gate according to the uniform distribution is also an efficient operation requiring at most $O(n^3)$ operations (see Proposition~\ref{prop: eff cliff 2}).
\end{proof}

Since $\{\operatorname{T}^\lambda\}_{\lambda\in\mathbb N}$ is an efficient family of LOCC channels, as established in Lemma~\ref{lemma: twirling is efficient}, from now on we denote this family by $\{ {\operatorname{\hat T^\lambda}} \}_{\lambda\in\mathbb N}$, indicating the efficiency of the family of bilateral twirling channels.

Lastly, we are interested in how twirling channels described in Definition~\ref{def: fam of twirling} act on $n$-qubit bipartite states. For this, we make use of the following proposition, established in~\cite{M24}.
\begin{proposition}[Section 9, Example 49~\cite{M24}]\label{prop: aciton of twirl}
    Let $\nu$ be a unitary 2-design. Then for every $\rho_{AB}\in\mathcal{D}(\mathbb C^d\otimes\mathbb C^d)$ the following holds
    \begin{equation*}
        \underset{U\sim\nu}{\mathbb E}\left[(U_A\otimes U_B^*)\rho_{AB}(U_A^\dagger\otimes{U_B^{*\dagger}})\right] =d^2 \frac{\left(1-\operatorname{Tr}(\Phi_d\rho_{AB})\right)}{d^2-1}\frac{\mathbb I_{d^2}}{d^2} + \left(1-d^2 \frac{\left(1-\operatorname{Tr}(\Phi_d\rho_{AB})\right)}{d^2-1}\right)\Phi_d,
    \end{equation*}
    where $\Phi_d := \ketbra{\phi_d}{\phi_d}$ with $|\phi_d\rangle := \frac{1}{\sqrt d}\sum_{i=0}^{d-1}|ii\rangle$ -- maximally entangled state of rank $d\in\mathbb N$ and $\mathbb I_{d^2}$ is the identity operator on $\mathbb C^d\otimes \mathbb C^d$.
\end{proposition}
We adapt this proposition for the setting of multiple qubits in the following corollary, which is one of the main technical tools used in our work.
\begin{proposition}\label{cor: action of eff twirl}
Let $\lambda\in\mathbb N$ and $n:\mathbb N\rightarrow\mathbb N$ be a polynomially bounded function. Let $\{{\operatorname{\hat T^\lambda}}\}_{\lambda\in\mathbb N}$ be a family of twirling channels acting on $n$-qubit bipartite states. Then for any family of states $\{\rho^\lambda_{AB}\}_{\lambda\in\mathbb N}$, with $\rho^\lambda_{AB}\in\mathcal D((\mathbb C^2)^{\otimes n(\lambda)}\otimes(\mathbb C^2)^{\otimes n(\lambda)} )$ for each value of $\lambda$, the following holds
    \begin{equation*}
        \operatorname{\hat T^\lambda}(\rho^\lambda_{AB}) = \Phi^{\otimes n(\lambda)}\operatorname{Tr}\left(\Phi^{\otimes n(\lambda)}\rho^\lambda_{AB}\right) +\frac{\mathbb I_4^{\otimes n(\lambda)} - \Phi^{\otimes n(\lambda)}}{2^{2n(\lambda)}-1}\operatorname{Tr}\left( \left(\mathbb I_4^{\otimes n(\lambda)} - \Phi^{\otimes n(\lambda)}\right) \rho_{AB}^\lambda \right),
    \end{equation*}
    where $\Phi := \ketbra{\phi}{\phi}$ with $|\phi\rangle:= \frac{1}{\sqrt 2 }(|00\rangle + |11\rangle)$ -- the maximally entangled state on 2 qubits.
\end{proposition}
\begin{proof}
    The claim of the proposition simply follows by substituting $d(\lambda):=2^{n(\lambda)}$ into Proposition \ref{prop: aciton of twirl} and identifying, for each $\lambda\in\mathbb N$, that $\mathbb C^{d(\lambda)} \cong (\mathbb C^2)^{\otimes n(\lambda)}$.
\end{proof}
\subsection{Vector Packing on Spheres in $\mathbb R^n$}
In this subsection, we briefly introduce the geometrical problem of packing vectors on the unit sphere in Euclidean space $\mathbb R^n$. 

We begin by introducing a definition of a packing on a sphere in a Euclidean space $\mathbb R^n$. Informally, a packing on a Euclidean sphere of dimension $n-1$, which we denote by $\mathbb S^{n-1}$, is a set of unit-length vectors such that each vector is separated from the others by a specified angle. We recall the formal definition in the following
\begin{definition}[Packing on a sphere~\cite{aubrun2017alice}]\label{def: spherical code}
    Let $n\in\mathbb N$ and $t\geq0$. A packing on a sphere $P(n,t)$ is a collection of unit vectors in $\mathbb R^n$, such that for all distinct pairs $x,y\in P(n,t)$ it holds that $|(x,y)|\leq t$, where $(\cdot,\cdot)$ denotes the standard scalar product in Euclidean space $\mathbb R^n$. For fixed $n\in\mathbb N$ and $t\geq0$, the maximal cardinality of $P(n,t)$ is denoted by $A(n,t)$. 
\end{definition}
Finding an expression for $A(n,t)$ for an arbitrary value of $t$ remains an open problem \cite{C14,aubrun2017alice}. However, for a specific range of values of the parameter $t$, there exist upper bounds on the maximal cardinality of a packing $P(n,t)$. We recall one such result in the following proposition. 
\begin{proposition}[Exercise 5.13 (i)~\cite{aubrun2017alice}]\label{prop: cardinality of codes}
    Let $n\in\mathbb N$ and $t<\frac{1}{\sqrt{n}}$. Let $\{x_i\}_{i=1}^N$ be the collection of unit vectors in $\mathbb R^n$, such that $|(x_i,x_j)|\leq t$ for all $i\neq j$. Then $ A(n,t)\leq n/(1-nt^2)$.
\end{proposition}
In the following, we will use this property to show that geometrical constraints of vector packing in Euclidean space imply bounds on efficiently extractable entanglement (see Section \ref{sec: complexity vs entanglement}).  
\section{Computational Entanglement}\label{sec: comp entanglement}
In this section, we extend selected properties of entanglement measures to the computational entanglement framework. These extensions address entanglement measures that admit functional lower/upper bounds. We then recall key notions from computational entanglement theory as introduced in~\cite{arnon2023computational}. 

In Subsection~\ref{subsec:entanglement properties}, we introduce extensions of some desirable properties of entanglement measures. In Subsection~\ref{subsec: computational entanglement}, we recall the notion of computational distillable entanglement and entanglement cost in the one-shot setting — computational analogs of the one-shot distillable entanglement and one-shot entanglement cost in Section \ref{sec: intro information theory}. Finally, in Subsection~\ref{subsec: uniform computational entanglement}, we recall the computational distillable entanglement and entanglement cost in the uniform setting.

\subsection{Properties of Computational Entanglement Measures}\label{subsec:entanglement properties}
In this subsection, we introduce refined formulations of well-known entanglement measure properties. These generalized definitions capture the essential features of properties such as convexity, concavity, subadditivity, superadditivity, invariance under local unitaries, LOCC monotonicity, and continuity bounds. In particular, computational entanglement theory provides only functional upper or lower bounds rather than scalar values. This motivates the definitions below, which adapt the properties to measures that admit only such bounds.
\begin{definition}
    Let $E$ be an entanglement measure and let $\{\rho_{x}^\lambda\}_{\lambda\in\mathbb N, x\in X}$ be a family of bipartite states indexed by a growth parameter $\lambda\in\mathbb N$ and a finite alphabet $X$. We say that $E$ is \emph{lower-bound convex} (\emph{upper-bound concave}) if
    \begin{equation*}
        E(\{\sigma^\lambda\})\underset{\displaystyle (\leq)}{\geq} n(\lambda)\implies \forall x\in X,\,E(\{\rho^\lambda_x\})\underset{\displaystyle (\leq)}{\geq} n(\lambda) \quad\text{with}\quad \sigma^\lambda:=\sum_{x\in X}p(x)\rho^\lambda_x.
    \end{equation*}
\end{definition}

\begin{definition}
    Let $\{\rho^\lambda_1\}_{\lambda\in\mathbb N}$ and $\{\rho^\lambda_2\}_{\lambda\in\mathbb N}$ be two families of bipartite states with entanglement measure $E$ admitting two lower bounds (upper bounds) $n_1(\lambda)$ and $n_2(\lambda)$. We say the entanglement measure $E$ is \emph{lower-bound superadditive} (\emph{upper-bound subadditive}) if the following holds
    \begin{equation*}
 E(\{\rho^\lambda_1\otimes\rho^\lambda_2\})\underset{\displaystyle (\leq)}{\geq}  n_1(\lambda)+n_2(\lambda).
    \end{equation*}
\end{definition}

\begin{definition}\label{def: invariance under LU}
    An entanglement measure $E$ with a lower bound (upper bound) $m(\lambda)$  for a family of states $\{\rho^\lambda_{AB}\}_{\lambda\in\mathbb N}$ is said to be invariant under a family of local unitaries if the following holds
    \begin{equation*}
        E(\{U^\lambda_{AB}\rho^\lambda_{AB}U^{\lambda\dagger}_{AB}\})\underset{\displaystyle (\leq)}{\geq} m(\lambda)\iff E(\{\rho^\lambda_{AB}\})\underset{\displaystyle (\leq)}{\geq} m(\lambda),
    \end{equation*}
    for all families of local unitaries $\{U^\lambda_{AB}:=U^\lambda_A\otimes V^\lambda_B\}_{\lambda\in\mathbb N}$.
\end{definition}

\begin{definition}\label{def: upper bound LOCC monotonicity}
    Let $E$ be an entanglement measure with a valid upper bound $n(\lambda)$ (or lower bound $m(\lambda)$) for a family of states $\{\rho^\lambda_{AB}\}_{\lambda\in\mathbb N}$. We say that $E$ is LOCC monotone if for any family of LOCC maps $\{\Lambda^\lambda\}_{\lambda\in\mathbb N}$ the following holds:
    \begin{align*}
        E(\{\rho^\lambda_{AB}\})&\leq n(\lambda)\implies E(\{\Lambda^\lambda(\rho^\lambda_{AB})\})\leq n(\lambda) \\
       E(\{\Lambda^\lambda(\rho^\lambda_{AB})\})&\geq m(\lambda)\implies E(\{\rho^\lambda_{AB}\})\geq m(\lambda).
    \end{align*}
\end{definition}
\begin{definition}\label{def: continuity}
    Let $E$ be an entanglement measure with a valid upper bound $n(\lambda)$ (or lower bound $m(\lambda)$) for a family of states $\{\rho^\lambda_{AB}\}_{\lambda\in\mathbb N}$. We say that $E$ has a continuity bound if for any family of states $\{\sigma^\lambda_{AB}\}_{\lambda\in\mathbb N}$, such that $\frac{1}{2}\|\rho^\lambda_{AB}-\sigma^\lambda_{AB}\|_1\leq \delta(\lambda)$ with $\delta:\mathbb N\rightarrow[0,1)$, the following holds
    \begin{align*}
        E(\{\rho^\lambda_{AB}\})\leq n(\lambda)\implies E(\{\sigma^\lambda_{AB}\}) &\leq n(\lambda) + f_\delta(\lambda) \\
         E(\{\rho^\lambda_{AB}\})\geq m(\lambda)\implies E(\{\sigma^\lambda_{AB}\}) &\geq m(\lambda) - f_\delta(\lambda),
    \end{align*}
    where $f_\delta:\mathbb N\rightarrow\mathbb N$ is a function depending on $\delta$ and, potentially,  other parameters. 
\end{definition}
\subsection{Computational One-Shot Entanglement}\label{subsec: computational entanglement} In this subsection, we recall the definitions of computational one-shot distillable entanglement and computational one-shot entanglement cost.

We begin by recalling the definition of a circuit representation for LOCC channels, followed by the notion of an efficient family of such channels. An efficient family consists of LOCC maps whose circuit descriptions scale polynomially with the parameter $\lambda$. This formalizes the idea of \emph{efficient scaling} in a computational setting; see \cite{arnon2023computational} for details.

\begin{definition}\label{def: definition of circuit representation of LOCC}
    Let $\rho_{AB}$ be an arbitrary input state on $n_A+n_B$ qubits: 
    \begin{equation*}
        \rho_{AB}\in\mathcal{D}(\mathcal{H}_{AB})\quad\text{where}\quad\mathcal{H}_{AB}:=(\C^2)^{\otimes n_A}\otimes(\C^2)^{\otimes n_B}.
    \end{equation*}
    An LOCC channel $\Gamma$ with $r$ rounds admits a circuit description if it can be represented by two families of circuits $\{\mathcal{C}_{A,i}\}_{i \in [r]}$ and $\{\mathcal{C}_{B,i}\}_{i \in [r]}$. The family $\{\mathcal{C}_{A,i}\}$ acts on registers $A$, $A'$, and $C$, consisting of $n_A + t_A + q$ qubits: $A$ is Alice's input register, $A'$ is her ancilla register, and $C$ is a shared classical communication register. The circuits $\{\mathcal{C}_{B,i}\}$ are defined analogously for Bob.
    
    Initially, $A$ and $B$ hold the input state $\rho_{AB}$; $A'$ and $B'$ are initialized in $|0\rangle^{\otimes t_A}$ and $|0\rangle^{\otimes t_B}$, respectively, and $C$ is initialized in $|0\rangle^{\otimes q}$. At each round $i \in [r]$, Alice applies $\mathcal{C}_{A,i}$ to $(A, A', C)$ and measures $C$ in the computational basis. Bob then applies $\mathcal{C}_{B,i}$ to $(B, B', C)$ and measures $C$ again. The output state is stored in the first $m_A$ qubits of $(A, A')$ and the first $m_B$ qubits of $(B, B')$.
\end{definition} 

\begin{figure*}
	\centering
    \includegraphics[width=0.7\linewidth]{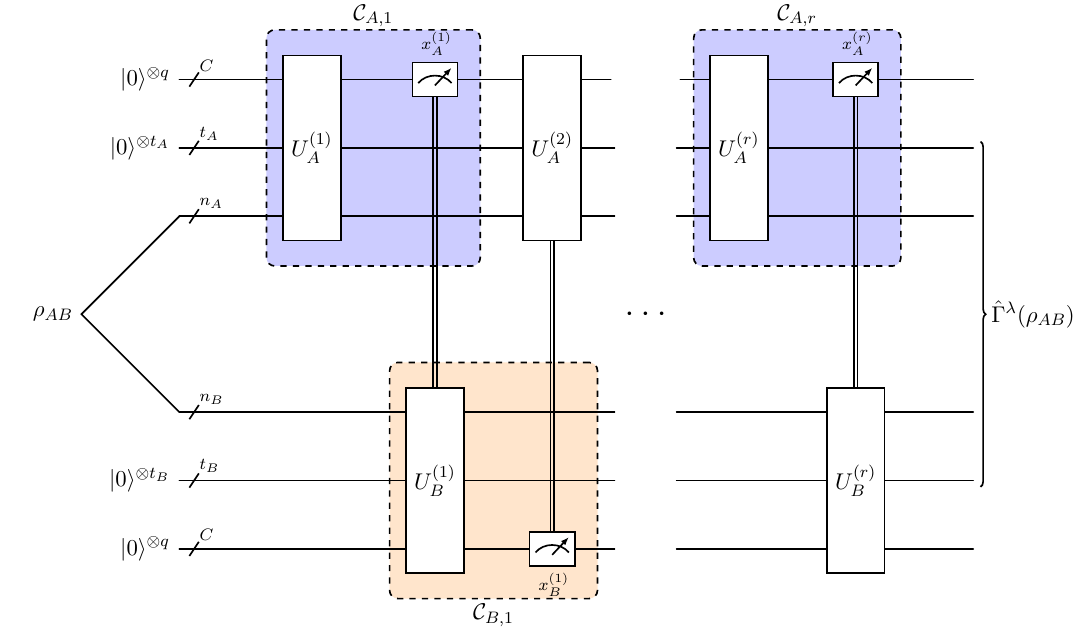}
    \caption{A circuit representation of an LOCC channel $\hat{\Gamma}$ according to Definition \ref{def: definition of circuit representation of LOCC}.}
    \label{fig:circuit LOCC}
\end{figure*}

\begin{definition}\label{def: circuit size}
     An LOCC channel $\Gamma$ is said to have \emph{circuit size} $c$ if it can be implemented by a circuit with at most $c$ total gates, including all unitaries decomposed according to a fixed finite universal gate set, ancilla creation, and measurements in the computational basis.     
\end{definition}

From a computational viewpoint, it is natural to restrict to circuit representations whose size grows at most polynomially. This motivates the following definition:

\begin{definition}\label{def: definition efficient LOCC}
A family of LOCC channels $\{\hat{\Gamma}^{\lambda}\}_{\lambda \in \mathbb N}$ is called \emph{efficient} if each $\hat{\Gamma}^\lambda$ has a circuit description of size at most $c(\lambda)$ for some polynomially bounded $c(\lambda)$.
\end{definition}
We note that Definition~\ref{def: definition efficient LOCC} corresponds to the nonuniform notion of efficiency. In particular, it does not require the existence of a polynomial-time classical algorithm that, given $\lambda$, generates the circuit implementing $\hat{\Gamma}^{\lambda}$. Quantum-network applications may additionally require the efficient synthesis of such circuits from descriptions of the states and relevant network parameters. In this work, however, we focus on the mathematical properties of the computational entanglement measures introduced in~\cite{arnon2023computational}, for which the nonuniform circuit-size restriction is sufficient. Throughout this work, the term \emph{efficient} therefore refers to this nonuniform notion unless stated otherwise.

We now recall the definitions of the \emph{computational one-shot distillable entanglement} and \emph{computational one-shot entanglement cost}. Each quantifies the number of EPR pairs that can be distilled from—or required to reproduce—a given state family, under the assumption that only efficient families of LOCC protocols are permitted. These were introduced in~\cite{arnon2023computational}.

\begin{definition}\label{def: comp distil vidick}
Let $\epsilon : \mathbb N \to [0,1)$ and let  $n_A, n_B : \mathbb N \to \mathbb N$ be polynomially bounded functions. Consider a family of bipartite states $\{\rho_{AB}^\lambda\}_{\lambda \in \mathbb N}$ with
\[
\rho_{AB}^\lambda \in \mathcal{D}(\mathcal{H}_{AB}), \quad \text{where}\quad  \mathcal{H}_{AB} := (\mathbb{C}^2)^{\otimes n_A(\lambda)} \otimes (\mathbb{C}^2)^{\otimes n_B(\lambda)}.
\]

Then the \emph{computational one-shot distillable entanglement with error $\epsilon$} of the state family $\{\rho_{AB}^\lambda\}$ is denoted by
\[
\hat{E}_D^\epsilon\left(\left\{\rho_{AB}^\lambda\right\}\right).
\]

A function $m : \mathbb N \to \mathbb N$ is a \emph{valid lower bound} on $\hat{E}_D^\epsilon\left(\{\rho_{AB}^\lambda\}\right)$ if there exists an efficient family of LOCC channels $\{\hat{\Gamma}^\lambda\}$ producing $2m(\lambda)$-qubit states such that
\[
\forall \lambda \in \mathbb N, \quad p_{\emph{err}}(\hat{\Gamma}^\lambda, \rho^\lambda_{AB}):=1 - F\left(\hat{\Gamma}^\lambda(\rho_{AB}^\lambda), \Phi^{\otimes m(\lambda)}\right) \leq \epsilon(\lambda).
\]
In this case, each channel $\hat{\Gamma}^\lambda$ maps the state $\rho_{AB}^\lambda$ to a state on $\tilde{\mathcal{H}}_{AB} := (\mathbb{C}^2)^{\otimes m(\lambda)} \otimes (\mathbb{C}^2)^{\otimes m(\lambda)}$.
\end{definition}
We note that, since ancilla creation is counted in the circuit size by Definition~\ref{def: circuit size}, the statement that a function $d:\mathbb N\rightarrow\mathbb N$ is a valid lower bound on computational one-shot distillable entanglement  implies that $d(\lambda)$ is polynomially bounded.
\begin{remark}[Probability of failure in distillation]\label{remark: SLOCC}
    The original definition in~\cite{arnon2023computational} utilizes families of distillation protocols that are not allowed to fail, i.e., whose failure probability is zero. However, in the complementary framework of~\cite{leone2025entanglement}, distillation protocols are allowed to have a nonzero failure probability; see~\cite[Definition~3]{leone2025entanglement}. In this work, we assume that the protocols have a zero failure probability, with the exception of the results in Section~\ref{sec: val-cont}, where we make use of the distillation protocols that output states in the form of
    \begin{equation*}
        (1-p_{\emph{fail}})\ketbra{1}{1}_X\otimes \omega_{AB} + p_{\emph{fail}}\ketbra{0}{0}_X\otimes\tau_{AB},
    \end{equation*}
    where $p_{\emph{fail}}$ is the failure probability and $\omega_{AB}$ is a state shared by Alice and Bob such that
    \begin{equation*}
        1-F(\omega_{AB},\Phi^{\otimes m})\leq \varepsilon.
    \end{equation*}
    Notice that any such protocol can be converted to a deterministic one by tracing out the classical register $X$. Doing so results in the modification of the error of distillation. That is, if a distillation protocol obtains $m$ $\varepsilon$-approximate EPR pairs with failure probability $p_{\emph{fail}}$, then there exists a protocol that distills $m$ $(\varepsilon+ (1-\varepsilon)p_{\emph{fail}})$-approximate EPR pairs. Thus, with a slight abuse of notation, in Section~\ref{sec: val-cont} we write $\hat E_D(\{\rho^\lambda\}_\lambda)\geq m(\lambda)$ with failure probability \(p_{\mathrm{fail}}(\lambda)\) when the displayed lower bound is achieved conditionally by a family of efficient flagged LOCC protocols of the form described above. Whenever this notation is used, the corresponding failure probability is stated explicitly. Under the original deterministic definition of~\cite{arnon2023computational} (Definition~\ref{def: comp distil vidick}), such a statement implies
\begin{equation*}
\hat E_D^{\overline\varepsilon}(\{\rho^\lambda\}_{\lambda})\geq m(\lambda)\quad\text{with}\quad\overline\varepsilon(\lambda):=\varepsilon(\lambda)+\bigl(1-\varepsilon(\lambda)\bigr)p_{\mathrm{fail}}(\lambda).
\end{equation*}
\end{remark}
\begin{definition}\label{def: comp entang cost vidick}
Let $\epsilon : \mathbb N \to [0,1)$ and let $n_A, n_B : \mathbb N \to \mathbb N$ be polynomially bounded functions. Consider a family of states $\{\rho_{AB}^\lambda\}_{\lambda \in \mathbb N}$ with
\[
\rho_{AB}^\lambda \in \mathcal{D}(\mathcal{H}_{AB}), \quad \text{where}\quad \mathcal{H}_{AB} := (\mathbb{C}^2)^{\otimes n_A(\lambda)} \otimes (\mathbb{C}^2)^{\otimes n_B(\lambda)}.
\]

Then the \emph{computational one-shot entanglement cost with error $\epsilon$} is denoted by
\[
\hat{E}_C^\epsilon\left(\left\{\rho_{AB}^\lambda\right\}\right).
\]

A function $n : \mathbb N \to \mathbb N$ is a \emph{valid upper bound} on $\hat{E}_C^\epsilon\left(\{\rho_{AB}^\lambda\}\right)$ if there exists an efficient family of LOCC channels $\{\hat{\Gamma}^\lambda\}$ such that
\[
\forall \lambda \in \mathbb N, \quad  p_{\emph{err}}(\hat{\Gamma}^\lambda, \rho^\lambda_{AB}) :=1 - F\left(\hat{\Gamma}^\lambda(\Phi^{\otimes n(\lambda)}), \rho_{AB}^\lambda\right) \leq \epsilon(\lambda).
\]
Here, each channel $\hat{\Gamma}^\lambda$ maps the EPR input state on $\mathcal{H}_{\tilde{A}\tilde{B}} := (\mathbb{C}^2)^{\otimes n(\lambda)} \otimes (\mathbb{C}^2)^{\otimes n(\lambda)}$ to an approximation of the target state $\rho_{AB}^\lambda$.
\end{definition}
Finally, since the set of valid LOCCs in the computational quantities is restricted, we trivially have the following relations: 
    \begin{align*}
        \hat E^\varepsilon_D(\{\rho^\lambda\}) \geq m(\lambda) &\implies E^\varepsilon_D(\rho^\lambda)\geq m(\lambda),\,\forall\lambda\in\mathbb N \\
        \hat E^\varepsilon_C(\{\rho^\lambda\}) \leq n(\lambda) &\implies E^\varepsilon_C(\rho^\lambda)\leq n(\lambda),\,\forall\lambda\in\mathbb N.
    \end{align*}

\subsection{Uniform Computational Entanglement}\label{subsec: uniform computational entanglement}
In this subsection, we recall the \emph{uniform} versions of computational distillable entanglement and cost. 

The concept of computational one-shot distillable entanglement/cost can be generalized to the case in which one wants to distill/dilute multiple states corresponding to one value of the parameter $\lambda \in \mathbb N$. The  requirement for this quantity is that the maximal operational error should remain the same for all states corresponding to the same $\lambda$. This requirement corresponds to the notion of \emph{uniformity}.

\begin{definition}\label{def: uniform computational distillable entanglement}
    Let $\epsilon:\mathbb N\to[0,1)$ and let $n_A,n_B,\kappa:\mathbb N\to \mathbb N$ be polynomially bounded functions. Let the encoded family of quantum states be $\{\rho^k_{AB}\}_{k\in\{0,1\}^{\kappa(\lambda)}}$ with
    
    \[
\rho_{AB}^k \in \mathcal{D}(\mathcal{H}_{AB}), \quad \text{where}\quad \mathcal{H}_{AB} := (\mathbb{C}^2)^{\otimes n_A(\lambda)} \otimes (\mathbb{C}^2)^{\otimes n_B(\lambda)}.
\]

    Then the \emph{computational uniform distillable entanglement with error $\epsilon$} is denoted by
\[
\hat{E}_D^\epsilon\left(\left\{k,\rho_{AB}^k\right\}\right).
\]

    A function $m(\lambda)$ is a valid lower bound on the computational distillable entanglement of the family $\{k,\rho^k_{AB}\}$
     if there exists an efficient family of LOCC channels $\{\hat{\Gamma}^{\lambda}\}$, such that 
    \begin{equation*}
        \forall \lambda\in\mathbb N,\,\forall k\in \{0,1\}^{\kappa(\lambda)},\quad p_{\emph{err}}(\hat{\Gamma}^{\lambda},k,\rho^k_{AB})\leq\epsilon(\lambda),
    \end{equation*}
    where
    \begin{equation*}
        p_{\emph{err}}(\hat{\Gamma}^{\lambda},k,\rho^k_{AB})=1-F\left(\hat{\Gamma}^{\lambda}(k,\rho^k_{AB}),\Phi^{\otimes m(\lambda)}\right)=1-\bra{\phi^{\otimes m(\lambda)}}\hat{\Gamma}^{\lambda}(k,\rho^k_{AB})\ket{\phi^{\otimes m(\lambda)}}.
    \end{equation*}
\end{definition}
It is clear from the definition above that a member of such an efficient family is the channel $\hat{\Gamma}$  that acts on the bipartition $A'A:BB'$. More precisely: 
    \begin{align*}
        \hat{\Gamma}^{\lambda}:\mathcal{D}(\mathcal{H}_{A'})\otimes \mathcal D(\mathcal{H}_{AB})\otimes \mathcal D(\mathcal H_{B'})&\to \mathcal{D}(\tilde{\mathcal{H}}_{AB}):=\mathcal{D}\left((\C^2)^{\otimes m(\lambda)}\otimes (\C^2)^{\otimes m(\lambda)}\right) \\
        \hat{\Gamma}^{\lambda}:\ketbra{k_{A'}}{k_{A'}}\otimes \rho^k_{AB}\otimes \ketbra{k_{B'}}{k_{B'}}&\to\hat{\Gamma}^{\lambda}(\ketbra{k_{A'}}{k_{A'}}\otimes \rho^k_{AB}\otimes \ketbra{k_{B'}}{k_{B'}}),
    \end{align*}
    where $\mathcal{H}_{A'}$ and $\mathcal{H}_{B'}$ represent, respectively, extra registers for Alice and Bob encoding the different keys used in the distillation protocol via the circuit representation of an LOCC channel. 
    
As with the uniform computational distillable entanglement, we recall the definition of uniform computational entanglement cost.
\begin{definition}\label{def: uniform comp cost vidick}
    Let $\epsilon:\mathbb N\to[0,1)$ and let $n_A,n_B,\kappa:\mathbb N\to \mathbb N$ be polynomially bounded functions. Let the encoded family of quantum states be $\{\rho^k_{AB}\}_{k\in\{0,1\}^{\kappa(\lambda)}}$ with
    \[
\rho_{AB}^k \in \mathcal{D}(\mathcal{H}_{AB}), \quad \text{where}\quad\mathcal{H}_{AB} := (\mathbb{C}^2)^{\otimes n_A(\lambda)} \otimes (\mathbb{C}^2)^{\otimes n_B(\lambda)}.
\]

    Then the \emph{computational uniform entanglement cost with error $\epsilon$} is denoted by
\[
\hat{E}_C^\epsilon\left(\left\{k,\rho_{AB}^k\right\}\right).
\]

    A function $n(\lambda)$ is a valid upper bound on the uniform computational entanglement cost of the family $\{k,\rho^k_{AB}\}$
     if there exists an efficient family of LOCC channels $\{\hat{\Gamma}^{\lambda}\}$ such that 
    \begin{equation*}
        \forall \lambda\in\mathbb N,\,\forall k\in \{0,1\}^{\kappa(\lambda)},\quad p_{\emph{err}}(\hat{\Gamma}^{\lambda},k,\rho^k_{AB})\leq\epsilon(\lambda),
    \end{equation*}
    where
    \begin{equation*}
        p_{\emph{err}}(\hat{\Gamma}^{\lambda},k,\rho^k_{AB})=1-F\left(\hat{\Gamma}^{\lambda}(k,\Phi^{\otimes n(\lambda)}),\rho^k_{AB}\right).
    \end{equation*}
\end{definition}
Here, a member of the efficient family is a channel $\hat{\Gamma}$  that acts on the bipartition $A'\tilde A:\tilde BB'$. More precisely:
    \begin{align*}
        \hat{\Gamma}^{\lambda}:\mathcal{D}(\mathcal{H}_{A'})\otimes \mathcal D(\mathcal H_{\tilde A\tilde B})\otimes \mathcal D(\mathcal H_{B'})&\to \mathcal{D}(\mathcal{H}_{AB}) \\
        \hat{\Gamma}^{\lambda}:\ketbra{k_{A'}}{k_{A'}}\otimes \Phi^{\otimes n(\lambda)}\otimes \ketbra{k_{B'}}{k_{B'}}&\to\hat{\Gamma}^{\lambda}(\ketbra{k_{A'}}{k_{A'}}\otimes \Phi^{\otimes n(\lambda)}\otimes \ketbra{k_{B'}}{k_{B'}}),
    \end{align*}
    where
    \begin{equation*}
        \mathcal H_{\tilde A\tilde B}:=(\C^2)^{\otimes n(\lambda)}\otimes (\C^2)^{\otimes n(\lambda)}.
    \end{equation*}
    Moreover, the spaces $\mathcal{H}_{A'}$ and $\mathcal{H}_{B'}$ represent, respectively, extra registers for Alice and Bob encoding the different keys used in the dilution protocol via the efficient circuit representation of LOCC channels. 

With a slight abuse of notation, in the following, we write $\{\rho^{\lambda}\}_{\lambda\in\mathbb N}$ and $\{\rho^k\}_{k\in\{0,1\}^{\kappa(\lambda)}}$ instead of  $\{\rho^{\lambda}_{AB}\}_{\lambda\in\mathbb N}$ and $\{\rho^{k}_{AB}\}_{k\in\{0,1\}^{\kappa(\lambda)}}$.

\section{Basic Operations: Mixing and Tensor Product}\label{sec: basic operations}
In this section, we analyze the behavior of computational one-shot entanglement measures under two basic types of quantum state manipulations: mixing and taking tensor products of quantum states. 

In Subsection \ref{sec: convex concave}, we establish that the computational one-shot and uniform distillable entanglement and cost are lower-bound convex and upper-bound concave, respectively. In Subsection \ref{sec: subb and super distillation and cost}, we prove that computational distillable entanglement is lower-bound superadditive and computational entanglement cost is upper-bound subadditive in both the one-shot and the uniform settings. At the end of each subsection, we discuss the quantum-network interpretation of the corresponding result.

\subsection{Mixing  States: Convexity and Concavity}\label{sec: convex concave}
In this subsection, we show that the computational one-shot distillable entanglement is \emph{lower-bound convex}, while the computational one-shot entanglement cost is \emph{upper-bound concave}. These properties then naturally extend to the uniform setting.

\begin{proposition}[Lower-bound convexity of $\hat E^\varepsilon_D$]\label{prop: convexity one shot dist}
Let $\{\rho^\lambda_x\}_{x \in X, \lambda \in \mathbb N}$ be a family of bipartite quantum states, where $X$ is a finite set, and let $\sigma^\lambda := \sum_{x \in X} p(x) \rho^\lambda_x$ for each $\lambda$, where $p(x)$ is a probability distribution. Let $\{\hat \Gamma^\lambda\}_{\lambda\in\mathbb N}$ be an efficient family of channels distilling $\{\sigma^\lambda\}_{\lambda\in\mathbb N}$, and define the following, for every value $x\in X$ and $\lambda\in\mathbb N$,
\begin{equation*}
    \varepsilon_x(\lambda):=1-F(\hat\Gamma^\lambda(\rho^\lambda_x),\Phi^{\otimes m(\lambda)}).
\end{equation*}
Assume that $\varepsilon_x(\lambda)\in[0,1)$ for every $x\in X$ and $\lambda\in\mathbb N$. Then the following holds, for any fixed $x \in X$,
\begin{equation*}
    \hat E^{\varepsilon_x}_D(\{\rho^\lambda_x\})\geq m(\lambda),
\end{equation*}
and $\sum_xp(x)\varepsilon_x(\lambda)\leq \varepsilon(\lambda)$ for any $\lambda\in\mathbb N$.
\end{proposition}
\begin{proof}
Let \(\hat\Gamma^\lambda\) be an efficient LOCC channel such that \(m(\lambda)\) is a valid lower bound on \(\hat{E}_D^{\varepsilon}(\{\sigma^\lambda\})\). Since $\Phi^{\otimes m(\lambda)}$ is pure, by linearity of quantum operations and fidelity, we obtain
\begin{align*}
\varepsilon(\lambda)
&\geq 1 - \sum_{x \in X} p(x) F(\hat\Gamma^\lambda(\rho^\lambda_x), \Phi^{\otimes m(\lambda)}) \\
&= \sum_{x \in X} p(x) \left(1 - F(\hat\Gamma^\lambda(\rho^\lambda_x), \Phi^{\otimes m(\lambda)})\right).
\end{align*}
 Since \(\hat\Gamma^\lambda\) is also a valid efficient map for each \(\rho^\lambda_x\), we conclude that
\[
\hat{E}_D^{\varepsilon_x}(\{\rho^\lambda_x\}) \geq m(\lambda),
\]
where \(\varepsilon_x(\lambda) := 1 - F(\hat\Gamma^\lambda(\rho^\lambda_x), \Phi^{\otimes m(\lambda)})\), 
for all \(x \in X\).
\end{proof}

\begin{proposition}[Upper-bound concavity of $\hat E^\varepsilon_C$]\label{prop: concavity one shot cost}
Let $\{\rho^\lambda_x\}_{x \in X, \lambda \in \mathbb N}$ be a family of bipartite quantum states, where $X$ is a finite set, and let $\sigma^\lambda := \sum_{x \in X} p(x) \rho^\lambda_x$ for each $\lambda$, where $p(x)$ is a probability distribution. Let $\{\hat \Gamma^\lambda\}_{\lambda\in\mathbb N}$ be an efficient family of channels diluting $\{\sigma^\lambda\}_{\lambda\in\mathbb N}$ that achieves error $\varepsilon(\lambda)$, and define the following, for every value $x\in X$ and $\lambda\in\mathbb N$,
\begin{equation*}
    \varepsilon_x(\lambda):=1-F(\hat\Gamma^\lambda(\Phi^{\otimes n(\lambda)}),\rho^\lambda_x).
\end{equation*}
Assume that $\varepsilon_x(\lambda)\in[0,1)$ for every $x\in X$ and $\lambda\in\mathbb N$. Then the following holds, for any fixed $x \in X$,
\begin{equation*}
    \hat E^{\varepsilon_x}_C(\{\rho^\lambda_x\})\leq n(\lambda),
\end{equation*}
and $\sum_xp(x)\varepsilon_x(\lambda)\geq \varepsilon(\lambda)$ for any $\lambda\in\mathbb N$.
\end{proposition}
\begin{proof}
Let $\hat{\Gamma}^\lambda$ be an efficient LOCC channel such that $n(\lambda)$ is a valid upper bound on $\hat{E}_C^\varepsilon(\{\sigma^\lambda\})$. By concavity of fidelity (Proposition \ref{prop: concav fidelity}) and the assumption that the dilution error is achieved with the map $\hat\Gamma^\lambda$,
\begin{align*}
\varepsilon(\lambda)
&= 1 - F\left(\hat{\Gamma}^\lambda(\Phi^{\otimes n(\lambda)}), \sum_{x \in X} p(x) \rho^\lambda_x\right) \\
&\leq \sum_{x \in X} p(x) \left(1 - F(\hat{\Gamma}^\lambda(\Phi^{\otimes n(\lambda)}), \rho^\lambda_x)\right).
\end{align*}
 Since $\hat{\Gamma}^\lambda$ is a valid efficient LOCC channel for all $\rho^\lambda_x$, it follows that
\[
\hat{E}_C^{\varepsilon_x}(\{\rho^\lambda_x\}) \leq n(\lambda),
\]
where $\epsilon_x(\lambda):= 1 - F(\hat{\Gamma}^\lambda(\Phi^{\otimes n(\lambda)}), \rho^\lambda_x)$ for all $x \in X$.
\end{proof}

The upper-bound concavity of the computational one-shot entanglement cost can be extended to the uniform setting. In the case of the computational one-shot distillable entanglement, a per-key extension can be shown as a consequence of lower-bound convexity of this computational measure. The proofs are similar to those of Propositions~\ref{prop: convexity one shot dist} and~\ref{prop: concavity one shot cost}. We omit the proofs because exactly the same arguments apply.

\begin{proposition}\label{th: conv uniform distill}
    Let $\{k,\{\rho^{k}_x\}_{x\in X}\}_{k\in\{0,1\}^{\kappa(\lambda)}}$ be a family of bipartite quantum states indexed by a key $k\in\{0,1\}^{\kappa(\lambda)}$, and
    let $\sigma^k:=\sum_{x\in X}p(x)\rho^{k}_x$ for each $k$, where $p(x)$ is a probability distribution. Let $\{\hat\Gamma^\lambda\}_{\lambda\in\mathbb N}$ be an efficient family of channels distilling $\{k,\sigma^k\}_{k\in\{0,1\}^{\kappa(\lambda)}}$, and define the following, for any $\lambda\in\mathbb N$ and $\boldsymbol{k}:=\{k_\lambda\}_{\lambda\in\mathbb N} $ with $k_\lambda\in\{0,1\}^{\kappa(\lambda)}$,
    \begin{equation*}
        \epsilon_{x,\boldsymbol{k}}(\lambda):=1-F(\hat{\Gamma}^{\lambda}(k_\lambda,\rho^{k_\lambda}_x),\Phi^{\otimes m(\lambda)}).
    \end{equation*}
    Assume that $\forall \lambda\in\mathbb N,\,\varepsilon_{x,\boldsymbol{k}}(\lambda)\in[0,1)$ for every $x\in X$ and $\boldsymbol{k}$. Then the following holds, for any fixed $x\in X$ and $\boldsymbol{k}$, 
    \begin{equation*}
         \hat E_D^{\varepsilon_{x,\boldsymbol{k}}}(\{k_\lambda,\rho^{k_\lambda}_x\})\geq m(\lambda),
    \end{equation*}
         and $\sum_{x\in X}\,p(x)\epsilon_{x,\boldsymbol{k}}(\lambda) \leq \epsilon(\lambda)$ for any $\lambda\in\mathbb N$.
\end{proposition}

\begin{proposition}\label{th: comp cost concav uniform}
Let $\{k,\{\rho^{k}_x\}_{x\in X}\}_{k\in\{0,1\}^{\kappa(\lambda)}}$ be a family of bipartite quantum states indexed by a key $k\in\{0,1\}^{\kappa(\lambda)}$, and let $\sigma^k:=\sum_{x\in X}p(x)\rho^{k}_x$ for each $k$, where $p(x)$ is a probability distribution.  Let $\{\hat \Gamma^\lambda\}_{\lambda\in\mathbb N}$ be an efficient family of channels diluting $\{k,\sigma^k\}_{k\in\{0,1\}^{\kappa(\lambda)}}$ that achieves error $\varepsilon(\lambda)$, i.e.,
\begin{equation*}
    \varepsilon(\lambda)=\max_{k}[1-F(\hat\Gamma^\lambda(k,\Phi^{\otimes n(\lambda)}),\sigma^k)],
\end{equation*}
and define the following, for every value $x\in X$ and $\lambda\in\mathbb N$,
\begin{equation*}
    \varepsilon_x(\lambda):=\max_{k}[1-F(\hat\Gamma^\lambda(k,\Phi^{\otimes n(\lambda)}),\rho^k_x)].
\end{equation*}
Assume that $\varepsilon_x(\lambda)\in[0,1)$ for every $x\in X$ and $\lambda\in\mathbb N$. Then the following holds, for any fixed $x \in X$,
\begin{equation*}
    \hat E^{\varepsilon_x}_C(\{k,\rho^k_x\})\leq n(\lambda),
\end{equation*}
and $\sum_xp(x)\varepsilon_x(\lambda)\geq \varepsilon(\lambda)$ for any $\lambda\in\mathbb N$.
\end{proposition}

We now discuss the practical meaning of Propositions~\ref{prop: convexity one shot dist}--~\ref{th: comp cost concav uniform} in a quantum-network setting. A probabilistic mixture can describe a source or link that generates different bipartite states according to a fixed probability distribution. A similar approach is used in link-layer implementations, where entanglement generation is treated as a probabilistic service and the quality of the generated states is specified for network applications~\cite{DSC19,PDT22}. Our results show that if a mixture can be efficiently distilled or diluted using a given protocol, then the same protocol determines valid bounds for each component of the mixture, with the corresponding errors satisfying the weighted-average relations. Therefore, these propositions describe how computational entanglement behaves when the output of a network source is probabilistic.

\subsection{Tensor Product: Superadditivity and Subadditivity}\label{sec: subb and super distillation and cost}

In this subsection, we analyze the behavior of the computational one-shot distillable entanglement and entanglement cost with respect to the tensor product. We then extend the obtained properties to the uniform setting.

In the following proposition, we prove that the one-shot computational distillable entanglement is lower-bound superadditive.
\begin{proposition}[Lower-bound superadditivity of $\hat E^\varepsilon_D$]\label{prop: supperaditive one shot dist}
Let \( \{\rho_1^\lambda\}_{\lambda\in\mathbb N} \) and \( \{\rho_2^\lambda\}_{\lambda\in\mathbb N} \) be families of bipartite quantum states. Suppose
\( \hat{E}_D^{\varepsilon_i}(\{\rho_i^\lambda\}) \geq m_i(\lambda) \) for \( i \in \{1,2\} \). Then:
\[
\hat{E}_D^{\varepsilon}(\{\rho_1^\lambda \otimes \rho_2^\lambda\}) \geq m_1(\lambda) + m_2(\lambda),
\]
where \( \varepsilon(\lambda) \leq \varepsilon_1(\lambda) + \varepsilon_2(\lambda) \).
\end{proposition}
\begin{proof}
For $i\in\{1,2\}$, let $\{\hat{\Gamma}_i^{\lambda}\}_{\lambda\in\mathbb N}$ be an efficient family of LOCC channels achieving $\hat{E}^{\epsilon_i}_D(\{\rho_i^{\lambda}\})\geq m_i(\lambda)$. Now consider the efficient family of LOCC channels 
$\{\hat{\Gamma}^{\lambda}_{12}\}$ defined as
\begin{equation*}
\hat{\Gamma}^{\lambda}_{12}:=\hat{\Gamma}^{\lambda}_1\otimes \hat{\Gamma}^{\lambda}_2,
\end{equation*}
allowing us to extract a certain amount of EPR pairs $\hat{E}^{\epsilon}_D(\rho_1^{\lambda}\otimes \rho_2^{\lambda})$ with error probability given by:
\begin{equation*}
    p_{\text{err}}(\hat{\Gamma}_{12}^{\lambda},\rho_1^{\lambda}\otimes \rho_2^{\lambda})=1-\bra{\phi^{\otimes m(\lambda)}}\hat{\Gamma}_{12}^{\lambda}(\rho^{\lambda}_1\otimes \rho^{\lambda}_2)\ket{\phi^{\otimes m(\lambda)}}=\epsilon(\lambda),
\end{equation*}
where $m(\lambda) \equiv m_1(\lambda) + m_2(\lambda)$.

By multiplicativity, we obtain that
\begin{equation}
\epsilon(\lambda)=1-\bra{\phi^{\otimes m(\lambda)}}\hat{\Gamma}_{12}^{\lambda}(\rho^{\lambda}_1\otimes \rho^{\lambda}_2)\ket{\phi^{\otimes m(\lambda)}}\leq 1-(1-\epsilon_1(\lambda))(1-\epsilon_2(\lambda))\leq\epsilon_1(\lambda)+\epsilon_2(\lambda),
\end{equation}
where the upper bound follows from the assumption and the particular choice of the efficient LOCC channels $\hat{\Gamma}^{\lambda}_{12}=\hat{\Gamma}_1^{\lambda}\otimes\hat{\Gamma}^{\lambda}_2$.
\end{proof}

In the following, we prove the upper-bound subadditivity property of the computational one-shot entanglement cost. 
\begin{proposition}[Upper-bound subadditivity of $\hat E^\varepsilon_C$]\label{prop: subbativity of one shot cost}
Let $\{\rho_1^{\lambda}\}_{\lambda\in\mathbb N}$ and $\{\rho_2^{\lambda}\}_{\lambda\in\mathbb N}$ be families of bipartite quantum states. Suppose that $\hat E^{\varepsilon_i}_C(\{\rho^\lambda_i\})\leq n_i(\lambda)$ for $i\in\{1,2\}$. Then:
\begin{equation*}
   \hat{E}_C^{\epsilon}(\{\rho_1^{\lambda}\otimes\rho_2^{\lambda}\})\leq n_1(\lambda) + n_2(\lambda),
\end{equation*}
where  $\epsilon(\lambda)\leq \epsilon_1(\lambda)+\epsilon_2(\lambda)$.
\end{proposition}
\begin{proof}
Consider the efficient family of LOCC channels $\{\hat{\Gamma}_{12}^{\lambda}\}$ defined as
\begin{equation*}
    \forall\lambda\in\mathbb N,\quad\hat{\Gamma}^{\lambda}_{12}:=\hat{\Gamma}^{\lambda}_1\otimes \hat{\Gamma}^{\lambda}_2,
\end{equation*}
allowing us to achieve the one-shot computational cost $\hat{E}_C^{\epsilon}(\rho_1^{\lambda}\otimes\rho_2^{\lambda})$ with error probability:
    \begin{equation*}
         p_{\text{err}}(\hat{\Gamma}_{12}^{\lambda},\rho_1^{\lambda}\otimes \rho_2^{\lambda})=1-F\left(\hat{\Gamma}^{\lambda}_{12}(\Phi_1^{\otimes n_1(\lambda)}\otimes \Phi_2^{\otimes n_2(\lambda)}),\rho_1^{\lambda}\otimes \rho_2^{\lambda}\right)=\epsilon(\lambda).
    \end{equation*}
By the tensor product property of the fidelity (Proposition \ref{prop: fidelity factorisation}), the following holds: 
\begin{equation*}
\epsilon(\lambda)=1-F\left(\hat{\Gamma}^{\lambda}_{12}(\Phi_1^{\otimes n_1(\lambda)}\otimes \Phi_2^{\otimes n_2(\lambda)}),\rho_1^{\lambda}\otimes \rho_2^{\lambda}\right)=
1-F\left(\hat{\Gamma}^{\lambda}_1(\Phi_1^{\otimes n_1(\lambda)}),\rho_1^{\lambda}\right)\,F\left(\hat{\Gamma}^{\lambda}_2(\Phi_2^{\otimes n_2(\lambda)}),\rho_2^{\lambda}\right), \end{equation*}
for the particular choice of family of LOCC channels $\{\hat{\Gamma}^{\lambda}_{12}\}_{\lambda\in\mathbb N}=\{\hat{\Gamma}^{\lambda}_1\otimes\hat{\Gamma}^{\lambda}_2\}_{\lambda\in\mathbb N}$. Therefore, from the assumption that $\hat E^{\epsilon_i}_C(\rho^{\lambda}_i)\leq n_i(\lambda)$ is achieved with error probability: 
\begin{equation*}
p_{\text{err}}(\hat{\Gamma}_i^{\lambda},\rho_i^{\lambda})=1-F\left(\hat{\Gamma}^{\lambda}_i(\Phi_i^{\otimes n_i(\lambda)}),\rho_i^{\lambda}\right)\leq\epsilon_i(\lambda),
\end{equation*}
we have 
\begin{equation}
\epsilon(\lambda)\leq 1-(1-\epsilon_1(\lambda))(1-\epsilon_2(\lambda))\leq\epsilon_1(\lambda)+\epsilon_2(\lambda).
\end{equation}
This bound is obtained for the particular choice of LOCC channels $\hat{\Gamma}^{\lambda}_{12}=\hat{\Gamma}^{\lambda}_1\otimes \hat{\Gamma}^{\lambda}_2$.

\end{proof}

We now extend the lower-bound superadditivity and upper-bound subadditivity of computational distillable entanglement and cost to the uniform scenario. The proofs are similar to those of Proposition~\ref{prop: supperaditive one shot dist} and Proposition~\ref{prop: subbativity of one shot cost} and we omit them as they use exactly the same arguments. 

\begin{proposition}\label{th: uniform distil superaditivity}
    Let $\{k_1,\rho_1^{k_1}\}_{k_1\in\{0,1\}^{\kappa(\lambda)}}$ and $\{k_2,\rho_2^{k_2}\}_{k_2\in\{0,1\}^{\kappa(\lambda)}}$ be families of encoded bipartite quantum states. Suppose that $\hat E^{\varepsilon_i}_D(\rho^{k_i}_i)\geq m_i(\lambda)$ for $i\in\{1,2\}$.
    Then: 
    \begin{equation*}
     \hat{E}_D^{\epsilon}\left(\rho_1^{k_1}\otimes\rho_2^{k_2}\right)\geq m_1(\lambda) + m_2(\lambda),
    \end{equation*}
    where  $\varepsilon(\lambda) \leq \varepsilon_1(\lambda) + \varepsilon_2(\lambda)$.
        \end{proposition}

\begin{proposition}\label{th: subadditivity unif cost}
Let $\{k_1,\rho_1^{k_1}\}_{k_1\in\{0,1\}^{\kappa(\lambda)}}$ and $\{k_2,\rho_2^{k_2}\}_{k_2\in\{0,1\}^{\kappa(\lambda)}}$ be families of encoded bipartite quantum states. Suppose that $\hat E^{\varepsilon_i}_C(\rho^{k_i}_i)\leq n_i(\lambda)$ for $i\in\{1,2\}$.
Then: 
\begin{equation*}
\hat{E}_C^{\epsilon}\left(\rho_1^{k_1}\otimes\rho_2^{k_2}\right)\leq n_1(\lambda) + n_2(\lambda),    
\end{equation*}
where $\epsilon(\lambda)\leq \epsilon_1(\lambda)+\epsilon_2(\lambda)$.
\end{proposition}
In the statement of Proposition~\ref{th: uniform distil superaditivity} and Proposition~\ref{th: subadditivity unif cost}, we make use of the implicit notations introduced in Subsection~\ref{subsec: uniform computational entanglement} when working with efficient families of encoded LOCC channels $\{\Gamma^k\}$. We also use shorthand notation for computational uniform distillable entanglement and cost: $\hat{E}_{D(C)}^{\epsilon}\left(\rho_1^{k_1}\otimes\rho_2^{k_2}\right)$ instead of $\hat{E}_{D(C)}^{\epsilon}\left(\left\{k_1\,k_2,\rho_1^{k_1}\otimes \rho_2^{k_2}\right\}\right)$ and $\hat{E}_{D(C)}^{\epsilon_i}\left(\rho_i^{k_i}\right)$ instead of $\hat{E}_{D(C)}^{\epsilon_i}\left(\left\{k_i,\rho_i^{k_i}\right\}\right)$ for $i\in\{1,2\}$.  

We now provide an interpretation of Propositions~\ref{prop: supperaditive one shot dist}--\ref{th: subadditivity unif cost} in the context of quantum networks. A tensor product can represent entanglement stored in separate network links, modes, or memory registers. Network routing and resource-allocation protocols may need to process several such resources simultaneously~\cite{CERW20,LLLC21}. Our results show that efficient distillation or dilution protocols for separate families of states can be combined into a single efficient protocol acting on their tensor product. As a result, the corresponding lower bounds on the number of distilled EPR pairs or upper bounds on the number of required EPR pairs add, while the overall error is bounded by the sum of the individual errors. Hence, these propositions describe how computational entanglement resources combine when several entangled states are processed in parallel.

\section{Quantitative Relation of Geometry and Computational Entanglement}\label{sec: complexity vs entanglement}
In this section, we investigate how the geometrical structure of a family of pure states is related to the value of its efficiently distillable entanglement. The discussion proceeds in the following way:
In Subsection~\ref{subsec: S_lambda}, we provide a  self-contained introduction to a special family of pure states $S^\lambda$, which is closely related to a geometrical construct of $\eta$-nets.  In Subsection~\ref{subsec: comp distil of s u}, we show a quantitative connection between the structure of $S^\lambda$ generated by the projective unitary group and its computational one-shot distillable entanglement $\hat E^\varepsilon_D$. In Subsection~\ref{subsec: erasing entanglement}, we show that when $S^\lambda$ is generated by the projective special orthogonal group one can make a more refined argument to erase a large fraction of efficiently distillable entanglement. We then discuss the practical implications of the obtained results. The results of this section establish an unconditional quantitative separation between computational and information-theoretic one-shot distillable entanglement and identify how its magnitude depends on the geometry of the underlying family of states. 

\subsection{Family $S^\lambda$ and Connection to $\eta$-nets on Projective Groups}\label{subsec: S_lambda}
In this subsection, we introduce a projectively well-defined modification of the family of pure states $S^\lambda$ proposed in~\cite[Lemma~4.1]{arnon2023computational}. This construction will allow us to quantitatively study the dependence of efficiently distillable entanglement on the geometrical structure of a family of states. 

First, we introduce an $\eta$–net on the unitary group \(\mathbb{U}(d)\) and the special orthogonal group \(\operatorname{SO}(d)\).
\begin{definition}\label{def: eta-net on U and SO}
    Let $M$ be either $\mathbb U(d)$ or $\operatorname{SO}(d)$. Then an $\eta$-net $\mathcal N^M_{\eta}$ in $(M,\|\cdot\|_2)$ is a finite subset of $M$ that satisfies the following condition:
    \begin{equation*}
   \forall V\in M: \exists U\in \mathcal N^M_{\eta},\,\frac{1}{\sqrt {d}}\|U-V\|_2\leq\eta.
    \end{equation*}
\end{definition}
As part of this work, we formalize the modified construction by considering $\eta$-nets on the projective unitary and projective special orthogonal groups. We recall the definitions of the underlying projective groups in the following.
\begin{definition}\label{def: projective groups}
    Let $d\in\mathbb N$. The projective unitary group $\operatorname{P\mathbb U}(d)$ is a quotient of the unitary group $\mathbb U(d)$ by the complex unit circle $\mathbb U(1)$. That is 
    \begin{equation*}
        \operatorname{P\mathbb U}(d):=\mathbb U(d)/\mathbb U(1) = \{[U]\,\mid\, U\in\mathbb U(d)\},
    \end{equation*}
    where for all $U,V\in\mathbb U(d)$ it holds that $V\in[U]$ if and only if $\exists\theta\in\mathbb R $ such that $V = e^{i\theta}U$.
    
    Likewise, the projective special orthogonal group $\operatorname{PSO}(d)$, with $d\in\mathbb N$ even,  is a quotient of the special orthogonal group $\operatorname{SO}(d)$ by $\{\mathbb I,-\mathbb I\}$, where $\mathbb I\in\operatorname{SO}(d)$ is the identity operator:
    \begin{equation*}
        \operatorname{PSO}(d):=\operatorname{SO}(d)/ \{\mathbb I,-\mathbb I\}= \{[V]\,\mid\, V\in\operatorname{SO}(d)\},
    \end{equation*}
    where for all $U,V\in \operatorname{SO}(d)$ it holds that $V\in[U]$ if and only if either $V=U$ or $V = -U$.
\end{definition}

To define an $\eta$-net on projective groups, we now recall an extension of the Schatten 2-norm (see Subsection~\ref{subsec:States, LOCC channels and fidelity}).
\begin{definition}\label{def: quotient dist}
    Let $d\in\mathbb N$. An \emph{optimized 2-distance} (Frobenius distance) on a projective unitary group $\operatorname{P\mathbb U}(d)$ is defined for any pair $[U],[V]\in\operatorname{P\mathbb U}(d)$ in the following way
    \begin{equation*}
        \Delta_2([U],[V]):=\min_{\theta\in\mathbb R}\left\| U-e^{i\theta}V \right\|_2,
    \end{equation*}
    where $\|\cdot\|_2$ is a 2-norm (see Subsection~\ref{subsec:States, LOCC channels and fidelity}). 

    In the case of a projective special orthogonal group $\operatorname{PSO}(d)$ with $d\in\mathbb N$ even, the optimized 2-distance reduces to the following
    \begin{equation*}
        \Delta_2([U],[V]):=\min\{\|U-V\|_2, \|U+V\|_2\},
    \end{equation*}
    for all $[U],[V]\in\operatorname{PSO}(d)$.
\end{definition}
Finally, as in Definition~\ref{def: eta-net on U and SO}, we provide definitions of $\eta$-nets on $\operatorname{P\mathbb U}(d)$ and $\operatorname{PSO}(d)$.  In the latter case, note that $d\in\mathbb N$ must be even to be compatible with Definition~\ref{def: projective groups}. In the following, we will always make this assumption without explicitly mentioning it to avoid overloading the text.
\begin{definition}\label{def: eta-nets on PU/PSO}
     Let $d\in\mathbb N$. Let $M$ be either $\operatorname{P\mathbb U}(d)$ or $\operatorname{PSO}(d)$. Then an $\eta$-net $\mathcal N^M_{\eta}$ in $(M,\Delta_2)$ is a finite subset of $M$ that satisfies the following condition:
    \begin{equation*}
   \forall [V]\in M: \exists [U]\in \mathcal N^M_{\eta},\,\frac{1}{\sqrt {d}}\Delta_2([U],[V])\leq\eta.
    \end{equation*}
\end{definition}
One of the main characteristics of an $\eta$-net that we utilize in this work is its size. In the following, we derive the bounds on the minimal cardinalities of $\mathcal N^{\operatorname{P\mathbb U}(d)}_\eta$ and $\mathcal N^{\operatorname{PSO}(d)}_\eta$. For this, we first recall the result that bounds the minimal cardinality of $\eta$-nets on unitary and special orthogonal groups.
\begin{proposition}[Theorem 5.11~\cite{aubrun2017alice}]\label{prp: min card eta net}
    Let $\eta\in(0,2]$. Let  $\mathcal{N}^M_{\eta}$ be an $\eta$-net on the space $(M,\|\cdot\|_2)$, where $M$ is either $\operatorname{SO}(d)$ or $\mathbb U(d)$. Then the minimal cardinalities $N^{\operatorname{SO}(d)}_{\eta}$ and $N^{\mathbb U(d)}_{\eta}$ of the corresponding $\eta$-nets satisfy: 
    \begin{align*}
    \left(\frac{c}{\eta}\right)^{d^2}&\leq N^{\mathbb U(d)}_{\eta}\leq\left(\frac{C}{\eta}\right)^{d^2} \\
    \left(\frac{c}{\eta}\right)^{\frac{d(d-1)}{2}}&\leq N^{\operatorname{SO}(d)}_{\eta}\leq\left(\frac{C}{\eta}\right)^{\frac{d(d-1)}{2}},
    \end{align*}
    where $c$ and $C$ are universal positive constants independent of $\eta$ and $d$.
\end{proposition}
Notice that, for choices of $\eta$ larger than the constant $c>0$ in Proposition~\ref{prp: min card eta net}, the corresponding lower bounds become trivial, i.e., close to zero. However, to demonstrate a large gap between computational and information-theoretic one-shot distillable entanglement, we will actively use $\eta$-nets with $\eta$ taking large values, potentially greater than $c$ of unknown magnitude. To avoid the lower bound becoming trivial in the described setting, we now establish an alternative bound on the minimal cardinality of an $\eta$-net on the special orthogonal group, which is non-trivial in the case of large $\eta$, and does not depend on a universal constant $c>0$.
\begin{proposition}\label{prop: so refine bound}
    Let $\eta\in (0,\sqrt 2]$ and $d\geq 3$ be an even number. Let $\mathcal{N}^{\operatorname{SO}(d)}_{\eta}$ be an $\eta$-net on the space $(\operatorname{SO}(d),\|\cdot\|_2)$. Then the minimal cardinality  $N^{\operatorname{SO}(d)}_{\eta}$ of the corresponding $\eta$-net has the following lower bound
    \begin{equation*}
        N^{\operatorname{SO}(d)}_{\eta}\geq \exp{\left(\frac{d(d-2)}{8}\left(1-\frac{\eta^2}{2}\right)^2\right)}.
    \end{equation*}
\end{proposition}
\begin{proof}
Let $\mu_H$ be the Haar measure on the special orthogonal group $\operatorname{SO}(d)$. Then the following bound holds for $N^{\operatorname{SO}(d)}_{\eta}$~\cite[Equation 5.2]{aubrun2017alice}
\begin{equation}\label{eq: bound 1}
    N^{\operatorname{SO}(d)}_{\eta}\geq \frac{1}{\mu_H(\mathcal B_\eta(\mathbb I_d))},
\end{equation}
where $\mathcal B_\eta(\mathbb I_d)$ is the ball around the identity operator $\mathbb I_d$, defined as follows
\begin{equation}\label{eq: ball around identity}
    \mathcal B_\eta(\mathbb I_d):= \left\{V\in\operatorname{SO}(d) \mid \frac{1}{\sqrt d}\|V-\mathbb I_d\|_2 \leq \eta\right\}.
\end{equation}
Now, we translate the latter condition on the normalized Schatten 2-norm to the condition on the trace of an operator. For this, we use the following observation
\begin{equation}
    \frac{1}{d}\|V-\mathbb I_d\|^2_2\ = 2-\frac{2\operatorname{Tr}V}{d},
\end{equation}
which follows from expanding the definition of the Schatten 2-norm (see Subsection~\ref{subsec:States, LOCC channels and fidelity}). Therefore, the definition of the ball in Equation~\eqref{eq: ball around identity} can be rewritten as
\begin{equation}\label{eq: ball around identity 2}
    \mathcal B_\eta(\mathbb I_d)=\left\{V\in\operatorname{SO}(d) \mid\operatorname{Tr}V\geq d\left(1-\frac{\eta^2}{2}\right)\right\}.
\end{equation}

We now evaluate the volume of the ball in Equation~\eqref{eq:  ball around identity 2}. For this, we first recall that the function $f(V):=\operatorname{Tr}V $ is $\sqrt d$-Lipschitz with respect to the unnormalized Schatten 2-norm. That is, for any $V,U\in\operatorname{SO}(d)$ the following holds
\begin{align}
    \begin{split}
            |\operatorname{Tr} U - \operatorname{Tr}V| = |\operatorname{Tr(\mathbb I_d(U-V))}|&\leq\|\mathbb I_d\|_2\|U-V\|_2
            \\
            &=\sqrt d\|U-V\|_2,
    \end{split}
\end{align}
where the inequality follows from the Cauchy–Bunyakovsky–Schwarz inequality for the Hilbert-Schmidt inner product~\cite{khatri2024principles}. 

Next, we show that the average of the trace over the Haar distribution on $\operatorname{SO}(d)$ is zero. To see this, recall that the Haar measure $\mu_H$ is invariant under a change of the sign~\cite{M24}. Therefore, the following holds
\begin{equation}
    \int_{\operatorname{SO}(d)}\operatorname{Tr}(V)d\mu_H(V) = \int_{\operatorname{SO}(d)}\operatorname{Tr}(-V)d\mu_H(-V) = -\int_{\operatorname{SO}(d)}\operatorname{Tr}(V)d\mu_H(V),
\end{equation}
where the last equality follows from the linearity of the Haar integral and trace. 

Finally, we utilize the consequence of a corresponding one-sided Herbst inequality for the function $g(V):=\operatorname{Tr}(V)/\sqrt{d}$ acting on $\operatorname{SO}(d)$~\cite[Theorem 5.5]{M19}. That is,
\begin{equation}\label{eq: measure 1}
    \mu_H(\operatorname{Tr}V - \mathbb E_{\mu_H}[\operatorname{Tr}V]\geq s) \leq \exp{\left(-\frac{s^2}{2C_{\operatorname{SO}(d)}d}\right)},
\end{equation}
where $C_{\operatorname{SO}(d)}=\frac{4}{d-2}$ is the log-Sobolev constant for $\operatorname{SO}(d)$~\cite[Theorem 5.16]{M19}. By substituting $\mathbb E_{\mu_H}[\operatorname{Tr}V]=0$ and $s=d(1-\frac{\eta^2}{2})$ in Equation~\eqref{eq: measure 1}, we obtain
\begin{equation}\label{eq: measure 2}
    \mu_H\left(\operatorname{Tr}V\geq d\left(1-\frac{\eta^2}{2}\right)\right) \leq\exp{\left(-\frac{d(d-2)}{8}\left(1-\frac{\eta^2}{2}\right)^2\right)}.
\end{equation}
Plugging Equation~\eqref{eq: measure 2} in Equation~\eqref{eq: bound 1} yields the final result:
\begin{equation}
    N^{\operatorname{SO}(d)}_\eta\geq\exp{\left(\frac{d(d-2)}{8}\left(1-\frac{\eta^2}{2}\right)^2\right)}.
\end{equation}
This establishes the claim of the proposition.
\end{proof}
To bound the minimal cardinality of $\eta$-nets on the projective unitary and special orthogonal groups, we establish a relation between each such group and its non-projective counterpart. We begin our analysis with the case of the projective special orthogonal group $\operatorname{PSO}(d)$ with $d\in\mathbb N$ even. 

\begin{proposition}\label{prop: PSO and SO relation}
    Let $d\in\mathbb N$ be even. Let $\mathcal{N}^{\operatorname{PSO}(d)}_{\eta}$ be an $\eta$-net on the space $(\operatorname{PSO}(d),\Delta_2)$. Then the minimal cardinality $N_\eta^{\operatorname{PSO}(d)}$ of the corresponding $\eta$-net satisfies:
    \begin{equation*}
        \frac{1}{2}N_\eta^{\operatorname{SO}(d)}\leq N_\eta^{\operatorname{PSO}(d)}\leq N_\eta^{\operatorname{SO}(d)}.
    \end{equation*}
\end{proposition}
\begin{proof}
    First, let $\pi: \operatorname{SO}(d)\rightarrow\operatorname{PSO}(d)$ be the canonical homomorphism, which is defined by $\forall V\in\operatorname{SO}(d),\,\pi(V)=[V]$. Since the map $\pi$ is distance non-increasing, any projection of an $\eta$-net on $\operatorname{SO}(d)$ under the homomorphism $\pi$ is trivially an $\eta$-net on ($\operatorname{PSO}(d),\Delta_2$). Therefore, for the minimal cardinalities of the corresponding $\eta$-nets, the following holds
    \begin{equation}\label{eq: bound on PSO 1}
        N_\eta^{\operatorname{PSO}(d)}\leq N_\eta^{\operatorname{SO}(d)}.
    \end{equation}

    Second, assume that $\mathcal N^{\operatorname{PSO}(d)}_\eta=\{[V_i]\}_{i=1}^N$ is an $\eta$-net with the cardinality $N\in\mathbb N$. Then, the following set is an $\eta$-net on ($\operatorname{SO}(d),\|\cdot\|_2$):
    \begin{equation}\label{eq: SO net duplication}
        \mathcal N^{\operatorname{SO}(d)}_\eta = \{V_i\}_{i=1}^N\cup\{-V_i\}_{i=1}^N,
    \end{equation}
    where $V_i\in[V_i],\,\forall i\in[N]$. To see this, we first expand the definition of an $\eta$-net on ($\operatorname{PSO}(d),\Delta_2$) as follows
    \begin{equation}\label{eq: cond on PSO net}
        \forall [V]\in \operatorname{PSO}(d),\,\exists [U]\in \mathcal N^{\operatorname{PSO}(d)}_\eta\quad\text{such that}\quad \frac{1}{\sqrt d}\min\{\|U-V\|_2,\|U+V\|_2\} \leq \eta.
    \end{equation}
    Since the minimum in Equation~\eqref{eq: cond on PSO net} is achieved either on $U$ or $-U$, both of which are included in the $\eta$-net defined in Equation~\eqref{eq: SO net duplication}, it holds that
    \begin{equation}
        \forall V\in\operatorname{SO}(d),\,\exists U\in \mathcal N^{\operatorname{SO}(d)}_\eta\quad\text{such that}\quad\frac{1}{\sqrt d}\|U-V\|_2\leq \eta.
    \end{equation}
    In turn, this establishes the converse bound in the following way
    \begin{equation}
        N^{\operatorname{SO}(d)}_\eta\leq 2N,
    \end{equation}
    where $N$ is the size of the $\eta$-net on ($\operatorname{PSO}(d),\Delta_2$). Minimizing over all nets, we finally obtain that 
    \begin{equation}\label{eq: bound on PSO 2}
        N^{\operatorname{SO}(d)}_\eta\leq 2 N_\eta^{\operatorname{PSO}(d)}.
    \end{equation}

    Combining both Equation~\eqref{eq: bound on PSO 1} and Equation~\eqref{eq: bound on PSO 2} yields the desired result. 
\end{proof}

Next, we obtain a similar relation for the minimal cardinality of an $\eta$-net on the space ($\operatorname{P\mathbb U}(d),\Delta_2$).
\begin{proposition}\label{prop: card eta net PU}
  Let $d\in\mathbb N$. Let $\mathcal{N}^{\operatorname{P\mathbb U}(d)}_{\eta}$ be an $\eta$-net on the space $(\operatorname{P\mathbb U}(d),\Delta_2)$. Then the minimal cardinality $N_\eta^{\operatorname{P\mathbb U}(d)}$ of the corresponding $\eta$-net satisfies:
    \begin{equation*}
        \frac{N_{2\eta}^{\mathbb U(d)}}{N_\eta^{\mathbb U(1)}}\leq N_\eta^{\operatorname{P\mathbb U}(d)}\leq N_\eta^{\mathbb U(d)}.
    \end{equation*}
\end{proposition}
\begin{proof}
Similar to Proposition~\ref{prop: PSO and SO relation}, we start by defining the canonical homomorphism $\pi:\mathbb U(d)\rightarrow\operatorname{P\mathbb U}(d)$, such that $\forall U\in\mathbb U(d),\,\pi(U):=[U]\in \operatorname{P\mathbb U}(d)$. Since $\pi$ is distance non-increasing, the image of any $\eta$-net on ($\mathbb U(d), \|\cdot\|_2$) under the homomorphism $\pi$ is automatically an $\eta$-net on ($\operatorname{P\mathbb U}(d), \Delta_2$). Therefore, this establishes the upper bound on the minimal cardinality of an $\eta$-net on ($\operatorname{P\mathbb U}(d),\Delta_2$):
\begin{equation}\label{eq: eta-net on PU 6}
    N_\eta^{\operatorname{P\mathbb U}(d)}\leq N_\eta^{\mathbb U(d)}.
\end{equation}

Now, we prove the lower bound. For this, let $\mathcal N_\eta^{\operatorname{P\mathbb U}(d)}$ and $\mathcal N_\eta^{\operatorname{\mathbb U}(1)}$ be arbitrary $\eta$-nets on ($\operatorname{P\mathbb U}(d),\Delta_2$) and ($\operatorname{\mathbb U}(1),\|\cdot\|_2$), respectively. For convenience, we write them out explicitly in the following way:
\begin{equation}
    \mathcal N_\eta^{\operatorname{P\mathbb U}(d)}=\{[U_j]\}_{j=1}^{K_{\operatorname{P\mathbb U}(d)}}\quad\text{and}\quad \mathcal N_\eta^{\operatorname{\mathbb U}(1)}=\{e^{i\theta_j}\}_{j=1}^{K_{\operatorname{\mathbb U(1)}}},
\end{equation}
where $K_{\operatorname{P\mathbb U}(d)},K_{\operatorname{\mathbb U(1)}}\in\mathbb N$ are the cardinalities of the corresponding $\eta$-nets. 

We now show that the following set constitutes a $2\eta$-net on the space ($\mathbb U(d),\|\cdot\|_2$):
\begin{equation}\label{eq: proposed eta-net on U}
    \mathcal N_{2\eta}^{\mathbb U(d)}:= \{e^{i\varphi_j}U_k\,\mid\, j\in[K_{\operatorname{\mathbb U(1)}}],\forall k\in[K_{\operatorname{P\mathbb U}(d)}]\},
\end{equation}
where $e^{i\varphi_j}\in\mathcal N_\eta^{\mathbb U(1)}, [U_k]\in\mathcal{N}_\eta^{\operatorname{P\mathbb U}(d)}$ for all values of $j\in[K_{\operatorname{\mathbb U(1)}}]$ and $k\in[K_{\operatorname{P\mathbb U}(d)}]$. 

For this, since $\mathcal N_\eta^{\operatorname{P\mathbb U}(d)}$ is an $\eta$-net by construction, for any $V\in\mathbb U(d)$ the following holds
\begin{equation}\label{eq: eta-net on PU 1}
    \exists \theta\in\mathbb R,\,\exists [U_k]\in\mathcal N_\eta^{\operatorname{P\mathbb U}(d)}\quad\text{such that}\quad\frac{1}{\sqrt{d}}\|V-e^{i\theta}U_k\|_2\leq\eta.
\end{equation}

Next, by construction of  $\mathcal N_\eta^{\operatorname{\mathbb U}(1)}$, for any $\theta\in\mathbb R$ the following holds
\begin{equation}\label{eq: eta-net on PU 2}
    \exists e^{i\theta_j}\in \mathcal N_\eta^{\operatorname{\mathbb U}(1)}\quad\text{such that}\quad\|e^{i\theta}-e^{i\theta_j}\|_2\leq \eta.
\end{equation}
Moreover, by the unitary invariance of 2-norm (Proposition~\ref{prop: inv 2 norm}), Equation~\eqref{eq: eta-net on PU 2} implies that
\begin{equation}\label{eq: eta-net on PU 3}
    \exists e^{i\theta_j}\in \mathcal N_\eta^{\operatorname{\mathbb U}(1)}\quad\text{such that}\quad \frac{1}{\sqrt{d}}\|e^{i\theta}U_k-e^{i\theta_j}U_k\|_2\leq \eta.
\end{equation}
Finally, combining Equation~\eqref{eq: eta-net on PU 1} and Equation~\eqref{eq: eta-net on PU 3} we obtain that, for any $V\in\mathbb U(d)$, there exist $e^{i\theta_j}\in \mathcal N_\eta^{\operatorname{\mathbb U}(1)}$ and $[U_k]\in\mathcal N_\eta^{\operatorname{P\mathbb U}(d)}$ such that
\begin{align}
    \begin{split}
             \frac{1}{\sqrt d}\|V- e^{i\theta_j}U_k\|_2&\leq\frac{1}{\sqrt d}\|V-e^{i\theta}U_k\|_2+ \frac{1}{\sqrt{d}}\|e^{i\theta}U_k-e^{i\theta_j}U_k\|_2\\
             &\leq2\eta,
    \end{split}
\end{align}
where the first inequality follows from the triangle inequality for the 2-norm~\cite{H12}. By Definition~\ref{def: eta-net on U and SO}, this translates to $ \mathcal N_{2\eta}^{\mathbb U(d)}$, defined in Equation~\eqref{eq: proposed eta-net on U}, being a $2\eta$-net on ($\mathbb U(d),\|\cdot\|_2$). The cardinality of such a $2\eta$-net is given by the following expression
\begin{equation}\label{eq: eta-net on PU 4}
    |\mathcal N_{2\eta}^{\mathbb U(d)}| \leq K_{\operatorname{P\mathbb U}(d)}K_{\operatorname{\mathbb U(1)}}.
\end{equation}
Minimizing Equation~\eqref{eq: eta-net on PU 4} over all possible nets $\mathcal N_\eta^{\operatorname{P\mathbb U}(d)}$ and $\mathcal N_\eta^{\operatorname{\mathbb U}(1)}$, we obtain that
\begin{equation}\label{eq: eta-net on PU 5}
   \frac{N_{2\eta}^{\mathbb U(d)}}{N_\eta^{\mathbb U(1)}} \leq N_\eta^{\operatorname{P\mathbb U}(d)}.
\end{equation}
Combining Equation~\eqref{eq: eta-net on PU 6} and Equation~\eqref{eq: eta-net on PU 5} proves the claim of the proposition.
\end{proof}

In the following, we introduce our modified definition of $S^\lambda$, generated by the corresponding projective $\eta$-nets. This definition is motivated by the family proposed in~\cite[Lemma 4.1]{arnon2023computational}. To shorten the notation, we drop the dimension of the underlying vector space of the utilized groups from the notation for $\eta$-nets when it is apparent from the context.
\begin{definition}[Family $S^\lambda$. Modified~\cite{arnon2023computational}]\label{def: S lambda}
        Let $\eta:\mathbb N \rightarrow(0,1)$, and let $m:\mathbb N\rightarrow\mathbb N$ be a polynomially bounded function. Let $\mathcal S_{\sqrt{2\eta}}^{\operatorname{P\mathbb U}}$ \emph($\mathcal S_{\sqrt{2\eta}}^{\operatorname{PSO}}$\emph) be a maximal \(\sqrt{2\eta}\)-separated set in the projective unitary group $\operatorname{P\mathbb U}(2^{m(\lambda)})$ \emph(projective special orthogonal group $\operatorname{PSO}(2^{m(\lambda)})$\emph). Then $S^\lambda_{\operatorname{P\mathbb U}}$ \emph($S^\lambda_{\operatorname{PSO}}$\emph) is a collection of pure states which is given as follows
    \begin{equation*}
         S^\lambda_{M} :=\left\{|\Phi_{[U]}^{\lambda}\rangle\,|\,
 [U] \in\mathcal S^M_{\sqrt{2\eta}} \right\},
    \end{equation*}
     where $M\in\{\operatorname{P\mathbb U}(2^{m(\lambda)}) , \operatorname{PSO}(2^{m(\lambda)})\}$ and $|\Phi_{[U]}^{\lambda}\rangle:=\mathbb I\otimes U|\phi^{\otimes m(\lambda)}\rangle$, with $|\phi\rangle_{AB}:=\frac{1}{\sqrt{2}}(|00\rangle_{AB} + |11\rangle_{AB})$ and a fixed representative $U\in [U]$.
\end{definition}
\begin{remark}[Relation to the original construction]\label{remark: comparison of def S}
    The family introduced in Definition~\ref{def: S lambda} is motivated by the construction proposed in~\cite[Lemma~4.1]{arnon2023computational}. However, to obtain the geometrical properties required for the results of this section, we modify the original construction in two ways. We now explain these modifications and their role in the subsequent analysis.

First, the construction in~\cite{arnon2023computational} considers an $\eta$-net on the unitary group $\mathbb U(d)$ equipped with the normalized $2$-norm and uses it to construct a family of pairwise separated maximally entangled states. For $U,V\in\mathbb U(d)$, the normalized $2$-distance satisfies
\begin{equation*}
\frac{1}{d}\|U-V\|_2^2
=
2\left(
1-\frac{\operatorname{Re}\operatorname{Tr}(U^\dagger V)}{d}
\right).
\end{equation*}
On the other hand, the modulus of the overlap between the corresponding maximally entangled states is given by
\begin{equation*}
\left|\langle\Phi_U|\Phi_V\rangle\right|
=
\frac{\left|\operatorname{Tr}(U^\dagger V)\right|}{d}
=
\sqrt{F\left(
\ket{\Phi_U},
\ket{\Phi_V}
\right)}.
\end{equation*}
Therefore, separation in the ordinary normalized $2$-norm does not, in general, imply separation of the corresponding quantum states. For example, if $V=-U$, then $U$ and $V$ are maximally separated in the ordinary normalized $2$-norm, while $\ket{\Phi_U}$ and $\ket{\Phi_V}$ represent the same quantum state. Consequently, the unitary-group covering argument in~\cite[Lemma~4.1]{arnon2023computational} does not directly provide the claimed lower bound on the cardinality of a pairwise fidelity-separated family of states. To account for this distinction, in Definition~\ref{def: S lambda} we construct $S^\lambda_{\operatorname{P\mathbb U}}$ using an $\eta$-net on the projective unitary group $\operatorname{P\mathbb U}(d)$ equipped with the optimized $2$-distance. In this case, it holds that
\begin{equation*}
\frac{1}{d}\Delta_2([U],[V])^2=2\left(
1-\frac{\left|\operatorname{Tr}(U^\dagger V)\right|}{d}\right)=2\left(1-\sqrt{F\left(\ket{\Phi_{[U]}},\ket{\Phi_{[V]}}\right)}
\right).
\end{equation*}
Hence, the distance between projective equivalence classes is directly related to the modulus of the overlap between the corresponding quantum states.

Second, the construction in~\cite{arnon2023computational} applies all possible quantum one-time pads to every state added at each iteration. This guarantees that the uniform mixture over the resulting family is maximally mixed and is used in the subsequent argument of that work. Since the proofs in this work do not require the uniform mixture over the entire family $S^\lambda$ to be maximally mixed, we omit the one-time-pad extension. Instead, we consider mixtures of appropriately chosen subsets of $S^\lambda$, whose geometrical and entropic properties can be controlled. In Appendix~\ref{app: cons s lambda}, we provide a formal iterative construction of both $S^\lambda_{\operatorname{P\mathbb U}}$ and $S^\lambda_{\operatorname{PSO}}$.
\end{remark}

We now state the properties of the families $S^\lambda_{\operatorname{P\mathbb U}}$ and $S^\lambda_{\operatorname{PSO}}$ that follow from the construction in Appendix~\ref{app: cons s lambda}. First, by construction and the correspondence between maximally $\eta$-separated sets and $\eta$-nets~\cite{aubrun2017alice}, both $S^\lambda_{\operatorname{P\mathbb U}}$ and $S^\lambda_{\operatorname{PSO}}$ are generated by corresponding $\eta^\prime$-nets, where $\eta^\prime (\lambda) := \sqrt{2\eta(\lambda)}$. Therefore, by Proposition~\ref{prop: PSO and SO relation}, we obtain that 
the number of states in the family $S^\lambda_{\operatorname{PSO}}$ is at least the cardinality of its corresponding $\eta^\prime$-net given in Proposition~\ref{prop: so refine bound}:
\begin{equation}\label{eq: lower bound s lambda PSO}
 |S_{\operatorname{PSO}}^{\lambda}|\geq\frac{1}{2}\exp{\left(\frac{2^{m(\lambda)}(2^{m(\lambda)}-2)}{8}\left(1-\eta(\lambda)\right)^2\right)}.
\end{equation}
For the family $S^\lambda_{\operatorname{P\mathbb U}}$ with fixed $\eta^\prime = \sqrt{2\eta}$ and $m:\mathbb N\rightarrow\mathbb N$ a non-decreasing polynomially bounded function satisfying $m\in\omega(\log_2\lambda)$, we combine Proposition~\ref{prop: card eta net PU} with Proposition~\ref{prp: min card eta net} to  obtain the following for large enough $\lambda$:
\begin{align}\label{eq: lower bound s lambda u}
    \begin{split}
         |S^\lambda_{\operatorname{P\mathbb U}}| \geq N_{\eta^\prime}^{\operatorname{P\mathbb U}(2^{m(\lambda)})}\geq\frac{N_{2\eta^\prime}^{\mathbb U(2^{m(\lambda)})}}{N_{\eta^\prime}^{\mathbb U(1)}} \geq \frac{\sqrt{2\eta}}{C}\left(\frac{c^2}{8\eta}\right)^{2^{2m(\lambda)-1}} &= \left(\frac{c^2{(2\eta)}^{\frac{1}{2^{2m(\lambda)}}}  }{8C^{\frac{1}{2^{2m(\lambda)-1}}}\eta}\right)^{2^{2m(\lambda)-1}}
         \\
         & =\left(\frac{c^2}{C^{\frac{1}{2^{2m(\lambda)-1}}}2^{3-\frac{1}{2^{2m(\lambda)}}}\eta^{1-\frac{1}{2^{2m(\lambda)}}}}\right)^{2^{2m(\lambda)-1}}
         \\
         &\geq\left(\frac{c_0^2}{2\eta}\right)^{2^{2m(\lambda)-1}}
         \\
         &=\left(\frac{c_0^2}{2\eta}\right)^{\Omega(2^{2m(\lambda)})}
    \end{split}
\end{align}
where $c_0:=\frac{c}{2\sqrt{2}}$ is a universal constant and $\eta\in(0,\min{\{\frac{1}{2},\frac{c_0^2}{2}}\})$. The last inequality follows because, for all sufficiently large values of $\lambda$,
\begin{equation}
\frac{c^2(2\eta)^{1/2^{2m(\lambda)}}}
{8C^{1/2^{2m(\lambda)-1}}\eta}
\geq
\frac{c^2}{16\eta}=
\frac{c_0^2}{2\eta}.
\end{equation}
Since a lower covering bound remains valid after decreasing its universal constant, we assume without loss of generality that \(c\leq 2\sqrt{2}\). Consequently, \(c_0=c/(2\sqrt{2})\leq1\), and hence $\min\left\{\frac12,\frac{c_0^2}{2}\right\}=\frac{c_0^2}{2}$.

Second, again by construction, both families $S^\lambda_{\operatorname{P\mathbb U}}$ and $S^\lambda_{\operatorname{PSO}}$ have a specific geometrical structure, inherited from the corresponding $\sqrt{2\eta}$-nets. That is, for all $\lambda\in\mathbb N$, the following relation holds for $S^\lambda_{\operatorname{P\mathbb U}}$
\begin{equation}\label{eq: cond on vec in S_U}
\forall\left|\Phi_{[U]}\right\rangle,\left|\Phi_{[V]}\right\rangle \in S_{\operatorname{P \mathbb U}}^\lambda:\left|\Phi_{[U]}\right\rangle \neq\left|\Phi_{[V]}\right\rangle \quad\text {it holds that} \quad\left|\left\langle\Phi_{[U]}| \Phi_{[V]}\right\rangle\right| \leq 1-\eta(\lambda),
\end{equation}
and for $S^\lambda_{\operatorname{PSO}}$
\begin{equation}\label{eq: cond on vec in S_SO}
\forall\left|\Phi_{[U]}\right\rangle,\left|\Phi_{[V]}\right\rangle \in S_{\operatorname{PSO}}^\lambda:\left|\Phi_{[U]}\right\rangle \neq\left|\Phi_{[V]}\right\rangle \quad\text {it holds that} \quad\left|\left\langle\Phi_{[U]}| \Phi_{[V]}\right\rangle\right| \leq 1-\eta(\lambda),
\end{equation}
Since, for all $\lambda\in\mathbb N$,  $S^\lambda_{\operatorname{PSO}}$ consists of unit-normalized real vectors, the above-mentioned restriction provides a geometrical interpretation of the resulting  family of states. That is, $S^\lambda_{\operatorname{PSO}}$ constitutes a vector packing on a Euclidean sphere. We make this observation formal in the following proposition.
\begin{proposition}[$S^\lambda_{\operatorname{PSO}}$ is a Packing on a Sphere]\label{lemma: S so spherical}
    Let $\lambda\in\mathbb N$, $\eta:\mathbb N\rightarrow(0,1)$ and $m:\mathbb N\rightarrow\mathbb N$ be a polynomially bounded function. Then, for each value of $\lambda$, the family of pure bipartite states $S^\lambda_{\operatorname{PSO}}$, constructed according to the algorithm in Appendix~\ref{app: cons s lambda}, is a vector packing on a sphere $P(2^{2m(\lambda)}, 1-\eta(\lambda))$ \emph(see Definition~\ref{def: spherical code}\emph).
\end{proposition}

\subsection{Computational One-Shot Distillable Entanglement of $S^\lambda_{\operatorname{P\mathbb U}}$}\label{subsec: comp distil of s u}
In this subsection, we compute an upper bound on the computational one-shot distillable entanglement of a family of pure states $S^\lambda_{\operatorname{P\mathbb U}}$, described in Subsection~\ref{subsec: S_lambda}, as a function of the parameter $\eta$.

For this, we first show that there are exponentially many more states in  $S^\lambda_{\operatorname{P\mathbb U}}$ than there are efficient LOCC channels acting on $2m(\lambda)$ qubits. 

\begin{lemma}\label{lemma:cardinality comparison LOCC}
    Let $\lambda\in\mathbb N$, $\eta\in(0,\frac{c_0^2}{2})$, and let $m,w:\mathbb N\rightarrow\mathbb N$ be polynomially bounded functions, with $m$ non-decreasing and $m\in\omega(\log_2\lambda)$. Let $S^\lambda_{\operatorname{P\mathbb U}}$ be the family of states given in Definition~\ref{def: S lambda}, and let $\hat{\Gamma}_p^\lambda$ be the set of efficient LOCC channels acting on $2m(\lambda)$ qubits and outputting $w(\lambda)$ qubits with circuit sizes bounded by $p\in\operatorname{Poly}(\lambda)$. Then $\exists\lambda_0\in\mathbb N$ such that, for all $\lambda\geq\lambda_0$
    \begin{equation*}
         \frac{|\hat{\Gamma}_p^\lambda|}{|S^\lambda_{\operatorname{P\mathbb U}}|} \leq \left(\frac{c_0^2}{2\eta}\right)^{-\Omega(2^{2m(\lambda)})},
    \end{equation*}
    where $c_0>0$ is a universal constant.
\end{lemma}
\begin{proof}
We compare the growth rates of the cardinality of the efficient LOCC family and the constructed state set $S^\lambda_{\operatorname{P\mathbb U}}$.

By assumption, the circuit representations of efficient LOCC channels have size  at most $p(\lambda)$. Hence, by a counting argument, the total number of quantum circuits of size at most $p\in\operatorname{Poly}(\lambda)$, which act on $2m(\lambda)$ qubits and output $w(\lambda)$ qubits, is itself bounded by 
\begin{equation*}
    |\hat{\Gamma}^{\lambda}_p|\leq 2^{\operatorname{Poly}(\lambda)}.
\end{equation*}
Moreover, recall from Equation~\eqref{eq: lower bound s lambda u} that 
    \begin{equation*}
    |S^\lambda_{\operatorname{P\mathbb U}}| \geq \left(\frac{c_0^2}{2\eta}\right)^{\Omega(2^{2m(\lambda)})}.
    \end{equation*}
Putting these estimates together gives

\begin{align*}
    \frac{|\hat{\Gamma}_p^\lambda|}{|S^\lambda_{\operatorname{P\mathbb U}}|} \leq \left(\frac{c_0^2}{2\eta}\right)^{-\Omega(2^{2m(\lambda)})} \left(\frac{1}{2}\right)^{-\operatorname{Poly}(\lambda)} = \left(\frac{c_0^2}{2\eta}\right)^{-\Omega(2^{2m(\lambda)})} \left(\frac{c_0^2}{2\eta}\right)^{-\left(\log_{\frac{2\eta}{c_0^2}} 2\right)\operatorname{Poly}(\lambda)}.
\end{align*}
Therefore, there exists $\lambda_0\in\mathbb N$ such that, for all $\lambda\geq\lambda_0$,
\begin{equation*}
    \frac{|\hat{\Gamma}_p^\lambda|}{|S^\lambda_{\operatorname{P\mathbb U}}|} \leq
    \left(\frac{c_0^2}{2\eta}\right)^{-\Omega(2^{2m(\lambda)})}.
\end{equation*}
\end{proof}

We now compute an upper bound on computational one-shot distillable entanglement $\hat E^\varepsilon_D$ of the family of states $S^\lambda_{\operatorname{P\mathbb U}}$ as a function of the parameter $\eta$. 
\begin{remark}[Pointwise notation for \(S^\lambda\)]
\label{remark: pointwise notation S lambda}
Let \(\{S^\lambda\}_{\lambda\in\mathbb N}\) be a family of finite sets of bipartite states, where, for every \(\lambda\in\mathbb N\),
\begin{equation*}
    S^\lambda:=\{\psi^\lambda_a\}_{a\in I_\lambda},
\end{equation*}
and \(I_\lambda\) is a finite index set. Throughout this work, we use the notation $\hat E_D^\varepsilon(S^\lambda)\geq d(\lambda)$
in the following sense: the function $d:\mathbb N\rightarrow\mathbb N$ is a valid lower bound if there exists an integer-valued polynomial $p\in \operatorname{Poly}(\lambda)$ such that, for every $\lambda\in\mathbb N$ and every $a\in I_\lambda$, there exists an LOCC channel $\hat\Gamma^\lambda_a$ whose circuit size is at most $p(\lambda)$ and which satisfies
\begin{equation*}
1-F\left(\hat\Gamma^\lambda_a(\psi^\lambda_a),\Phi^{\otimes d(\lambda)}\right)\leq \varepsilon(\lambda).
\end{equation*}

Similarly, we use the notation $\hat E_C^\varepsilon(S^\lambda)\leq c(\lambda)$ in the following sense: the function $c:\mathbb N\rightarrow\mathbb N$ is a valid upper bound if there exists an integer-valued polynomial $p\in\operatorname{Poly}(\lambda)$ such that, for every \(\lambda\in\mathbb N\) and every $a\in I_\lambda$, there exists an LOCC channel $\hat\Gamma^\lambda_a$ whose circuit size is at most $p(\lambda)$ and which satisfies
\begin{equation*}
    1-F\left(\hat\Gamma^\lambda_a\left(\Phi^{\otimes c(\lambda)}\right),\psi^\lambda_a\right)\leq\varepsilon(\lambda).
\end{equation*}
\end{remark} 

\begin{proposition}[Upper Bound on $\hat E^\varepsilon_D$ of $S^\lambda_{\operatorname{P\mathbb U}}$]\label{lemma: E_D of S lambda}
    Let $\eta\in(0,c_0^2/2)$ where $c_0>0$ is a universal constant in Equation~\eqref{eq: lower bound s lambda u}, $m:\mathbb N\rightarrow\mathbb N$ be a non-decreasing polynomially bounded function with $m\in\omega(\log_2\lambda)$, and let $d:\mathbb N\rightarrow\mathbb N$ be an arbitrary function.  Let $S^\lambda_{\operatorname{P\mathbb U}}$ be a family of states generated by a $\sqrt{2\eta}$-net on the projective unitary group $\operatorname{P\mathbb U}(2^{m(\lambda)})$. Assume that $\hat E^\varepsilon_D(S^\lambda_{\operatorname{P\mathbb U}})\geq d(\lambda)$, where 
$\varepsilon:\mathbb N\rightarrow[0,1)$ is the error of distillation to $d(\lambda)$ EPR pairs. Then, for $\lambda$ large enough, the following holds
    \begin{equation*}
        d(\lambda) \leq \frac{1}{1-\sqrt{\varepsilon(\lambda)}}\left(m(\lambda)-\frac{1}{2}\operatorname{H}^*(AB)[\eta] + g_2\left(\sqrt{\varepsilon(\lambda)}\right)\right),
    \end{equation*}
    where $\operatorname{H}^*(AB)[\eta]:=1-\frac{1}{2}\eta\operatorname{log}_2\eta - \frac{1}{2}(2-\eta)\operatorname{log}_2(2-\eta)$ and the function $g_2(\cdot)$ is defined in Proposition \ref{prop: distillation and squashed}.
\end{proposition}
\begin{proof}
    First, let $p\in\operatorname{Poly}(\lambda)$ be a polynomial bound on the sizes of the circuit representations of the distillation channels in $\{\hat\Gamma^\lambda\}_{\lambda\in\mathbb N}$, which exists by the assumption $\hat E^\varepsilon_D(S^\lambda_{\operatorname{P\mathbb U}})\geq d(\lambda)$. The same assumption implies that $d(\lambda)$ is a polynomially bounded function (see discussion after Definition~\ref{def: comp distil vidick}). Recall that, for all $\lambda$ large enough, the cardinality of $S^\lambda_{\operatorname{P\mathbb U}}$ is greater than the cardinality of the set of efficient LOCC channels that act on $2m(\lambda)$ qubits and output $2d(\lambda)$ qubits with circuit size at most $p\in\operatorname{Poly}(\lambda)$ (Lemma~\ref{lemma:cardinality comparison LOCC}):
    \begin{equation*}
            \frac{|\hat{\Gamma}_p^\lambda|}{|S^\lambda_{\operatorname{P\mathbb U}}|} \leq \left(\frac{c_0^2}{2\eta}\right)^{-\Omega(2^{2m(\lambda)})}.
        \end{equation*}
By the pigeonhole principle, this implies that there exists an efficient LOCC channel that distills two states from $S^\lambda_{\operatorname{P\mathbb U}}$:
     \begin{equation*} \hat{\Gamma}^\lambda(\psi^\lambda_1)\in\mathcal B_{\varepsilon(\lambda)}(\Phi^{\otimes d(\lambda)})\quad\text{and}\quad\hat{\Gamma}^\lambda(\psi^\lambda_2)\in\mathcal B_{\varepsilon(\lambda)}(\Phi^{\otimes d(\lambda)}).
     \end{equation*}
     Therefore, the map $\hat{\Gamma}^\lambda$ distills the uniform mixture of those states with precision $\varepsilon(\lambda)$:
     \begin{equation}\label{eq: const of psi S_U}
         \hat{\Gamma}^\lambda(\psi^\lambda) \in \mathcal{B}_{\varepsilon(\lambda)}(\Phi^{\otimes d(\lambda)})\quad\text{with}\quad\psi^{\lambda}:=\frac{1}{2}\psi^{\lambda}_1+\frac{1}{2}\psi^{\lambda}_2.
     \end{equation}    
     By assumption, the valid lower bound on the 
 number of EPR pairs one can efficiently extract from $\psi^{\lambda}$ is $d(\lambda)$. Therefore, it should also be a valid lower bound for the information-theoretic one-shot distillable entanglement, hence: 
 \begin{equation*}
     \hat E^\varepsilon_D(\{\psi^\lambda\}) \geq d(\lambda) \implies E^\varepsilon_D(\psi^\lambda)\geq d(\lambda),\,\forall\lambda\in\mathbb N.
 \end{equation*}
 We have the following chain of inequalities:
     \begin{align}
         d(\lambda)&\leq E^\varepsilon_D(\psi^\lambda)\\
         &\leq\frac{1}{1-\sqrt{\varepsilon(\lambda)}}\left(E_{sq}(\psi^\lambda) + g_2\left(\sqrt{\varepsilon(\lambda)}\right)\right) \\
         &\leq \frac{1}{1-\sqrt{\varepsilon(\lambda)}}\left(m(\lambda)- \frac{1}{2}\operatorname{H}_{AB}(\psi^\lambda)+g_2\left(\sqrt{\varepsilon(\lambda)}\right)\right),\label{eq:mlupperbound}
     \end{align}
where the second inequality follows from Proposition \ref{prop: distillation and squashed}. The last inequality follows from taking the trivial extension in Definition \ref{def:squashed entanglement}. 
 
Computation shows that $\operatorname{H}_{AB}(\psi^\lambda) \equiv \operatorname{H}_{AB}(\psi^\lambda)[x]$ is given by
\begin{equation*}
        \operatorname{H}_{AB}(\psi^\lambda)[x] := 1-\frac{1}{2}(1-x)\operatorname{log}_2(1-x)-\frac{1}{2}(1+x)\operatorname{log}_2(1+x),
\end{equation*}
which is strictly decreasing on the interval $x\in(0,1)$ with $x:=|\langle\psi_1^{\lambda}|\psi_2^{\lambda}\rangle|\in[0,1)$.
   
Recall that the modulus of the overlap between two states in $S^\lambda_{\operatorname{P\mathbb U}}$ is upper bounded by $1-\eta$ by construction (see Subsection~\ref{subsec: S_lambda} for details). 
Hence, we have the following lower bound
\begin{equation}\label{eq: H star eta}
        \operatorname{H}_{AB}(\psi^\lambda)[x]\geq1-\frac{1}{2}\eta\operatorname{log}_2\eta-\frac{1}{2}(2-\eta)\operatorname{log}_2(2-\eta)=: \operatorname{H}^*(AB)[\eta].
\end{equation}
Therefore, we can further bound $d(\lambda)$ from above:
     \begin{equation}\label{eq: bound on d of S_U}
     d(\lambda)\leq \frac{1}{1-\sqrt{\varepsilon(\lambda)}}\left(m(\lambda) -\frac{1}{2}\operatorname{H}^*(AB)[\eta]+g_2\left(\sqrt{\varepsilon(\lambda)}\right)\right),
     \end{equation}
where the function $\operatorname{H}^*(AB)[\eta]$ is defined in Equation~\eqref{eq: H star eta}. This concludes the proof of the proposition.
\end{proof}
We now discuss the result obtained in Proposition~\ref{lemma: E_D of S lambda}. Proposition~\ref{lemma: E_D of S lambda} establishes a quantitative connection between the geometrical structure of $S^\lambda_{\operatorname{P\mathbb U}}$ and its efficiently distillable entanglement in the one-shot setting. We note that the function on the r.h.s. of Equation~\eqref{eq: bound on d of S_U} is a strictly decreasing function of the parameter $\eta\in(0,c_0^2/2)$. Therefore, the upper bound in Proposition~\ref{lemma: E_D of S lambda} converges to its infimum as $\eta$ approaches its supremum, for which Proposition~\ref{prp: min card eta net} gives a meaningful bound ($c_0^2/2$). In turn, this corresponds to vectors in $S^\lambda_{\operatorname{P\mathbb U}}$ being pairwise maximally separated according to Equation~\eqref{eq: cond on vec in S_U}. Therefore, we conclude that the less ``dense'' the family of states $S^\lambda_{\operatorname{P\mathbb U}}$ is, the less entanglement one can extract from it efficiently. This highlights that, despite what one might initially think, faster growth rates of the size of a family of states do not, in general, correspond to a lower value of the efficiently accessible entanglement. 

We now compare the argument of Proposition~\ref{lemma: E_D of S lambda} with the intuition introduced in~\cite[Lemma~4.1]{arnon2023computational}. The latter argument similarly compares the number of states in $S^\lambda$ with the number of efficient LOCC channels and then applies the pigeonhole principle. There is, however, a delicate nuance in using this comparison to analyze the uniform mixture over the entire family. Indeed, let $N_\lambda:=|S^\lambda|$ and let $M_\lambda$ denote the number of efficient LOCC channels under consideration. The pigeonhole principle guarantees the existence of a channel that distills at least $N_\lambda/M_\lambda$ states. Although this number can be exponentially large, the resulting subset is only guaranteed to constitute a fraction $1/M_\lambda$ of $S^\lambda$, which may vanish as $\lambda$ increases. Therefore, the cardinality comparison alone does not ensure that a single channel successfully distills almost every state in $S^\lambda$, and an additional argument would be required to transfer the fidelity guarantee to the uniform mixture over the entire family.

Proposition~\ref{lemma: E_D of S lambda} avoids the need for such an argument by considering an equal mixture of two separated states that are distilled by the same efficient LOCC channel. This is sufficient to obtain an explicit upper bound on their computational one-shot distillable entanglement. In this sense, Proposition~\ref{lemma: E_D of S lambda} builds on the counting observation used in~\cite[Lemma~4.1]{arnon2023computational}, while introducing a controlled two-state mixture and an explicit dependence on projective-unitary separation. In Subsection~\ref{subsec: erasing entanglement}, we refine this approach by using the geometrical structure of $S^\lambda_{\operatorname{PSO}}$ to control the number of linearly independent vectors in a larger subset distilled by the same channel.

We now show that the result of Proposition~\ref{lemma: E_D of S lambda} implies the existence of a family of states, for which the valid lower bound on computational one-shot distillable entanglement, in the sense of Definition~\ref{def: comp distil vidick}, is bounded by the function in Proposition~\ref{lemma: E_D of S lambda} for all large enough values of $\lambda$. Effectively, this result translates efficient distillability of the ``family of families of states'' (Remark~\ref{remark: pointwise notation S lambda}) to the distillability of a family of states. 
\begin{proposition}[Upper bound on $\hat E^\varepsilon_D$ of the diagonalized  $S^\lambda_{\operatorname{P\mathbb U}}$]\label{prop: seq of S_lambda}
    Let $\eta,\varepsilon,m$ satisfy the corresponding assumptions of Proposition~\ref{lemma: E_D of S lambda}, and let \(d:\mathbb N\to\mathbb N\) be an arbitrary function. Then there exists a family of states $\{\rho^\lambda_{AB}\}_{\lambda\in\mathbb N}$, where for each $\lambda\in\mathbb N$ it holds that $\rho^\lambda_{AB}\in S^\lambda_{\operatorname{P\mathbb U}}$, such that, if $\hat E^\varepsilon_D(\{\rho^\lambda_{AB}\}_\lambda)\geq d(\lambda)$, then $d(\lambda)$ obeys the bound in Proposition~\ref{lemma: E_D of S lambda} for all sufficiently large values of $\lambda$.
\end{proposition}
\begin{proof}
For convenience, denote the bound in
Proposition~\ref{lemma: E_D of S lambda} by
\begin{equation*}
\forall\lambda\in\mathbb N,\quad B_\eta(\lambda):=\frac{1}{1-\sqrt{\varepsilon(\lambda)}}\left(m(\lambda)-\frac{1}{2}\operatorname{H}^*(AB)[\eta]+g_2\left(\sqrt{\varepsilon(\lambda)}\right)\right).
\end{equation*}

We distinguish between two cases. First, let $d(\lambda)$ be greater than $B_\eta(\lambda)$ for only finitely many values of $\lambda$. In this case, $d(\lambda)\leq B_\eta(\lambda)$ for all sufficiently large $\lambda$, and hence the claim of the proposition follows automatically. The family
$\{\rho^\lambda_{AB}\}_{\lambda\in\mathbb N}$ can then be constructed by choosing an arbitrary state $\rho^\lambda_{AB}\in S^\lambda_{\operatorname{P\mathbb U}}$ for every $\lambda\in\mathbb N$.

Now, let $d(\lambda)$ be greater than $B_\eta(\lambda)$ for infinitely
many values of $\lambda$. For each $j\in\mathbb N$, 
define a positive polynomial $p_j\in\operatorname{Poly}(\lambda)$
as follows:
\begin{equation}\label{eq: p_j}
p_j(\lambda):=j\lambda^j.
\end{equation}
Notice that, for any polynomially bounded function $q(\lambda)$,
there exist $j_0,\lambda_0\in\mathbb N$ such that, for every
$j\geq j_0$ and every $\lambda\geq \lambda_0$,
\begin{equation}\label{eq: p_j upper bounds q}
q(\lambda)\leq p_j(\lambda),
\end{equation}
which translates to the family of polynomials $\{p_j\}_{j\in\mathbb N}$ being cofinal with respect to the ordering of eventual domination. 

Next, fix an arbitrary $j\in\mathbb N$. The proof of
Proposition~\ref{lemma: E_D of S lambda}, applied to the
polynomial $p_j(\lambda)$ in Equation~\eqref{eq: p_j}, implies
that there exists $\lambda_j\in\mathbb N$ with the following property:
for every $\lambda\geq \lambda_j$ for which
$d(\lambda)>B_\eta(\lambda)$, there exists at least one state
\begin{equation*}
\rho^\lambda_{AB}\in
S^\lambda_{\operatorname{P\mathbb U}}
\end{equation*}
that cannot be distilled to $d(\lambda)$ EPR pairs with
precision $\varepsilon(\lambda)$ by any LOCC channel whose
circuit size is at most $p_j(\lambda)$. Indeed, if every state
from $S^\lambda_{\operatorname{P\mathbb U}}$ could be distilled by some LOCC channel whose circuit size is at most \(p_j(\lambda)\), then the argument in the proof of
Proposition~\ref{lemma: E_D of S lambda} would imply that
$d(\lambda)\leq B_\eta(\lambda)$, giving a contradiction.

We now construct the family of states
$\{\rho^\lambda_{AB}\}_{\lambda\in\mathbb N}$. Since
$d(\lambda)>B_\eta(\lambda)$ for infinitely many values of
$\lambda$, we can choose a strictly increasing sequence
$\{\tilde\lambda_j\}_{j\in\mathbb N}$ such that, for every
$j\in\mathbb N$,
\begin{equation*}
\tilde\lambda_j\geq \lambda_j
\qquad\text{and}\qquad
d(\tilde\lambda_j)>B_\eta(\tilde\lambda_j).
\end{equation*}
For every $j\in\mathbb N$, choose a state
\begin{equation*}
\rho^{\tilde\lambda_j}_{AB}\in
S^{\tilde\lambda_j}_{\operatorname{P\mathbb U}}
\end{equation*}
that cannot be distilled to $d(\tilde\lambda_j)$ EPR pairs with
precision $\varepsilon(\tilde\lambda_j)$ by any LOCC channel whose
circuit size is at most $p_j(\tilde\lambda_j)$. For all values of
$\lambda$ not contained in
$\{\tilde\lambda_j\}_{j\in\mathbb N}$, define
$\rho^\lambda_{AB}$ to be an arbitrary state from
$S^\lambda_{\operatorname{P\mathbb U}}$.

Finally, we show that the constructed family
$\{\rho^\lambda_{AB}\}_{\lambda\in\mathbb N}$ cannot be
efficiently distilled to $d(\lambda)$ EPR pairs with precision
$\varepsilon(\lambda)$. Suppose, toward a contradiction, that
\begin{equation*}
\hat E_D^\varepsilon
(\{\rho^\lambda_{AB}\}_{\lambda})
\geq d(\lambda).
\end{equation*}
This implies that there exists a family of LOCC channels whose
circuit sizes are upper-bounded by some polynomially bounded
function $q(\lambda)$ and which distills $d(\lambda)$ EPR pairs
with precision $\varepsilon(\lambda)$.

By Equation~\eqref{eq: p_j upper bounds q}, there exist
$j_0,\lambda_0\in\mathbb N$ such that, for every $j\geq j_0$ and every $\lambda\geq \lambda_0$, the following holds
\begin{equation*}
q(\lambda)\leq p_j(\lambda).
\end{equation*}
Moreover, since
$\{\tilde\lambda_j\}_{j\in\mathbb N}$ is strictly increasing, there
exists $j_1\in\mathbb N$ such that $\tilde\lambda_j\geq \lambda_0$ for every
$j\geq j_1$. Therefore, for every
$j\geq\max\{j_0,j_1\}$,
\begin{equation*}
q(\tilde\lambda_j)\leq p_j(\tilde\lambda_j).
\end{equation*}
However, by construction, the corresponding state
$\rho^{\tilde\lambda_j}_{AB}$ cannot be distilled to
$d(\tilde\lambda_j)$ EPR pairs with precision
$\varepsilon(\tilde\lambda_j)$ by any LOCC channel whose circuit size
is at most $p_j(\tilde\lambda_j)$. This establishes the contradiction.

Consequently, the constructed family of
states cannot be efficiently distilled to $d(\lambda)$ EPR
pairs with precision $\varepsilon(\lambda)$. Combining the two
cases, we conclude that
\begin{equation*}
\hat E_D^\varepsilon(\{\rho^\lambda_{AB}\}_{\lambda}) \geq d(\lambda)
\end{equation*}
implies that $d(\lambda)\leq B_\eta(\lambda)$ for all sufficiently
large $\lambda$.
\end{proof}

\subsection{Erasing Efficient Entanglement of $S^\lambda_{\operatorname{PSO}}$}\label{subsec: erasing entanglement}
In this subsection, we strengthen the result of Subsection~\ref{subsec: comp distil of s u} by exploiting the additional geometrical structure of the family $S^\lambda_{\operatorname{PSO}}$. That is, for a sufficiently small error of distillation, we show that the fraction of the information-theoretically available entanglement that can be efficiently distilled tends to zero with an appropriate choice of the projective-net parameter.

The discussion proceeds in five steps. First, we recall from Proposition~\ref{lemma: S so spherical} that $S^\lambda_{\operatorname{PSO}}$ forms a vector packing on a Euclidean sphere. Second, we show that there are exponentially many more states in $S^\lambda_{\operatorname{PSO}}$ than efficient LOCC channels. The pigeonhole principle then implies the existence of an efficient LOCC channel that distills a large subset  $M^\lambda_{\operatorname{PSO}}\subset S^\lambda_{\operatorname{PSO}}$. Third, we use a spherical-code bound to relate the cardinality of $M^\lambda_{\operatorname{PSO}}$ to the dimension of its linear span. This establishes a connection between the sphere-packing structure of $S^\lambda_{\operatorname{PSO}}$ and the number of linearly independent states distilled by the same channel. Fourth, we combine this connection with a Rényi-2 entropic analysis to obtain the rank-controlled upper bound in Theorem~\ref{lemma: bound on E_D so}. Finally, we refine the entropic analysis by considering a uniform mixture of $\lceil 2^{r(\lambda)}\rceil$ states from $M^\lambda_{\operatorname{PSO}}$, which yields the sharper result in Theorem~\ref{lemma: bound on E_D so max}.

This approach refines the argument used in Proposition~\ref{lemma: E_D of S lambda}. There, the comparison between the cardinalities of the state family and the set of efficient LOCC channels guarantees the existence of two separated states distilled by the same channel. The substantially larger common-channel subset obtained for $S^\lambda_{\operatorname{PSO}}$ permits a more elaborate entropic analysis. Theorem~\ref{lemma: bound on E_D so} uses the spherical-packing structure to extract a large linearly independent subset and thereby makes explicit the relation between projective geometry and the rank of the resulting mixture. Theorem~\ref{lemma: bound on E_D so max} instead uses the entire selected subset and controls its Rényi-2 entropy directly from the pairwise-overlap bound. The former result isolates the connection between sphere packings and linear independence, while the latter provides the stronger quantitative bound on efficiently distillable entanglement.

We start by showing the relation between the number of efficient LOCC channels and the number of states in $S^\lambda_{\operatorname{PSO}}$.
\begin{lemma}\label{lemma:cardinality comparison LOCC and s so}
    Let $\lambda\in\mathbb N$, $k:\mathbb N\rightarrow(1,+\infty)$ and let  $m,w:\mathbb N\rightarrow\mathbb N$ be polynomially bounded functions with $m\in\omega(\log_2\lambda)$. Define $\eta(\lambda):=1-2^{-m(\lambda)/k(\lambda)}$ and suppose that $m(\lambda)\left(1-1/k(\lambda)\right)\in\omega(\log_2\lambda)$. Let $S^\lambda_{\operatorname{PSO}}$ be the family of states given in Definition~\ref{def: S lambda}, and let $\hat{\Gamma}^\lambda_p$ be the set of efficient LOCC channels acting on $2m(\lambda)$ qubits and outputting $w(\lambda)$ qubits with circuit sizes at most $p\in\operatorname{Poly}(\lambda)$. Then $\exists\lambda_0\in\mathbb N$ such that for all $\lambda\geq\lambda_0$:
    \begin{equation*}
         \frac{|S^\lambda_{\operatorname{PSO}}|}{|\hat{\Gamma}^\lambda_p|} \geq 2^{r(\lambda)},
    \end{equation*}
    where $r(\lambda):=\frac{1}{32\ln2}2^{2m(\lambda)\left(1-\frac{1}{k(\lambda)}\right)}$ for all values of $\lambda\in\mathbb N$.
\end{lemma}
\begin{proof}
    First, by Equation~\eqref{eq: lower bound s lambda PSO} in Subsection~\ref{subsec: S_lambda}, we obtain the following lower bound on the cardinality of $S^\lambda_{\operatorname{PSO}}$ for large enough $\lambda$:
    \begin{equation}\label{eq: card  of s lambda so}
        |S^\lambda_{\operatorname{PSO}}| \geq \frac{1}{2}\exp{\left(\frac{2^{m(\lambda)}(2^{m(\lambda)}-2)}{8}\left(1-\eta(\lambda)\right)^2\right)}.
    \end{equation}
    Now, we plug in the value of $\eta=\eta(\lambda)$ in Equation~\eqref{eq: card  of s lambda so} to obtain that, for all large enough values of $\lambda\in\mathbb N$, the following holds
    \begin{align}\label{eq: refined bound on s so}
        \begin{split}
            |S^\lambda_{\operatorname{PSO}}| \geq \frac{1}{2}\exp{\left(\frac{2^{m(\lambda)}(2^{m(\lambda)}-2)}{8}\left(1-\eta(\lambda)\right)^2\right)}
            &=2^{\left[\frac{1}{{8\ln 2}}2^{m(\lambda)}2^{-\frac{2m(\lambda)}{k(\lambda)}}(2^{m(\lambda)}-2)-1\right]} \\
            & \geq 2^{\frac{1}{16\ln2}2^{2m(\lambda)\left(1-\frac{1}{k(\lambda)}\right)}-1},
        \end{split}
    \end{align}
    where in the second inequality we used the fact that, for large enough $\lambda\in\mathbb N$, the following holds: $2^{m(\lambda)}(2^{m(\lambda)}-2) \geq 2^{2m(\lambda)-1}$.
    
    Finally, denoting $d(\lambda):=\frac{1}{16\ln2}2^{2m(\lambda)\left(1-\frac{1}{k(\lambda)}\right)}$ and dividing Equation~\eqref{eq: refined bound on s so} by $|\hat\Gamma^\lambda_p|\leq 2^{\operatorname{Poly}(\lambda)}$ (see Lemma~\ref{lemma:cardinality comparison LOCC}), we obtain that for all $\lambda$ large enough the following holds 
    \begin{align}
    \begin{split}
                 \frac{|S^\lambda_{\operatorname{PSO}}|}{|\hat{\Gamma}^\lambda_p|} \geq 2^{d(\lambda)-\operatorname{Poly}(\lambda)-1} &= 2^{d(\lambda)\left(1-\frac{\operatorname{Poly}(\lambda)}{d(\lambda)}-\frac{1}{d(\lambda)}\right)} \\&=2^{d(\lambda)(1-o(1))}\\
                 &\geq 2^{r(\lambda)},
    \end{split}
    \end{align}
    where $r(\lambda):=\frac{d(\lambda)}{2} = \frac{1}{32\ln2}2^{2m(\lambda)\left(1-\frac{1}{k(\lambda)}\right)}$ for all values of $\lambda\in\mathbb N$. Therefore, we obtain the claim of the lemma. 
\end{proof}

We now show that it is impossible to pack any ``large enough'' subset  of $S^\lambda_{\operatorname{PSO}}$ in a lower-dimensional sphere. As we will show in the following, this gives a lower bound on the number of linearly independent components of any ``large enough'' subset of $S^\lambda_{\operatorname{PSO}}$. 
	Since we use this result in the fixed-$k$ regime (see Theorem~\ref{lemma: bound on E_D so}), we state and prove the following lemma for a constant $k$ in the admissible interval. The result can also be extended to $k=k(\lambda)$ with minor modifications to the assumptions.
\begin{lemma}[Packing of $S^\lambda_{\operatorname{PSO}}$]\label{lemma: pacing of s_lambda}
    Let $\lambda\in\mathbb N$, $k\in(1,2)$, $m:\mathbb N\rightarrow\mathbb N$ be a polynomially bounded function with $m\in\omega(\log_2\lambda)$,  and $r:\mathbb N\rightarrow\mathbb R$, such that $r(\lambda):=\frac{1}{32\ln2}2^{2m(\lambda)\left(1-\frac{1}{k}\right)}$. Let $S^\lambda_{\operatorname{PSO}}$ be a family of states generated by an $\eta$-net on the group $\operatorname{PSO}(2^{m(\lambda)})$ according to Definition~\ref{def: S lambda} with $\eta(\lambda): =  
    1-\frac{1}{2^{m(\lambda)/k}}$. Then for all large enough values of $\lambda$, for any subset $M^\lambda_{\operatorname{PSO}}\subset S^\lambda_{\operatorname{PSO}}$ with cardinality $|M^\lambda_{\operatorname{PSO}}|\geq \left\lceil 2^{r(\lambda)}\right\rceil$, 
    it holds that 
    \begin{equation*}
       \dim\operatorname{span}M^\lambda_{\operatorname{PSO} }> \lfloor r(\lambda) \rfloor.
    \end{equation*}
   	\end{lemma}
\begin{proof}
    We prove this statement by contradiction. For this, assume that $M^\lambda_{\operatorname{PSO}}$ is a vector packing on a sphere $P(\lfloor r(\lambda)\rfloor,1-\eta(\lambda))$. Then the cardinality of $M^\lambda_{\operatorname{PSO}}$ must be upper bounded by $A(\lfloor r(\lambda) \rfloor,1-\eta(\lambda))$ by definition of a vector packing on a sphere (see Definition~\ref{def: spherical code}). By substituting $\lfloor r(\lambda)\rfloor$ for $n$  and $1-\eta(\lambda)$ for $t$ into the expression for  $A(n,t)$ in Proposition~\ref{prop: cardinality of codes}, we obtain
    the following bound for all sufficiently large values of $\lambda\in\mathbb  N_+$
    \begin{equation}\label{eq: KL bound for s}
         A(\lfloor r(\lambda) \rfloor,1-\eta(\lambda)) \leq \frac{\lfloor r(\lambda)\rfloor}{1-\lfloor r(\lambda)\rfloor(1-\eta(\lambda))^2}.
    \end{equation}
    Proposition~\ref{prop: cardinality of codes} holds since 
   \begin{equation}
       \left\lfloor r(\lambda)\right\rfloor
\bigl(1-\eta(\lambda)\bigr)^2\leq\frac{1}{32\ln 2}2^{2m(\lambda)(1-2/k)}
\longrightarrow0,
   \end{equation}
   because $k\in(1,2)$, by assumption of the lemma. Consequently, $1-\eta(\lambda)<\lfloor r(\lambda)\rfloor^{-1/2}$, for all sufficiently large $\lambda$.
    
   Substituting values of $r(\lambda)$ and $1-\eta(\lambda)$ given by assumption of the lemma in Equation~\eqref{eq: KL bound for s}, we compute that, for large enough $\lambda\in\mathbb N$, the following holds
    \begin{align}\label{eq: KL bound for s exp}
    \begin{split}
        A(\lfloor r(\lambda)\rfloor,1-\eta(\lambda)) \leq \frac{\frac{1}{32\ln2}2^{2m(\lambda)\left(1-\frac{1}{k}\right)}}{1-\frac{1}{32\ln2}2^{2m(\lambda)\left(1-\frac{1}{k}\right)}2^{\frac{-2m(\lambda)}{k }}}&\leq\frac{2^{2m(\lambda)\left(1-\frac{1}{k}\right)}}{1-2^{2m(\lambda)\left(1-\frac{2}{k}\right)}}
        \\
        & \leq 2^{2m(\lambda)\left(1-\frac{1}{k}\right)}\left(1+ 2^{2m(\lambda)\left(1-\frac{2}{k}\right)+1}\right) \\
        & = 2^{2m(\lambda)\left(1-\frac{1}{k}\right)}(1+o(1)),
    \end{split}     
    \end{align}
    where the third inequality is obtained by using the inequality $\frac{1}{1-x}\leq 1+2x, \, \forall x\in[0,\frac{1}{2}]$. The fourth equality is obtained by noticing that $2^{2m(\lambda)\left(1-\frac{2}{k}\right)+1}\in o(1),\ \forall k\in(1,2)$. 
    
    Next, we compare the right-hand side of Equation~\eqref{eq: KL bound for s exp} with the lower bound on $|M^\lambda_{\operatorname{PSO}}|$, which is given by the assumption of the lemma. For this, consider the following chain of inequalities for all large enough values of $\lambda\in\mathbb N$
    \begin{align}
        \begin{split}
            \frac{  A(\lfloor r(\lambda)\rfloor ,1-\eta(\lambda))}{|M^\lambda_{\operatorname{PSO}}|} \leq \frac{2^{2m(\lambda)\left(1-\frac{1}{k}\right)}(1+o(1))}{\left\lceil 2^{r(\lambda)}\right\rceil} &\leq \frac{\frac{3}{2}2^{2m(\lambda)\left(1-\frac{1}{k}\right)}}{ 2^{r(\lambda)}}\\
            &=
            \frac{\frac{3}{2}2^{2m(\lambda)\left(1-\frac{1}{k}\right)}}{ 2^{\frac{1}{32\ln2}2^{2m(\lambda)\left(1-\frac{1}{k}\right)}}}
            \\
            &<1,
        \end{split}
    \end{align}
    where the equality follows by substituting  $r(\lambda):=\frac{1}{32\ln2}2^{2m(\lambda)\left(1-\frac{1}{k}\right)}$, which is given by the assumption of the lemma.
    
    Therefore, for large enough $\lambda\in\mathbb N$, we get that  it is impossible to pack the elements of $M^\lambda_{\operatorname{PSO}}$ on a sphere  $\mathbb S^{\lfloor r(\lambda)\rfloor-1}$. However, if $M^\lambda_{\operatorname{PSO}}$ has at most $\lfloor r(\lambda)\rfloor$ linearly independent vectors, then its elements would constitute a vector packing $P(\lfloor r(\lambda) \rfloor,1-\eta(\lambda))$ (see Definition~\ref{def: spherical code}). As was just shown, the latter is not possible. This establishes the proof of the lemma.
\end{proof}

Finally, we use results of Lemma~\ref{lemma:cardinality comparison LOCC and s so} and Lemma~\ref{lemma: pacing of s_lambda} to quantitatively show that one can erase a large fraction of computational distillable entanglement from $S^\lambda_{\operatorname{PSO}}$. One can repeat the same diagonalization argument as in Proposition~\ref{prop: seq of S_lambda} to obtain a family of states that has an upper bound on its computational one-shot distillable entanglement. 
\begin{theorem}[Upper Bound on $\hat E^\varepsilon_D$ of $S^\lambda_{\operatorname{PSO}}$]\label{lemma: bound on E_D so}
    Let $\lambda\in\mathbb N$, $m:\mathbb N \rightarrow \mathbb N$ be a non-decreasing polynomially bounded function with $m\in\omega(\log_2\lambda)$, and $d:\mathbb N\rightarrow\mathbb N$ be an arbitrary function. Let $\eta:\mathbb N \rightarrow(0,1)$ be such that $\eta(\lambda):= 1-\frac{1}{2^{m(\lambda)/k}}$, with $k\in (1,2)$. Assume that $S^\lambda_{\operatorname{PSO}}$, the family of pure $2m(\lambda)$-bipartite states as described in Definition \ref{def: S lambda}, is efficiently distillable to $d(\lambda)$ EPR pairs with error of distillation $\varepsilon:\mathbb N\rightarrow(0,1)$ such that $\hat E^\varepsilon_D(S^\lambda_{\operatorname{PSO}})\geq d(\lambda)$. Then, for large enough $\lambda$, the following upper bound holds
    \begin{equation*}
        d(\lambda) \leq \frac{1}{1-\sqrt{\varepsilon(\lambda)}}\left(\frac{m(\lambda)}{k} +\frac{1}{2}\log_2(32\ln2)+ g_2\left(\sqrt{\varepsilon(\lambda)}\right) + \tilde \nu(\lambda)\right),
    \end{equation*}  
    where $g_2(\delta)$ is a function defined in Proposition~\ref{prop: distillation and squashed}, and $\tilde\nu (\lambda)$ is a negligible function (see Definition~\ref{def: negligible function}).
\end{theorem}
\begin{proof}
The proof follows the same strategy as the proof of Proposition~\ref{lemma: E_D of S lambda}. Therefore, we omit the repeated parts and focus solely on the main difference: the construction of states $\psi^\lambda$ (see proof of Proposition~\ref{lemma: E_D of S lambda}). First, we construct a family of states $\{\psi^\lambda\}_\lambda$, which will be used to obtain an upper bound on the amount of efficiently distillable entanglement of $S^\lambda_{\operatorname{PSO}}$. Second, we use the lower bound on the Rényi-2 entropy to estimate the entropy of states in $\{\psi^\lambda\}_\lambda$. Finally, we use the resulting estimate of the von Neumann entropy and the formula in Equation~\eqref{eq:mlupperbound} of Proposition~\ref{lemma: E_D of S lambda} to show the claim of the theorem.
    \paragraph{\textbf{Construction of $\psi^\lambda$}.}
    First, let $\hat \Gamma^\lambda_p$ be the set of LOCC channels acting on $2m(\lambda)$ qubits and outputting $2d(\lambda)$ qubits with circuit sizes bounded by $p\in\operatorname{Poly}(\lambda)$. Here, the polynomial $p(\lambda)$ is the upper bound on the circuit sizes of distillation channels in $\{\hat\Gamma^\lambda\}_{\lambda\in\mathbb N}$, which exists by the assumption $\hat E^\varepsilon_D(S^\lambda_{\operatorname{PSO}})\geq d(\lambda)$. The latter also implies that $d(\lambda)$ is a polynomially bounded function (see discussion after Definition~\ref{def: comp distil vidick}). Then, by leveraging the results of Lemma~\ref{lemma:cardinality comparison LOCC and s so}, for all $\lambda\in\mathbb N$ large enough, the following relation holds
    \begin{equation}\label{eq: s/G}
        \frac{|S^\lambda_{\operatorname{PSO}}|}{|\hat\Gamma^\lambda_p|}\geq 2^{r(\lambda)} ,
    \end{equation}
    where $r=r(\lambda)$ is defined as $r(\lambda):=\frac{1}{32\ln2}2^{2m(\lambda)(1-\frac{1}{k})}$.
    
    Therefore, by the pigeonhole principle, we obtain that there exists an LOCC channel $\hat\Gamma^\lambda\in\hat{\Gamma}^\lambda_p$ such that it distills at least $\left\lceil 2^{r(\lambda)}\right\rceil$ states from $S^\lambda_{\operatorname{PSO}}$. We denote this collection of states by $M^\lambda_{\operatorname{PSO}}$. Hence, for all large enough $\lambda\in\mathbb N$,  the lower bound on the cardinality of $M^\lambda_{\operatorname{PSO}}$ has the following form
    \begin{equation}\label{eq: card mso lower bound}
        |M^\lambda_{\operatorname{PSO}}| \geq \left\lceil 2^{r(\lambda)}\right\rceil.
    \end{equation}

    Now, since Equation \eqref{eq: card mso lower bound} holds, we apply results of Lemma~\ref{lemma: pacing of s_lambda} to $M^\lambda_{\operatorname{PSO}}$ to obtain that it is impossible to pack $M^\lambda_{\operatorname{PSO}}$ on the sphere $\mathbb S^{\lfloor r(\lambda)\rfloor-1}$, where $r=r(\lambda)$ is defined above. In turn, this implies that, for all large enough $\lambda$, there exist at least $\lfloor r(\lambda) \rfloor$ linearly independent unit vectors in every such $M^\lambda_{\operatorname{PSO}}$ (see discussion after the proof of Lemma~\ref{lemma: pacing of s_lambda}). We denote this collection of distilled vectors by $\{|\psi^\lambda_i\rangle\}_{i=1}^{\lfloor r(\lambda) \rfloor}\subset M^\lambda_{\operatorname{PSO}}$ and denote the associated set of rank-1 projectors by $\{\psi^\lambda_i\}_{i=1}^{\lfloor r(\lambda) \rfloor}:=\{\ketbra{\psi_i^\lambda}{\psi_i^\lambda}\}_{i=1}^{\lfloor r(\lambda) \rfloor}$.

    Next, similar to the argument in Proposition~\ref{lemma: E_D of S lambda}, we note that, for large enough $\lambda$, the same LOCC channel $\hat \Gamma^\lambda\in\hat{\Gamma}^\lambda_p$ distills the uniform mixture of states in $\{\psi^\lambda_i\}_{i=1}^{\lfloor r(\lambda) \rfloor}$ with precision $\varepsilon(\lambda)$:
    \begin{equation}\label{eq: psi lambda def}
        \hat \Gamma^\lambda(\psi^\lambda)\in \mathcal B_{\varepsilon(\lambda)}(\Phi^{\otimes d(\lambda)}_{AB})\quad\text{with}\quad\psi^\lambda := \frac{1}{\lfloor r(\lambda) \rfloor}\sum_{i=1}^{\lfloor r(\lambda) \rfloor}\psi^\lambda_i.
    \end{equation}
    
    \paragraph{\textbf{Putting bounds on $\operatorname{H}_{AB}(\psi^\lambda)$.}}
    Following the arguments of the proof of Proposition~\ref{lemma: E_D of S lambda}, the final step is to lower-bound $\operatorname{H}_{AB}(\psi^\lambda)$. For this, we employ the Gram matrix approach, introduced in~\cite{JS00}. For completeness, we now recall the main idea behind this method. That is, let $\{p_x,|\psi_x\rangle\}_{x\in X}$ be a finite ensemble of pure states in $\mathcal H$ with $\{p_x\}_{x\in X}$ being a collection of prior probabilities. Then, define a Gram matrix $G=(G_{ij})_{i,j=1}^{|X|}$ as follows
    \begin{equation}\label{eq: gram matrix}
        G = \sum_{i,j=1}^{|X|} \sqrt{p_ip_j}\langle\psi_i|\psi_j\rangle \ketbra{i}{j},
    \end{equation}
    where $\{|i\rangle\}_{i=1}^{|X|}$ is the computational basis of dimension $|X|$. 

    The main property of the Gram matrix that we exploit in this work lies in its eigenvalues. That is, nonzero eigenvalues of the Gram matrix (Equation~\eqref{eq: gram matrix}) are identical to those of the density matrix $\rho\in\mathcal D(\mathcal H)$ of the form $\rho:=\sum_{x\in X} p_x\ketbra{\psi_x}{\psi_x}$~\cite{JS00}. This, in turn, implies that, for all ensembles of bipartite pure states $\{p_x,|\psi_x\rangle_{AB}\}_{x\in X}$ and $n\in\mathbb N$, the following relations hold
    \begin{equation}\label{eq: properties of gram}
        \operatorname{H}_{AB}(\rho) = \operatorname{H}(G)\quad\text{and}\quad \operatorname{Tr}(\rho^n) = \operatorname{Tr}(G^n),
    \end{equation}
    where $\rho$ is the density matrix associated with the ensemble. 

    Now, for each $\lambda\in\mathbb N$, we apply the Gram matrix approach described above to the ensemble $\mathcal E^\lambda := \{\lfloor r(\lambda)\rfloor^{-1}, |\psi_i^\lambda\rangle\}_{i=1}^{\lfloor r(\lambda)\rfloor}$, where $\{|\psi^\lambda_i\rangle\}_{i=1}^{\lfloor r(\lambda) \rfloor}$ is the subset of distilled states used in Equation~\eqref{eq: psi lambda def}. Since Equation~\eqref{eq: properties of gram} holds, the von Neumann entropy of $\psi^\lambda$ is equal to that of $G^\lambda$, where $G^\lambda$ is the Gram matrix of the ensemble $\mathcal E^\lambda$ (Equation~\eqref{eq: gram matrix}). That is, for all large enough values of $\lambda\in\mathbb N$, it holds that $\operatorname{H}_{AB}(\psi^\lambda) = \operatorname H(G^\lambda)$.

    Finally, we evaluate the von Neumann entropy of $G^\lambda$. For this, consider the following
    \begin{align}\label{eq: bound on H 1}
        \begin{split}
            \operatorname{H}(G^\lambda)\geq \operatorname{H}_2(G^\lambda) = -\log_2\operatorname{Tr}\left({(G^\lambda)}^2\right) & = -\log_2\sum_{i=1}^{\lfloor r(\lambda)\rfloor}\sum_{j=1}^{\lfloor r(\lambda)\rfloor}(G^\lambda_{ij})^2\\
            &=-\log_2\left(\frac{1}{\lfloor r(\lambda)\rfloor^2}\sum_{i=1}^{\lfloor r(\lambda)\rfloor}\sum_{j=1}^{\lfloor r(\lambda)\rfloor}\langle\psi^\lambda_i|\psi^\lambda_j\rangle^2\right)\\
            &\geq \log_2\lfloor r(\lambda)\rfloor - \log_2 (\lfloor r(\lambda)\rfloor(1-\eta(\lambda))^2+1),
        \end{split}
    \end{align}
    where the first inequality follows from $\operatorname{H}(\rho)=\operatorname{H}_1(\rho)\geq \operatorname{H}_2(\rho)$~\cite{khatri2024principles}. The second equality holds since $G^\lambda$ is real symmetric by construction. The last inequality is a consequence of the geometrical structure of $S^\lambda_{\operatorname{PSO}}$ (see Subsection~\ref{subsec: S_lambda}) and trivial algebraic manipulations. 
    
    Furthermore, substituting the value of $\eta(\lambda)=1-2^{-\frac{m(\lambda)}{k}}$ and $r(\lambda):=\frac{1}{32\ln2}2^{2m(\lambda)(1-\frac{1}{k})}$, given by assumption of the theorem, into Equation~\eqref{eq: bound on H 1} yields the following for sufficiently large $\lambda\in\mathbb N$
    \begin{align}\label{eq: bound on H 2}
    \begin{split}
                \operatorname H_{AB}(\psi^\lambda) = \operatorname H(G^\lambda) &\geq \log_2\lfloor r(\lambda)\rfloor - \log_2(\lfloor r(\lambda)\rfloor(1-\eta(\lambda))^2+1)\\
                &= \log_2\lfloor r(\lambda)\rfloor - \log_2\left(\left\lfloor\frac{1}{32\ln2}2^{2m(\lambda)\left(1-\frac{1}{k}\right)}\right\rfloor2^{-\frac{2m(\lambda)}{k}}+1\right)
                \\
                &= \log_2\lfloor r(\lambda)\rfloor - \nu(\lambda),
    \end{split}
    \end{align}
    where $\nu(\lambda):=\log_2\left(\left\lfloor\frac{1}{32\ln2}2^{2m(\lambda)\left(1-\frac{1}{k}\right)}\right\rfloor2^{-\frac{2m(\lambda)}{k}}+1\right)$ is a negligible function because $k$ lies in the interval $(1,2)$ and $m\in\omega(\log_2\lambda)$ by the assumption of the theorem.

    We further simplify this bound by substituting the value of $r(\lambda):=\frac{1}{32\ln 2}2^{2m(\lambda)\left(1-\frac{1}{k}\right)}$ in Equation~\eqref{eq: bound on H 2}.

\begin{align}\label{eq: bound on H 3}
    \begin{split}
        \mathrm{H}_{A B}(\psi^\lambda) \geq \log _2(\lfloor r(\lambda)\rfloor)-\nu(\lambda) & \geq \log _2(r(\lambda)-1)-\nu(\lambda) \\
& =\log _2\left(\frac{1}{32\ln2}2^{2m(\lambda)\left(1-\frac{1}{k}\right)}-1\right)-\nu(\lambda) \\
& =2m(\lambda)\left(1-\frac{1}{k}\right) - \log_2(32\ln2) - \tilde\nu(\lambda),
    \end{split}
\end{align}
where $\tilde\nu \in \operatorname{negl}(\lambda)$.
\paragraph{\textbf{Bounding $d(\lambda)$}.}
    Finally, we use the bound on $d(\lambda)$ obtained in the proof of Proposition~\ref{lemma: E_D of S lambda}. For this, we plug Equation~\eqref{eq: bound on H 3} into Equation~\eqref{eq:mlupperbound} of Proposition~\ref{lemma: E_D of S lambda} to obtain that, for large enough $\lambda$, the following holds 
    \begin{align}
        \begin{split}
            d(\lambda) &\leq \frac{1}{1-\sqrt{\varepsilon(\lambda)}}\left( m(\lambda) - \frac{1}{2}\operatorname{H}_{AB}(\psi^\lambda) + g_2\left(\sqrt{\varepsilon(\lambda)}\right) \right)\\
            &\leq \frac{1}{1-\sqrt{\varepsilon(\lambda)}}\left(m(\lambda) - m(\lambda)\left(1-\frac{1}{k}\right) +\frac{1}{2}\log_2(32\ln2)+ g_2\left(\sqrt{\varepsilon(\lambda)}\right) + \tilde\nu(\lambda)\right)\\
            & = \frac{1}{1-\sqrt{\varepsilon(\lambda)}}\left( \frac{m(\lambda)}{k} +\frac{1}{2}\log_2(32\ln2)+ g_2\left(\sqrt{\varepsilon(\lambda)}\right) + \tilde\nu(\lambda)\right),
        \end{split}
    \end{align}
    where $g_2$ is a function defined in Proposition~\ref{prop: distillation and squashed}. This concludes the proof.
\end{proof}
Theorem~\ref{lemma: bound on E_D so} obtains its entropy estimate by extracting a linearly independent subset of $M^\lambda_{\operatorname{PSO}}$. However, linear independence is not required in order to bound the Rényi-2 entropy. We may therefore use $\lceil2^{r(\lambda)}\rceil$ states distilled by the same channel and bound the von Neumann entropy of their uniform mixture. We formalize this in the following.
\begin{theorem}[Sharper bound on $\hat E^\varepsilon_D$ of $S^\lambda_{\operatorname{PSO}}$]\label{lemma: bound on E_D so max}
   Let $m:\mathbb N\rightarrow\mathbb N$ be a non-decreasing polynomially bounded function with $m\in\omega(\log_2\lambda)$, and let $k:\mathbb N\rightarrow(1,+\infty)$. Define  $\eta(\lambda):=1-2^{-m(\lambda)/k(\lambda)}$ and suppose that $m(\lambda)(1-1/k(\lambda))\in\omega(\log_2\lambda)$. Let $\varepsilon:\mathbb N\to(0,1)$ and $d:\mathbb N\to\mathbb N$, and let $S^\lambda_{\operatorname{PSO}}$ be the family of states from Definition~\ref{def: S lambda} corresponding to this choice of $\eta$. Assume that $\hat E^\varepsilon_D(S^\lambda_{\operatorname{PSO}})\geq d(\lambda)$. Then, for large enough $\lambda$, the following holds
   \begin{equation*}
       d(\lambda) \leq \frac{1}{1-\sqrt{\varepsilon(\lambda)}}\left( \left[1-\frac{1}{k(\lambda)}\right]m(\lambda) + g_2\left(\sqrt{\varepsilon(\lambda)}\right)  + \tilde\nu(\lambda)\right),
   \end{equation*}
   where $\tilde\nu(\lambda)$ is a negligible function.
\end{theorem}
\begin{proof}
    The proof follows directly from a modification of the proof of Theorem~\ref{lemma: bound on E_D so}. That is, instead of considering the mixture of a linearly independent set of vectors, we take a uniform mixture of  $\lceil 2^{r(\lambda)} \rceil$ distinct states selected from $M^\lambda_{\operatorname{PSO}}$, all of which are distilled by the same channel (see Equation~\eqref{eq: card mso lower bound} and Equation~\eqref{eq: psi lambda def}). The latter exists by applying Lemma~\ref{lemma:cardinality comparison LOCC and s so} for the chosen $\eta=\eta(\lambda)$ and $k=k(\lambda)$. This results in the existence of an LOCC channel $\hat\Gamma^\lambda\in\hat \Gamma^\lambda_p$ such that, for large enough $\lambda$, the following holds
    \begin{equation}
        \hat \Gamma^\lambda(\psi^\lambda)\in \mathcal B_{\varepsilon(\lambda)}(\Phi^{\otimes d(\lambda)}_{AB})\quad\text{with}\quad\psi^\lambda := \frac{1}{\lceil 2^{r(\lambda)} \rceil}\sum_{i=1}^{\lceil 2^{r(\lambda)} \rceil}\psi^\lambda_i,
    \end{equation}
    where we, similarly to Theorem~\ref{lemma: bound on E_D so}, denote the set of distilled states by $\{|\psi^\lambda_i\rangle\}_{i=1}^{\lceil 2^{r(\lambda)} \rceil}\subset M^\lambda_{\operatorname{PSO}}$ and the associated set of projectors by $\{\psi^\lambda_i\}_{i=1}^{\lceil 2^{r(\lambda)} \rceil}:=\{\ketbra{\psi_i^\lambda}{\psi_i^\lambda}\}_{i=1}^{\lceil 2^{r(\lambda)} \rceil}$.

    Next, repeating the Gram-matrix argument for ensemble $\tilde{ \mathcal E} := \{\lceil 2^{r(\lambda)} \rceil^{-1}, |\psi_i^\lambda\rangle\}_{i=1}^{\lceil 2^{r(\lambda)} \rceil}$ and the  derivation of Equation~\eqref{eq: bound on H 1}, we arrive at the following bound on the von Neumann entropy:
\begin{align}\label{eq: bound on H max}
    \begin{split}
                \operatorname H_{AB}(\psi^\lambda) = \operatorname H(G^\lambda) \geq \operatorname{H}_2(G^\lambda) &=-\log_2\left(\frac{1}{\lceil 2^{r(\lambda)} \rceil^2}\sum_{i=1}^{\lceil 2^{r(\lambda)} \rceil}\sum_{j=1}^{\lceil 2^{r(\lambda)} \rceil}\langle\psi^\lambda_i|\psi^\lambda_j\rangle^2\right)\\
                &\geq -\log_2\left(\frac{1}{\lceil 2^{r(\lambda)} \rceil} + \left[1-\frac{1}{\lceil 2^{r(\lambda)} \rceil}\right] (1-\eta(\lambda))^2\right)
                \\
                &= \frac{2m(\lambda)}{k(\lambda)} - \nu(\lambda),
    \end{split}
    \end{align}
    where $\nu(\lambda):= \log_2\left(1+\frac{2^{2m(\lambda)/k(\lambda)}-1}{\lceil 2^{r(\lambda)} \rceil}\right)\leq \log_2(1 + 2^{2m(\lambda)/k(\lambda)-r(\lambda)})$ is a negligible function because the assumption $m(\lambda)(1-1/k(\lambda))\in\omega(\log_2\lambda)$ implies that $r(\lambda)=\lambda^{\omega(1)}$ is superpolynomial, whereas $2m(\lambda)/k(\lambda)\leq 2m(\lambda)$ is polynomially bounded.

    Finally, substituting the bound in Equation~\eqref{eq: bound on H max} into Equation~\eqref{eq:mlupperbound} of Proposition~\ref{lemma: E_D of S lambda} shows that, for large enough $\lambda$, the following holds 
    \begin{equation}
         d(\lambda) \leq \frac{1}{1-\sqrt{\varepsilon(\lambda)}}\left( \left[1-\frac{1}{k(\lambda)}\right]m(\lambda) + g_2\left(\sqrt{\varepsilon(\lambda)}\right)  + \tilde\nu(\lambda)\right),
    \end{equation}
    where $\tilde\nu(\lambda):=\frac{1}{2}\nu(\lambda)$ is a negligible function. This establishes the claim of the theorem.
\end{proof}
        Theorems~\ref{lemma: bound on E_D so} and~\ref{lemma: bound on E_D so max} describe two different consequences of the projective geometry of $S^\lambda_{\operatorname{PSO}}$. Theorem~\ref{lemma: bound on E_D so} converts the spherical-packing property into a lower bound on the number of linearly independent states in the common-channel subset. Theorem~\ref{lemma: bound on E_D so max} instead uses $\lceil 2^{r(\lambda)}\rceil$ vectors of this subset and their pairwise-overlap bound to estimate the von Neumann entropy using the Gram-matrix technique. Consequently, for vanishing $\varepsilon(\lambda)$, the fraction of efficiently accessible entanglement, $d(\lambda)/m(\lambda)$, can be made vanishing as $\lambda\rightarrow+\infty$ by choosing $k=k(\lambda)$  such that $m(\lambda)\left(1-\frac{1}{k(\lambda)}\right)\in o(m(\lambda))$. For example, one can choose $k(\lambda):=\left(1-\sqrt{\frac{\log_2\lambda}{m(\lambda)}}\right)^{-1}$ and assert that, for large enough $\lambda$, the latter condition is satisfied. In general, one can take a function $k(\lambda)$ that satisfies the following:
        \begin{equation*}
          \log_2\lambda
    \in o\left(m(\lambda)\left(1-\frac1{k(\lambda)}\right)\right) \quad \text{and}\quad m(\lambda)\left(1-\frac1{k(\lambda)}\right) \in o(m(\lambda)),
        \end{equation*}
 where the first condition ensures that the circuit-counting argument applies (see Lemma~\ref{lemma:cardinality comparison LOCC and s so}), while the second yields the vanishing ratio.
        
We note that the results of Theorem ~\ref{lemma: bound on E_D so max}, Theorem~\ref{lemma: bound on E_D so} and Proposition~\ref{lemma: E_D of S lambda} implicitly provide a proof of existential separations between computational and information-theoretic one-shot distillable entanglement. That is, states in both $S^\lambda_{\operatorname{P\mathbb U}}$ and $S^\lambda_{\operatorname{PSO}}$ can be inefficiently transformed back to maximally entangled states with an appropriate choice of local unitaries. Therefore, their one-shot distillable entanglement is equal to that of the corresponding maximally entangled states, while, as was shown in the result of this section, the value of the computational measure is lower. We make this statement precise in the following (see Subsection~\ref{subsec: invariance local unitaries}).

 An unconditional -- that is, not dependent on computational assumptions --  separation was first considered in~\cite[Lemma~4.1]{arnon2023computational} through a counting argument. Subsequently, computational distillation and dilution rates were characterized in~\cite{leone2025entanglement} for the complementary setting of many-copy i.i.d. pure states manipulated by efficient LOCC protocols with nonzero failure probability. In Theorem~\ref{lemma: bound on E_D so}-\ref{lemma: bound on E_D so max} and Propositions~\ref{lemma: E_D of S lambda}–\ref{prop: seq of S_lambda}, we provide, to the best of our knowledge, the first complete quantitative proof of an unconditional separation in the deterministic one-shot setting. Our proof builds on the counting intuition of~\cite[Lemma~4.1]{arnon2023computational}, while introducing a new quantitative connection between the geometry of projective-net-generated families and their computational one-shot distillable entanglement. Moreover, the proof of Theorem~\ref{lemma: bound on E_D so} illuminates a potentially interesting connection between projective  $\eta$-nets
 and the geometry of the generated family of states that allows us to control the number of linearly independent distillable states.
 
Finally, we discuss the practical implications of the main result of this section for the field of quantum networks. While the families of states $S^\lambda_{\operatorname{P\mathbb U}}$ and $S^\lambda_{\operatorname{PSO}}$ may not represent any model of a particular quantum-network architecture, they illustrate a phenomenon that becomes relevant when network nodes must manipulate states drawn from large families. As discussed above, every state in such families is maximally entangled from the information-theoretic perspective. However, our results show that the entanglement content of these states may not be efficiently accessible. In quantum-network stacks, delivered entangled states are subsequently processed using local quantum operations and classical control~\cite{PDT22,DVD22}. Therefore, information-theoretic entanglement measures alone may not be sufficient to characterize the usefulness of a state when the computational resources available to the network nodes are limited.
 
 \section{Transforming States: Unitary Invariance and LOCC Monotonicity}\label{sec: invariance local unitaries}
In this section, we examine how computational constraints affect two classical desiderata for entanglement measures in a one-shot setting: LU invariance and LOCC monotonicity. The latter result will be crucial in showing a ``thermodynamic'' property of computational entanglement in Section~\ref{sec: val-cont}. We proceed in two steps:

In Subsection~\ref{subsec: invariance local unitaries}, we show that, in both the one-shot and the uniform settings, the computational distillable entanglement and the computational entanglement cost are LU-invariant under efficient families of local unitaries, whereas this invariance can fail for inefficient families. In Subsection~\ref{subsec: LOCC monotnicity}, we show, as a corollary of LU invariance, that these computational entanglement measures are LOCC-monotone under efficient LOCC protocols, whereas monotonicity can fail for arbitrary families of LOCC channels, in both one-shot and uniform settings. Thus, unconstrained local operations can increase or decrease computational entanglement, while their efficient counterparts preserve the familiar monotonic behavior.

\subsection{Invariance Under Local Unitaries}\label{subsec: invariance local unitaries}
In this subsection, we analyze the invariance of each computational entanglement measure under local unitaries. The constructions underlying~\cite[Lemma~4.1]{arnon2023computational} suggest that computational entanglement may depend on the complexity of the local basis. Related basis-complexity phenomena also arise in the many-copy pure-state setting of~\cite{leone2025entanglement}. Here, we formulate and prove the corresponding LU-invariance statements for the computational one-shot measures of~\cite{arnon2023computational}. The discussion proceeds as follows:

We first provide the formal definition of efficient families of unitaries. 
After this, we show, in Proposition~\ref{th:comp one shot cost inva uni}, that the computational one-shot entanglement cost is invariant under efficient families of local unitaries, but is not necessarily invariant under arbitrary local unitary families under the error condition stated below. We then show, in Proposition~\ref{th: local dist unita invar epsilon diff zero}, similar results for the computational one-shot distillable entanglement, concerning invariance under arbitrary and efficient families of local unitaries. Corollaries~\ref{corr: uniform cost inv} and~\ref{corr: unif distil inva} extend all results to the uniform (keyed) scenario in a straightforward component-wise manner.

First, we adapt the definition of efficient LOCC maps from Definition \ref{def: definition efficient LOCC} to \emph {efficient families of unitaries} as follows. 

\begin{definition}[Efficient local unitaries]\label{def: efficient local unitaries}
Let $\{U^{\lambda}_{AB}\}_{\lambda\in\mathbb N}$ be a family of bipartite local unitaries and let $\{\tilde U^{\lambda}\}_{\lambda\in\mathbb N}$ be the associated family of unitary channels, where
\[
U^{\lambda}_{AB}:=U^{\lambda}_{A}\otimes V^{\lambda}_{B},\qquad
\tilde U^{\lambda}_{AB}(\rho):=U^{\lambda}_{AB}\rho\,U^{\lambda\dagger}_{AB}.
\]
We say the family of local unitaries $\{U^{\lambda}_{AB}\}_{\lambda\in\mathbb N}$ is \emph{efficient} if every $U^{\lambda}_{AB}$ admits a quantum-circuit description containing at most $\operatorname{poly}(\lambda)$ gates.
\end{definition}
 
 \begin{definition}
     Let $\{k,U^{k}_{AB}\}_{k\in\{0,1\}^{\kappa(\lambda)}}$ be a family of encoded local unitaries and let $\{k,\tilde{U}^k_{AB}\}_{k\in\{0,1\}^{\kappa(\lambda)}}$ be the associated family of unitary channels acting on the encoded family of states $\{k,\rho^k_{AB}\}_{k\in\{0,1\}^{\kappa(\lambda)}}$:
     \begin{align*}
         \forall k\in\{0,1\}^{\kappa(\lambda)},\quad\tilde U^k_{AB}:\mathcal D(\mathcal H_{AB})&\to\mathcal D(\mathcal H_{AB})\\
         \tilde U^k_{AB}:\rho^k_{AB}&\to\tilde U^k_{AB}(\rho^k_{AB})=U^k_{AB}\,\rho^k_{AB}\,U^{k\dagger}_{AB},
     \end{align*}
     where for all $k\in\{0,1\}^{\kappa(\lambda)}$ we have $ U^k_{AB}:=U^k_A\otimes V^k_B$.
    
    We say the family of encoded unitaries $\{k,U^k_{AB}\}_{k\in\{0,1\}^{\kappa (\lambda)}}$ is efficient if there exists a family of polynomial-size quantum circuits $\{\mathcal C^\lambda\}_{\lambda\in\mathbb N}$ such that, for every $\lambda$, given any key $k\in\{0,1\}^{\kappa(\lambda)}$, $\mathcal C^\lambda$ implements $U^k$ and admits a circuit representation that has at most $c(\lambda)$ gates for some $c\in \operatorname{Poly}(\lambda)$.
 \end{definition}
 
We now show that the value of computational one-shot entanglement cost $\hat E^\varepsilon_C$ remains the same if one applies an efficient family of local unitaries to the original family of states. However, if one applies a general family of local unitaries, the upper bound on $\hat{E}^\varepsilon_C$ can change. This establishes a precise LU-invariance statement for computational one-shot entanglement cost and quantifies when invariance can fail under inefficient families of local unitaries. 
\begin{proposition}[LU Invariance of $\hat E_C^\varepsilon$]\label{th:comp one shot cost inva uni} 
Let $\{\rho^{\lambda}_{AB}\}_{\lambda \in \mathbb N}$ be a family of bipartite states, and let $\{U^\lambda_{AB}\}_{\lambda \in \mathbb N}$ be an efficient family of local unitaries. Then the computational one-shot entanglement cost is invariant under such unitary transformations:
\begin{equation*}
    \hat E^{\epsilon}_C(\{\rho^{\lambda}_{AB}\}) \leq n(\lambda) 
    \iff 
    \hat E^{\epsilon}_C(\{U^{\lambda}_{AB}\rho^{\lambda}_{AB}U^{\lambda\dagger}_{AB}\}) \leq n(\lambda).
\end{equation*}
In contrast, for general (not necessarily efficient) families of local unitaries, where $n\in\omega(\log_2\lambda)$ is a non-decreasing polynomially bounded function and $\varepsilon:\mathbb N\rightarrow[0,1)$ is the dilution error such that $\limsup_{\lambda\rightarrow+\infty}2\sqrt{\varepsilon(\lambda)}<\frac{c_0^2}{2}$ with $c_0>0$ a universal constant in Equation~\eqref{eq: lower bound s lambda u}, this invariance does not always hold. 
\end{proposition}
\begin{proof}
We first prove the positive direction: invariance under efficient local unitaries. Assume that $\hat E^{\epsilon}_C(\{\rho^{\lambda}_{AB}\}) \leq n(\lambda)$. Consider the transformed family $\{\rho'^\lambda_{AB} := U^\lambda_{AB} \rho^\lambda_{AB} U^{\lambda\dagger}_{AB}\}$. Using the unitary invariance of fidelity (Proposition \ref{prop: unitary inv fidelity}) and the efficient implementability of $U^\lambda_{AB}$, one can define an efficient LOCC family $\{\Gamma'^\lambda\}$ by conjugation:
\[
\Gamma'^\lambda(\cdot) := U^\lambda_{AB} \hat{\Gamma}^\lambda(\cdot) U^{\lambda\dagger}_{AB}.
\]
Then, for each $\lambda$ we have:
\[
F\left( \Gamma'^\lambda(\Phi^{\otimes n(\lambda)}), \rho'^\lambda_{AB} \right) = F\left( \hat{\Gamma}^\lambda(\Phi^{\otimes n(\lambda)}), \rho^\lambda_{AB} \right) \geq 1 - \epsilon(\lambda),
\]
showing that the transformed family is also dilutable at cost $n(\lambda)$ via an efficient family of LOCC channels. The converse direction follows by symmetry, proving the equivalence.

We now show by contradiction that this invariance fails for arbitrary (not necessarily efficient) families of local unitaries. Assume that $S^\lambda_{\operatorname{P\mathbb U}}$ is $n(\lambda)$ dilutable, i.e., $\hat E^{\epsilon}_C(S^\lambda_{\operatorname{P\mathbb U}}) \leq n(\lambda)$. From the construction of $S^{\lambda}_{\operatorname{P\mathbb U}}$ we have that for two distinct states from $S^{\lambda}_{\operatorname{P\mathbb U}}$ it holds that  $|\langle\psi^\lambda_a|\psi^\lambda_b\rangle|\leq 1-\eta$ with $a\neq b$.

Define the set of states around a given state $\psi_i\in S^{\lambda}$ by $\mathcal{B}_{\varepsilon(\lambda)}(\psi^\lambda_i)$ where $\mathcal{B}_\varepsilon (\rho_{AB})$ is given by: 
\begin{equation*}
    \mathcal{B}_{\varepsilon}(\rho_{AB}) :=
    \left\{ \sigma_{AB}\in\mathcal D(\mathcal{H}_{AB}) \,|\, F\left(\sigma_{AB},\rho_{AB}\right)\geq1-\varepsilon\right\}.
\end{equation*}
  
From Lemma \ref{lemma:cardinality comparison LOCC}, the total number of states in $S^{\lambda}_{\operatorname{P\mathbb U}}$ is exponentially larger than the total number of efficient LOCC channels with circuit size bounded by $p\in\operatorname{Poly}(\lambda)$ acting on $2n(\lambda)$ qubits and outputting $n_A(\lambda) + n_B(\lambda)$ qubits:
    \begin{equation*}
            \frac{|\hat{\Gamma}_p^\lambda|}{|S_{\operatorname{P\mathbb U}}^\lambda|} \leq \left(\frac{c_0^2}{2\eta}\right)^{-\Omega(2^{2n(\lambda)})},
        \end{equation*}
         where $c_0>0$ is a universal constant in Equation~\eqref{eq: lower bound s lambda u}.
Therefore, there exists an efficient LOCC channel $\hat{\Gamma}^\lambda$ which dilutes two distinct pure states $\psi_a^\lambda$ and $\psi_b^\lambda$: 
\begin{equation*}
a\neq b,\quad\hat{\Gamma}^{\lambda}(\Phi_{AB}^{\otimes n(\lambda)})=:\sigma^\lambda \in \mathcal{B}_{\epsilon(\lambda)}(\psi^{\lambda}_a)\cap\mathcal{B}_{\epsilon(\lambda)}(\psi^{\lambda}_b).
\end{equation*}
Applying the triangle inequality and the Fuchs–van de Graaf inequality in Proposition~\ref{prop: fidelity and one norm}, we obtain:
    \begin{equation*}
        \frac{1}{2}\left\|\psi^{\lambda}_a -\psi^{\lambda}_b\right\|_1\leq \frac{1}{2}\left\| \psi^{\lambda}_a-\sigma^{\lambda}\right\|_1+\frac{1}{2}\left\| \psi^{\lambda}_b -\sigma^\lambda\right\|_1\leq2\sqrt{\varepsilon(\lambda)},
    \end{equation*}
    where we have used that for $i\in\{a,b\}$ 
\begin{equation*}
            \frac{1}{2}\left\|\psi_i^{\lambda}- \sigma^\lambda\right\|_1\leq \sqrt{1-F(\sigma^\lambda, \psi_i^{\lambda})} \leq \sqrt{\varepsilon(\lambda)}.
\end{equation*}
On the other hand, the distance between $\psi^{\lambda}_a$ and $\psi^{\lambda}_b$ is lower bounded by construction: 
    \begin{equation*}
        \frac{1}{2}\left\|\psi^{\lambda}_a- \psi^{\lambda}_b\right\|_1\geq1-\sqrt{F\left(\psi^{\lambda}_a,\psi^{\lambda}_b\right)}\geq\eta.
    \end{equation*}
    By combining the upper and lower bounds, we obtain
    \begin{equation*}
        \eta\leq\frac{1}{2}\left\|\psi_a^{\lambda}-\psi_b^{\lambda}\right\|_1\leq2\sqrt{\epsilon(\lambda)},
    \end{equation*}
    which holds for all $\eta$ and $\epsilon$ in $(0,1)$ such that $\eta\leq2\sqrt{\varepsilon(\lambda)}$. Choosing $\eta\in(0,\frac{c_0^2}{2})$ to be greater than $\limsup_{\lambda\rightarrow+\infty}2\sqrt{\varepsilon(\lambda)}$ contradicts the assumption that $S_{\operatorname{P\mathbb U}}^\lambda$ is dilutable with precision $\varepsilon(\lambda)$ in the sense of Remark~\ref{remark: pointwise notation S lambda}.

    Now, repeating the diagonal argument used in the proof of Proposition~\ref{prop: seq of S_lambda}, with dilution channels in place of distillation channels, we obtain a family of states $\{\rho^\lambda_{AB}\}_{\lambda\in\mathbb N}$, where $\rho^\lambda_{AB}\in S^\lambda_{\operatorname{P\mathbb U}}$ for every $\lambda$, such that $\{\rho^\lambda_{AB}\}_{\lambda\in\mathbb N}$ is not dilutable with $n(\lambda)$ EPR pairs. Indeed, one applies the argument to the sequence of polynomial circuit-size bounds \(p_j(\lambda)=j\lambda^j\). For every \(j\) and all sufficiently large \(\lambda\), the counting argument above provides a state in \(S^\lambda_{\operatorname{P\mathbb U}}\) that cannot be prepared from \(\Phi^{\otimes n(\lambda)}\) with error \(\varepsilon(\lambda)\) by any LOCC channel of circuit size at most \(p_j(\lambda)\). Selecting such states along a strictly increasing sequence of values of \(\lambda\), as in Proposition~\ref{prop: seq of S_lambda}, produces a single family that cannot be diluted by any polynomially bounded family of LOCC channels.
    
    Finally, for every $\rho^\lambda_{AB}\in\{\rho^\lambda_{AB}\}_{\lambda\in\mathbb N}$, by construction of $S^\lambda_{\operatorname{P\mathbb U}}$ (Appendix~\ref{app: cons s lambda}), it holds that
    \begin{equation}
        \rho^\lambda_{AB} = (\mathbb I\otimes U^\lambda_\rho) \Phi^{\otimes n(\lambda)}
(\mathbb I\otimes U^{\lambda\dagger}_\rho),
    \end{equation}
    for some $U^\lambda_{\rho}\in\mathbb U(2^{n(\lambda)})$. However, $\hat E^\varepsilon_C(\{\Phi^{\otimes n(\lambda)}\}_\lambda)\leq n(\lambda)$, which is achieved by the following family of dilution channels $\{\mathbb I^\lambda\}_{\lambda\in\mathbb N}$. Therefore, computational one-shot entanglement cost is not invariant under arbitrary families of local unitaries. Combined with the positive direction, this establishes the proof of the proposition. 
\end{proof}
\begin{corollary}\label{corr: uniform cost inv}
    Let the family of encoded states be $\{k,\rho^k_{AB}\}_{k\in\{0,1\}^{\kappa(\lambda)}}$. Consider an efficient family of encoded local unitaries $\{k,U^k_{AB}:=U^k_A\otimes U^k_B\}_{k\in\{0,1\}^{\kappa(\lambda)}}$. Then the uniform entanglement cost remains invariant: 
    \begin{equation*}
        \hat E^{\epsilon}_C(\{k,\rho^k_{AB}\})\leq n(\lambda)\iff\hat E^{\epsilon}_C(\{k,U^k_{AB}\,\rho^{k}_{AB}\,U^{k\dagger}_{AB}\})\leq n(\lambda),
    \end{equation*}
    where $n(\lambda)$ is a valid upper bound on $\hat{E}_C^\varepsilon(\{k,\rho^k_{AB}\})$. For arbitrary, not necessarily efficiently implementable, families of local unitaries  where $n:\mathbb N\rightarrow\mathbb N$ is a non-decreasing polynomially bounded function satisfying $n\in\omega(\log_2\lambda)$, and where $\varepsilon:\mathbb N\rightarrow[0,1)$ is such that $\limsup_{\lambda\rightarrow+\infty}2\sqrt{\varepsilon(\lambda)}<\frac{c_0^2}{2}$ with $c_0>0$ a universal constant in Equation~\eqref{eq: lower bound s lambda u}, the uniform computational entanglement cost is not always invariant.
\end{corollary}

In the following, we argue that computational distillable
entanglement is invariant under the action of efficient families of local unitaries. However, it can
change under the action of non-efficient families of local unitaries. 
We prove the lack of invariance under general local unitaries by using an upper bound on the one-shot computational distillable entanglement of the family $S^\lambda_{\operatorname{P\mathbb U}}$ (see Proposition~\ref{lemma: E_D of S lambda}). 
Note that one can use the family $S^\lambda_{\operatorname{PSO}}$, since Proposition~\ref{prop: seq of S_lambda} extends to the case of the family $S^\lambda_{\operatorname{PSO}}$ with the upper bound in Proposition~\ref{lemma: E_D of S lambda} replaced by the upper bound in Theorem~\ref{lemma: bound on E_D so} or Theorem~\ref{lemma: bound on E_D so max}. The proof is unchanged because its diagonal argument uses only the pointwise circuit-size obstruction established in the proof of Theorem~\ref{lemma: bound on E_D so}.

\begin{proposition}[LU Invariance of $\hat E^\varepsilon_D$]\label{th: local dist unita invar epsilon diff zero}
Let $\{\rho^{\lambda}_{AB}\}_{\lambda\in\mathbb N}$ be a family of bipartite states, and let $\{U^{\lambda}_{AB}\}$ be a family of local unitaries. 
Then the following holds if the family $\{U^{\lambda}_{AB}\}_{\lambda\in\mathbb N}$ is an efficient family of local unitaries:
     \begin{equation*}
         \hat{E}^{\epsilon}_D(\{\rho^{\lambda}_{AB}\})\geq m(\lambda)\iff\hat{E}^{\epsilon}_D(\{U^{\lambda}_{AB}\rho^{\lambda}_{AB}U^{\lambda\dagger}_{AB}\})\geq m(\lambda).
     \end{equation*}
      If $\epsilon:\mathbb N\rightarrow[0,1)$ is such that $m(\lambda)\sqrt{\varepsilon(\lambda)}\rightarrow0$ as $ \lambda\rightarrow+\infty$  and $m\in\omega(\log_2\lambda)$ is a non-decreasing polynomially bounded function, then for arbitrary families of local unitaries the  computational one-shot distillable entanglement is not always invariant.
 \end{proposition} 
 \begin{proof}
For efficient families of local unitaries, the proof follows by the same composition argument as the positive part of Proposition~\ref{th:comp one shot cost inva uni}: a distillation protocol for either family can be precomposed with the corresponding efficient inverse unitary channel to obtain a protocol for the other family. We therefore focus on the failure of invariance for arbitrary families of local unitaries.
 
  For arbitrary, not necessarily efficiently implementable, families of local unitaries, we provide a counterexample. Applying Proposition~\ref{prop: seq of S_lambda} with $d(\lambda)=m(\lambda)$ and $\eta\in (0,\frac{c_0^2}{2})$, which we specify in the following, we obtain a family of states $ \{\rho^\lambda_{AB}\}_{\lambda\in\mathbb N}$ with $\rho^\lambda_{AB}\in S^\lambda_{\operatorname{P\mathbb U}}$
for every $\lambda\in\mathbb N$, with the following property: if $\hat E_D^\varepsilon(\{\rho^\lambda_{AB}\}_{\lambda})\geq m(\lambda)$ then, for all sufficiently large $\lambda$, the following inequality holds:
     \begin{equation*}
     m(\lambda)\leq \frac{1}{1-\sqrt{\varepsilon(\lambda)}}\left(m(\lambda) -\frac{1}{2}H^*(AB)[\eta]+g_2\left(\sqrt{\varepsilon(\lambda)}\right)\right),
     \end{equation*}
where, with the right choice of $\eta$, the right-hand side of the inequality above can be strictly smaller than $m(\lambda)$.
Explicitly, it suffices to choose $\eta\in(0,\frac{c_0^2}{2})$ such that the following inequality holds:
     \begin{equation}\label{eq: bound on H lack of inv}
          H^*(AB)[\eta] > 2\left(g_2\left(\sqrt{\varepsilon(\lambda)}\right) + m(\lambda)\sqrt{\varepsilon(\lambda)}\right).
     \end{equation}
From the definition of $g_2(\delta)$ one can see that $g_2(\delta) \rightarrow0,\,\delta\rightarrow 0$. Therefore, since by assumption, $m(\lambda)\sqrt{\varepsilon(\lambda)}\rightarrow0,\,\lambda\rightarrow+\infty$ and $m\in\omega(\log_2\lambda)$, one can pick $\eta\in(0,\frac{c_0^2}{2})$ such that the above inequality is satisfied.

Now, recall that, for any $\lambda\in\mathbb N$, the state $\rho^\lambda_{AB}\in\{\rho^\lambda_{AB}\}_{\lambda\in\mathbb N}$ has the following form
\begin{equation}
    \rho^\lambda_{AB} = (\mathbb I\otimes U^\lambda_\rho)\Phi^{\otimes m(\lambda)}(\mathbb I\otimes U_\rho^{\lambda\dagger}),
\end{equation}
for some unitary $U^\lambda_\rho\in\mathbb U(2^{m(\lambda)})$. Therefore, the family of states $\{\rho^\lambda_{AB}\}_{\lambda\in\mathbb N}$ is obtained by applying unitary transformations to states in the family $\{\Phi^{\otimes m(\lambda)}\}_{\lambda\in\mathbb N}$. The latter family is efficiently distillable by the family of LOCC channels $\{\mathbb I^\lambda\}_{\lambda\in\mathbb N}$ with a valid lower bound $m(\lambda)$. However, by computation in Equation~\eqref{eq: bound on H lack of inv}, the family   $\{\rho^\lambda_{AB}\}_{\lambda\in\mathbb N}$ is not distillable to $m(\lambda)$ EPR pairs with the given distillation parameters. This establishes the lack of invariance under general local unitary transformations.
This completes the proof.
 \end{proof}

 \begin{corollary}\label{corr: unif distil inva}
         Let $\{k,\rho^k_{AB}\}_{k\in\{0,1\}^{\kappa(\lambda)}}$ be a family of encoded states. Suppose that $m(\lambda)$ is a valid lower bound on $\hat{E}_D^\varepsilon(\{\rho^k_{AB}\})$. Consider the efficient family of encoded local unitaries $\{k,U^k_{AB}\}_{k\in\{0,1\}^{\kappa(\lambda)}}$. Then the uniform distillable entanglement remains invariant: 
    \begin{equation*}
        \hat E^{\epsilon}_D(\{k,\rho^k_{AB}\})\geq m(\lambda)\iff\hat E^{\epsilon}_D(\{k,U^k_{AB}\,\rho^{k}_{AB}\,U^{k\dagger}_{AB}\})\geq m(\lambda).
    \end{equation*}
    If $\epsilon:\mathbb N\rightarrow[0,1)$ is such that $m(\lambda)\sqrt{\varepsilon(\lambda)}\rightarrow0, \lambda\rightarrow+\infty$ and $m:\mathbb N\rightarrow\mathbb N$ is a non-decreasing polynomially bounded function satisfying $m\in\omega(\log_2\lambda)$, then for  arbitrary families of local unitaries the uniform computational distillation is not always invariant.
\end{corollary}
 
\subsection{LOCC Monotonicity}\label{subsec: LOCC monotnicity}
In this subsection, we analyze the monotonicity properties of computational entanglement measures under LOCC channels in the sense of Definition \ref{def: upper bound LOCC monotonicity}.
\begin{proposition}[LOCC Monotonicity of $\hat E^\varepsilon_C$]\label{thm: locc not monotone comp ent cost}
    Let $\{\rho_{AB}^\lambda\}_{\lambda\in\mathbb N}$ be a family of states, and let $\{\Lambda^\lambda\}_{\lambda\in\mathbb N}$ be an efficient family of LOCC channels. Then the following holds
    \begin{equation*}
       \hat E^{\varepsilon}_C(\{\rho_{AB}^\lambda\}) \leq n(\lambda)\implies\hat E^{\varepsilon}_C(\{\Lambda^\lambda(\rho^\lambda_{AB})\}) \leq n(\lambda).
    \end{equation*}
    However, if $n\in\omega(\log_2\lambda)$ is a non-decreasing polynomially bounded function and $\varepsilon:\mathbb N\rightarrow[0,1)$ is such that $\limsup_{\lambda\rightarrow+\infty}2\sqrt{\varepsilon(\lambda)}<\frac{c_0^2}{2}$ with $c_0>0$ a universal constant in Equation~\eqref{eq: lower bound s lambda u}, there exists a family of non-efficient LOCC channels and a family of bipartite states for which LOCC monotonicity breaks. 
\end{proposition}
\begin{proof} 

For efficient families of LOCC channels, consider the efficient family of LOCC maps given by $\{\Lambda^{\lambda}\}$, and define the family of LOCC maps $\{\tilde \Gamma^\lambda\}_{\lambda\in\mathbb N}:=\{ \Lambda^\lambda\circ \hat \Gamma^\lambda\}_{\lambda\in\mathbb N}$, where $\{\hat \Gamma^\lambda\}_{\lambda\in\mathbb N}$ is the family that achieves $\hat E^\varepsilon_C(\{\rho_{AB}^\lambda\}) \leq n(\lambda)$. The family $\{\tilde{\Gamma}^{\lambda}\}_{\lambda\in\mathbb N}$ is efficient since $\Lambda^\lambda\circ \hat \Gamma^\lambda$ admits an efficient circuit representation for all values of $\lambda\in\mathbb N$. 

Using the family $\{\tilde{\Gamma}^\lambda\}_{\lambda\in\mathbb N}$ we obtain the following chain of inequalities:
\begin{align*}
    \forall\,\lambda\in\mathbb N,\quad p_{\text{err}}(\tilde{\Gamma}^{\lambda}, \Lambda^\lambda(\rho_{AB}^\lambda)) &= 1 - F(\tilde{\Gamma}^{\lambda}(\Phi^{\otimes n(\lambda)}), \Lambda^\lambda(\rho_{AB}^\lambda)) \\
    &=1 - F(\Lambda^\lambda\circ \hat \Gamma^\lambda(\Phi^{\otimes n(\lambda)}), \Lambda^\lambda(\rho_{AB}^\lambda)) \\
    &\leq 1- F(\hat\Gamma^\lambda(\Phi^{\otimes n(\lambda)}), \rho^\lambda_{AB})=:p_{\text{err}}(\hat{\Gamma}^{\lambda},\rho^{\lambda}_{AB}),
\end{align*}
    
where the last inequality is obtained by using the data-processing inequality from Proposition \ref{prop: data-processing fidelity},
\begin{equation*}
    F(\Lambda^\lambda\circ \hat \Gamma^\lambda(\Phi^{\otimes n(\lambda)}), \Lambda^\lambda(\rho_{AB}^\lambda)) \geq F(\hat\Gamma^\lambda(\Phi^{\otimes n(\lambda)}), \rho^\lambda_{AB}).
\end{equation*}

From the assumption of needing at most $n(\lambda)$ EPR pairs to recover the family $\{\rho^{\lambda}_{AB}\}_{\lambda\in\mathbb N}$, we get:
\begin{equation*}
    p_{\text{err}}(\tilde{\Gamma}^{\lambda}, \Lambda^\lambda(\rho_{AB}^\lambda))\leq p_{\text{err}}(\hat{\Gamma}^{\lambda},\rho^{\lambda}_{AB})\leq \epsilon(\lambda),
\end{equation*}
therefore, by Definition~\ref{def: comp entang cost vidick}, we have 
\begin{equation*}
    \hat E^\varepsilon_C(\{\rho^\lambda_{AB}\}) \leq n(\lambda)\implies \hat E^\varepsilon_C(\{\Lambda^\lambda(\rho^\lambda_{AB})\}) \leq n(\lambda).
\end{equation*}

The result for the non-efficient families of LOCC follows as a simple corollary of Proposition~\ref{th:comp one shot cost inva uni}, since it is essentially a reformulation of the said proposition.

\end{proof}
\begin{corollary}\label{cor: unif locc monotone cost}
     Let $\{k,\rho^k_{AB}\}_{k\in\{0,1\}^{\kappa(\lambda)}}$ be a family of encoded states, and let $\{k,\Lambda^k\}_{k\in\{0,1\}^{\kappa(\lambda)}}$ be an efficient family of encoded  LOCC channels. 
     Then, 
    \begin{equation*}
     \hat E^\varepsilon_C(\{k,\rho^k_{AB}\}) \leq n(\lambda)\implies  \hat E_C^\varepsilon(\{k,\Lambda^k(\rho^k_{AB})\}) \leq n(\lambda) .
    \end{equation*}
    However, if $n\in\omega(\log_2\lambda)$ is a non-decreasing polynomially bounded function and $\varepsilon:\mathbb N\rightarrow[0,1)$ is such that $\limsup_{\lambda\rightarrow+\infty}2\sqrt{\varepsilon(\lambda)}<\frac{c_0^2}{2}$ with $c_0>0$ a universal constant in Equation~\eqref{eq: lower bound s lambda u}, the same relation does not always hold for non-efficient families of encoded LOCC channels. 
\end{corollary}

We conclude this section by discussing the LOCC monotonicity property for the non-asymptotic computational distillable entanglement. 

\begin{proposition}[LOCC Monotonicity of $\hat E^\varepsilon_D$]\label{thm: locc not monotone comp dist ent}
 Let $\{\rho_{AB}^\lambda\}_{\lambda\in\mathbb N}$ be a family of bipartite states, and let $\{\Lambda^\lambda\}_{\lambda\in\mathbb N}$ be an efficient family of LOCC channels. Then,
    \begin{equation*}
     \hat E_D^\varepsilon(\{\Lambda^\lambda(\rho^\lambda_{AB})\}) \geq m(\lambda) \implies \hat E^\varepsilon_D(\{\rho^\lambda_{AB}\}) \geq m(\lambda).
    \end{equation*}
   However, if  $\epsilon:\mathbb N\rightarrow[0,1)$ is such that $m(\lambda)\sqrt{\varepsilon(\lambda)}\rightarrow0, \lambda\rightarrow+\infty$ and $m\in\omega(\log_2\lambda)$ is a non-decreasing polynomially bounded function, this does not always hold for arbitrary, not necessarily efficiently implementable, families of LOCC channels.
    
\end{proposition}

\begin{proof}

    For the efficient family of LOCC channels, we use arguments similar to those used in Proposition~\ref{thm: locc not monotone comp ent cost}. Let $\{\Lambda^{\lambda}\}_{\lambda\in\mathbb N}$ be an efficient family of LOCC maps and assume that we can distill $m(\lambda)$ EPR pairs from $\{\Lambda^{\lambda}(\rho^{\lambda}_{AB})\}_{\lambda\in\mathbb N}$. Then, for $\epsilon(\lambda)\in [0,1)$, the following holds
    \begin{equation*}
        \forall\, \lambda\in\mathbb N,\quad p_{\text{err}}(\hat \Gamma^\lambda, \Lambda^\lambda(\rho^\lambda_{AB}))=1-F(\hat{\Gamma}^{\lambda}\circ\Lambda^{\lambda}(\rho^{\lambda}),\Phi^{\otimes m(\lambda)})=:p_{\text{err}}( \tilde{\Gamma}^\lambda, \rho^\lambda_{AB})\leq \varepsilon(\lambda).
    \end{equation*}
   Therefore, one can distill $m(\lambda)$ EPR pairs from $\{\rho^{\lambda}_{AB}\}_{\lambda\in\mathbb N}$ with the efficient family of LOCC maps $\{\tilde{\Gamma}^{\lambda}\}_{\lambda\in\mathbb N}:=\{\hat{\Gamma}^{\lambda}\circ\Lambda^{\lambda}\}_{\lambda\in\mathbb N}$, which proves the implication.
   
    For the non-efficient family of LOCC channels, we establish a contradiction. Assume that the implication in the proposition holds for an arbitrary inefficient family of LOCC channels.  Now, fix $\eta\in(0,\frac{c_0^2}{2})$ such that the bound in Proposition~\ref{prop: seq of S_lambda} is smaller than $m(\lambda)$ for sufficiently large $\lambda$. Then, consider the states $\rho_{AB}^{\lambda}\in \{\rho^\lambda_{AB}\}_{\lambda\in\mathbb N}$ with $\rho^\lambda_{AB}\in S^\lambda_{\operatorname{P\mathbb U}}$, which exist as a result of Proposition~\ref{prop: seq of S_lambda} applied to chosen parameters $\eta,\varepsilon(\lambda)$ and $d(\lambda)=m(\lambda)$, and the LOCC maps $\Lambda^{\lambda}$ of the form
    \begin{equation*}
        \rho^{\lambda}_{AB}=(\mathbb I\otimes U^\lambda_\rho)\Phi^{\otimes m(\lambda)}(\mathbb I\otimes U_\rho^{\lambda\dagger}),\quad \Lambda^{\lambda}(\cdot):=(\mathbb I\otimes U_\rho^{\lambda\dagger})(\cdot)(\mathbb I\otimes U^\lambda_\rho),
    \end{equation*}
    for every value of $\lambda\in\mathbb N$.
    One can distill $\{\Lambda^{\lambda}(\rho^{\lambda}_{AB})\}_{\lambda\in\mathbb N}=\{\Phi^{\otimes m(\lambda)}\}_{\lambda\in\mathbb N}$ with the use of the trivial LOCC channel $\{\Gamma^{\lambda}\}_{\lambda\in\mathbb N}=\{\mathbb I^\lambda\}_{\lambda\in\mathbb N}$. By the contradiction hypothesis, this implies that $\hat E^\varepsilon_D(\{ \rho^{\lambda}_{AB}\}_\lambda)\geq m(\lambda)$. However, Proposition~\ref{prop: seq of S_lambda}, applied to $d(\lambda)=m(\lambda)$, then implies that, for all sufficiently large $\lambda$, the following holds
\begin{equation}\label{eq: locc 1}
    m(\lambda)\leq B_\eta(\lambda),
\end{equation}
where $ B_\eta(\lambda)$ denotes the bound in Proposition~\ref{prop: seq of S_lambda}. However, by the choice of $\eta$, for all sufficiently large $\lambda$, it also holds that
\begin{equation}\label{eq: locc 2}
    B_\eta(\lambda)<m(\lambda).
\end{equation}
 Combining Equation~\eqref{eq: locc 1} and  Equation~\eqref{eq: locc 2} yields the contradiction. 
\end{proof}

\begin{corollary}\label{cor: unifrom locc monotone dist}
     Let the family of encoded states be $\{k,\rho^k_{AB}\}_{k\in\{0,1\}^{\kappa(\lambda)}}$ and let $\{k,  \Lambda^k\}_{k\in\{0,1\}^{\kappa(\lambda)}}$ be an efficient family of encoded  LOCC channels.  
     Then,
    \begin{equation*}
     \hat E_D^\varepsilon(\{k,\Lambda^k(\rho^k_{AB})\}) \geq m(\lambda) \implies \hat E^\varepsilon_D(\{k,\rho^k_{AB}\}) \geq m(\lambda).
    \end{equation*}
    However, if  $\varepsilon:\mathbb N\rightarrow[0,1)$ is such that $m(\lambda)\sqrt{\varepsilon(\lambda)}\rightarrow0, \lambda\rightarrow+\infty$ and $m\in\omega(\log_2\lambda)$ is a non-decreasing polynomially bounded function, the same relation does not always hold for the non-efficient families of encoded LOCC channels. 
\end{corollary}

We now provide an interpretation of the results of this section in the context of entanglement theory. In conventional entanglement theory, LOCC channels represent free transformations. Therefore, entanglement measures are expected to be monotone under LOCC and invariant under local unitaries~\cite{horodecki2009quantum,CG19}. In the computational framework, the permitted transformations are instead restricted to efficiently implementable families of LOCC channels. Our results show that the familiar monotonicity and invariance properties continue to hold under efficient transformations but need not hold under arbitrary families of LOCC channels or local unitaries. This distinction is relevant to quantum networks, where local gates, measurements, classical processing, and communication must be implemented by finite programs executed at the network nodes~\cite{DVD22}. Hence, an LOCC transformation that is permitted information-theoretically may not be available under the computational restrictions of a practical network.

\section{Action of Noise and Second Law of Efficient Entanglement Manipulation}\label{sec: val-cont}
In this section, we quantitatively show how the value of one-shot computational distillable entanglement changes under small perturbations of the original family of states. We refer to this property as a ``continuity bound'' (see Definition~\ref{def: continuity}). As an application of the obtained continuity bound, we then derive a general relation between non-asymptotic computational distillable entanglement and cost, which we term the ``Second Law of Efficient Entanglement Manipulation''. We recall that the results of this section are obtained using families of distillation protocols with nonzero failure probability, which can be subsequently mapped to deterministic protocols (see Remark~\ref{remark: SLOCC}).

In Subsection~\ref{subsec: value cont comp}, we show the continuity bound for computational one-shot distillable entanglement $\hat{E}_D^\varepsilon$ and discuss its operational meaning and its quantum-network interpretation. 
In Subsection~\ref{sec: application}, we use the obtained continuity bound to derive the ``Second Law of Efficient Entanglement Manipulation''.  We then provide a comparison between this result and the information-theoretic analogue established in~\cite{W21}.
\subsection{Continuity of One-Shot Computational Distillable Entanglement}\label{subsec: value cont comp}
In this subsection, we show a continuity bound on the computational distillable entanglement in a non-asymptotic setting. We then discuss the operational interpretation of this result and illustrate how the amount of efficiently extractable entanglement depends on the magnitude of the error between the states in the corresponding families.

\begin{theorem}[Continuity Bound of One-Shot $\hat E_D^\varepsilon$]\label{thm: val-cont E_d}
    Let $\{\rho^\lambda_{AB}\}_{\lambda\in\mathbb N}$ and $\{\sigma^\lambda_{AB}\}_{\lambda\in\mathbb N}$ be two families of $2n(\lambda)$--bipartite states, where $n:\mathbb N\rightarrow\mathbb N$ is an arbitrary polynomially bounded function. Assume that for all values of $\lambda\in\mathbb N$, $\frac{1}{2}\|\rho_{AB}^\lambda - \sigma_{AB}^\lambda\|_1\leq\delta(\lambda)$, with $\delta:\mathbb N\rightarrow(0,1)$. Suppose that $\hat E^\varepsilon_D(\{\rho_{AB}^\lambda\}) \geq m(\lambda)$, with $m:\mathbb N\rightarrow\mathbb N$ an arbitrary function and the error of distillation $\varepsilon:\mathbb N\rightarrow(0,1)$ such that $\forall\lambda\in\mathbb N,\,\varepsilon(\lambda)+\delta(\lambda)<1-2^{-2m(\lambda)}$. Then there exists a family of distillation protocols that achieves
    \begin{equation*}
        \hat E_D^\varepsilon(\{\sigma_{AB}^\lambda\}) \geq m(\lambda) - \min\{m(\lambda),s_{\delta,\varepsilon,m}(\lambda)\}
        ,
    \end{equation*}
    with failure probability $p_{\emph{fail}}(\lambda)\leq\frac{2^{2m(\lambda)}}{2^{2m(\lambda)}-1}(\varepsilon(\lambda)+\delta(\lambda))(1-2^{-s_{\delta,\varepsilon,m}(\lambda)}),\,\forall \lambda\in\mathbb N$. Here, the function $s_{\delta,\varepsilon,m}:\mathbb N\rightarrow\mathbb N$ is given for all values of $\lambda$ as
    \begin{equation*}
        s_{\delta,\varepsilon,m}(\lambda) :=\left\lceil \log_2\left(1 + \frac{\delta(\lambda)+2^{-2m(\lambda)}\varepsilon(\lambda)}{\varepsilon(\lambda)[1-\varepsilon(\lambda)-\delta(\lambda)-2^{-2m(\lambda)}]}\right)\right\rceil.
    \end{equation*}
\end{theorem}
\begin{proof}
The proof unfolds in three steps. In the first step, we present a family of distillation channels $\{\hat\Lambda^\lambda\}_{\lambda\in\mathbb N}$ and then argue for its efficiency. In the second step, we derive the effect of applying the obtained distillation channels to a given family of states $\{\sigma_{AB}^\lambda\}_{\lambda\in\mathbb N}$ and describe the resulting distillation protocol. Finally, we calculate the error of distilling $m(\lambda)-s(\lambda)$ EPR pairs from the family of states $\{\sigma_{AB}^\lambda\}_{\lambda\in\mathbb N} $ using the proposed family of distillation protocols.

\paragraph{\textbf{Efficient family of distillation channels.}} We begin by constructing an efficient family of distillation channels $\{\hat\Lambda^\lambda\}_{\lambda\in\mathbb N}$, which will be used in the following to construct a family of distillation protocols. Let $s:\mathbb N\rightarrow\mathbb N$ be a function such that $\forall \lambda\in\mathbb N,\, m(\lambda)\geq s(\lambda)$, where $\hat E ^\varepsilon_D(\{\rho^\lambda_{AB}\}_\lambda)\geq m(\lambda)$ (by assumption of the theorem). Then, consider the following LOCC channels
    \begin{equation}\label{eq: Lambda constr}
        \forall \lambda\in\mathbb N,\quad\hat\Lambda^\lambda:= \hat{\operatorname M}^\lambda_{1,s}\circ {\operatorname{\hat T^\lambda}} \circ\hat\Gamma^\lambda,
    \end{equation}
    with ${\operatorname{\hat T^\lambda}}$ denoting the twirling channel (see Definition~\ref{def: fam of twirling}), $\hat\Gamma^\lambda$ denoting distillation channels of $\{\rho_{AB}^\lambda\}_{\lambda\in\mathbb N}$ (see Definition~\ref{def: comp distil vidick}), and $\hat{\operatorname{M}}_{1,s}$ denoting a measurement in the computational basis of $s(\lambda)$ joint registers of Alice and Bob, which can be expressed as follows
    \begin{equation}
        \hat{\operatorname{M}}^\lambda_{1,s}:= \bigotimes_{i=1}^{s(\lambda)}\hat{\operatorname{M}}^\lambda_{A_iB_i},
    \end{equation}
    where, for all $i\in[s(\lambda)]$, the POVM on the two-qubit system $\hat{\operatorname{M}}^\lambda_{A_iB_i}:=\{\operatorname{M}_1^{(i)\lambda}, \operatorname{M}^{(i)\lambda}_0\}$ is defined as follows
    \begin{equation}\label{eq:M_i}
        \hat{\operatorname{M}}^\lambda_{A_iB_i}:=\{\ketbra{00}{00}_{A_iB_i} + \ketbra{11}{11}_{A_iB_i}, \mathbb I_4 - (\ketbra{00}{00}_{A_iB_i} + \ketbra{11}{11}_{A_iB_i})\}.
    \end{equation}
    Operationally, each $\hat{\operatorname{M}}^\lambda_{A_iB_i}$ corresponds to a measurement that outputs 1 if the classical outcomes of measurements in the computational basis on $A_i$ and $B_i$ are the same and 0 if they are different. Such measurements can be implemented by Alice and Bob measuring their respective registers in the computational basis and then communicating their results to each other via some classical channel.

    We now discuss the efficiency of $\{\hat\Lambda^\lambda\}_{\lambda\in\mathbb N}$. First, $\{\hat\Gamma^\lambda\}_{\lambda\in\mathbb N}$ is efficient by the assumption of computational distillability of $\{\rho^\lambda_{AB}\}_{\lambda\in\mathbb N}$ (see Definition~\ref{def: comp distil vidick}). Next, the family of twirling channels $\{{\operatorname{\hat T^\lambda}}\}_{\lambda\in\mathbb N}$ is efficient by Lemma~\ref{lemma: twirling is efficient}. Lastly, the family of measurements $ \{\hat{\operatorname{M}}_{1,s}^\lambda\}_{\lambda\in\mathbb N}$ is efficient since it is a parallel application of efficient measurement channels $\{\hat{\operatorname{M}}^\lambda_{A_iB_i}\}_{\lambda\in\mathbb N}$. Therefore, since a finite composition of polynomial circuits is also a polynomial circuit, we get that $\{\hat\Lambda^\lambda\}_{\lambda\in\mathbb N}$ is an efficient family of LOCC channels, in the sense of Definition~\ref{def: definition efficient LOCC}. 

    \paragraph{\textbf{The effect of $\hat\Lambda^\lambda$ and the resulting distillation protocol.}} We begin our analysis of the proposed efficient family of  distillation channels. For this, at each $\lambda\in\mathbb N$, we expand $\hat\Lambda^\lambda(\sigma^\lambda_{AB})$ in the following way
    \begin{equation}\label{eq: Lambda on sigma}
        \hat\Lambda^\lambda(\sigma^\lambda_{AB})= \hat{\operatorname{M}}_{1,s}^\lambda\left({\operatorname{\hat T^\lambda}} \circ \hat\Gamma^\lambda (\sigma^\lambda_{AB})\right) = \hat{\operatorname{M}}_{1,s}^\lambda\left({\operatorname{\hat T^\lambda}}(\gamma^\lambda_\sigma)\right),
     \end{equation}
     where in the last equality we have introduced the notation $\hat\Gamma^\lambda(\sigma^\lambda_{AB})=:\gamma^{\lambda}_{\sigma}$.

     Next, for all values of $\lambda\in\mathbb N$, we recall the result of Proposition~\ref{cor: action of eff twirl} to obtain that the action of each twirling channel ${\operatorname{\hat T^\lambda}}$ on $\gamma^\lambda_{\sigma}$ can be written as follows:
     \begin{align}\label{eq: T on gamma}\begin{split}
        \tau^\lambda_\sigma := {\operatorname{ \hat T^\lambda}} (\gamma^\lambda_\sigma)&= \Phi^{\otimes m(\lambda)}\operatorname{Tr}\left(\Phi^{\otimes m(\lambda)}\gamma^\lambda_\sigma\right) +\frac{\mathbb I_4^{\otimes m(\lambda)} - \Phi^{\otimes m(\lambda)}}{2^{2m(\lambda)}-1}\operatorname{Tr}\left( \left(\mathbb I_4^{\otimes m(\lambda)} - \Phi^{\otimes m(\lambda)}\right) \gamma_\sigma^\lambda \right)\\
         &=\Phi^{\otimes m(\lambda)}\left( \alpha_\lambda\left(1+\frac{1}{2^{2m(\lambda)}-1}\right) - \frac{1}{2^{2m(\lambda)}-1} \right) + \frac{\mathbb I^{\otimes m(\lambda)}_4}{2^{2m(\lambda)}}\frac{2^{2m(\lambda)}}{2^{2m(\lambda)}-1}(1-\alpha_\lambda),
        \end{split}
     \end{align}
     where  $\mathbb I_4$ is the identity operator on $\mathbb C^2\otimes\mathbb C^2$ and $\alpha_\lambda:= \operatorname{Tr}\left(\Phi^{\otimes m(\lambda)}\gamma^\lambda_\sigma\right)$.

     We now discuss the effect of the measurement channel $\operatorname{M}^\lambda_{1,s}$ on $\tau^\lambda_\sigma$. The classical outcomes of this measurement procedure are described  by a string of bits shared between Alice and Bob $x^{\lambda}\in\{0,1\}^{s(\lambda)}$. Then, we are interested in the case in which $x^\lambda = {1}^{s(\lambda)}$, which  corresponds to the case in which all measurements $\operatorname{M}^\lambda_{A_iB_i}$ (see Equation \eqref{eq:M_i}) return 1. Operationally, this corresponds to Alice and Bob having perfectly correlated results in all of the $s(\lambda)$ measured registers.

     With this, we calculate the probability of an outcome $x^\lambda$ being $1^{s(\lambda)}$, when the measurement is done on the state $\tau^\lambda_\sigma$, which we denote by $p_1(\lambda)$. For this, we notice that  $\tau^\lambda_\sigma$ given in Equation \eqref{eq: T on gamma} is a statistical mixture of two states $\Phi^{\otimes m(\lambda)}$ and $\mathbb I_4^{\otimes m (\lambda)}/2^{2m(\lambda)}$. Therefore, for each $\lambda\in\mathbb N$, the probability of an outcome being $1^{s(\lambda)}$ can be calculated as follows
     \begin{align}\label{eq: p_1}
         \begin{split}
                      p_1(\lambda) &= \beta_\lambda\operatorname{P}\Biggl( 1^{s(\lambda)}\mid \Phi^{\otimes m(\lambda)}\Biggr) + (1-\beta_\lambda)\operatorname P\left(1^{s(\lambda)}\mid \frac{\mathbb I^{\otimes m(\lambda)}_4}{2^{2m(\lambda)}}\right) \\
                      &=\beta_\lambda + \frac{1}{2^{s(\lambda)}}(1-\beta_\lambda),
         \end{split}
     \end{align}
    where we denote $\beta_\lambda:= \alpha_\lambda\left(1+\frac{1}{2^{2m(\lambda)}-1}\right) - \frac{1}{2^{2m(\lambda)}-1}$. Here, we used the observation that $\operatorname{P}( 1^{s(\lambda)}\mid \Phi^{\otimes m(\lambda)})=1$, because $\Phi^{\otimes m(\lambda)}$ is a maximally entangled state, and that $\operatorname{P}( 1^{s(\lambda)}\mid2^{-2m(\lambda)}\mathbb I^{\otimes m(\lambda)}_4 )=2^{-s(\lambda)}$  since the probability for each register $A_i$ and $B_i$ to have the same measurement outcomes is exactly $1/2$, for every $i\in[s(\lambda)]$.  Using this, we derive the effect of $\operatorname{M}^\lambda_{1,s}$ on $\tau^\lambda_\sigma$. Namely, if the observed classical outcome of applying $\operatorname{M}^\lambda_{1,s}$ is $1^{s(\lambda)}$,  then the post-measurement state $\mu^\lambda_1$ has the following form
    \begin{equation}\label{eq: mu_1}
        \forall \lambda\in\mathbb N,\quad\mu_1^\lambda = \frac{1}{p_1(\lambda)}\left[\beta_\lambda\Phi^{\otimes( m(\lambda)-s(\lambda))} + (1-\beta_\lambda)\frac{1}{2^{s(\lambda)}}\frac{\mathbb I_4^{\otimes(m(\lambda)-s(\lambda))}}{2^{2m(\lambda)-2s(\lambda)}}\right].
    \end{equation}
    If the classical output of the measurement is not $1^{s(\lambda)}$, then the post-measurement state $\mu_0^\lambda$ is simply a maximally mixed state of the following form
    \begin{equation}\label{eq: mu_0}
        \forall \lambda\in\mathbb N,\quad \mu_0^\lambda = \frac{1}{1-p_1(\lambda)}\frac{2^{s(\lambda)}-1}{2^{s(\lambda)}}(1-\beta_\lambda)\frac{\mathbb I_4^{\otimes (m(\lambda)-s(\lambda))}}{2^{2m(\lambda)-2s(\lambda)}} = \frac{\mathbb I_4^{\otimes (m(\lambda)-s(\lambda))}}{2^{2m(\lambda)-2s(\lambda)}}.
    \end{equation}

    Now, we describe an efficient family of distillation protocols. The first step of the protocol, which we denote by $\tilde \Gamma^\lambda$ for each $\lambda$, consists of applying $\hat\Lambda^\lambda$ to $\sigma^\lambda$ and outputting 1 if $x^\lambda=1^{s(\lambda)}$ and 0 otherwise. Formally, we can write this as follows
    \begin{equation}
    \forall \lambda\in\mathbb N,\quad\tilde \Gamma^\lambda(\sigma^\lambda) = p_1(\lambda)\mu_1^\lambda\otimes \ketbra{1}{1}_X + (1-p_1(\lambda))\mu_0^\lambda\otimes \ketbra{0}{0}_X,
    \end{equation}
    where $p_1(\lambda)$ is defined in Equation \eqref{eq: p_1}, $\mu_1^\lambda$ and $\mu_0^\lambda$ are post-measurement states defined in Equations~\eqref{eq: mu_1} and \eqref{eq: mu_0}, respectively. 
    The final step consists of Alice and Bob projecting register $X$ onto the state $|1\rangle_X$. This yields the final state $\mu_1^\lambda$. Note that the failure probability of this protocol is as follows:
        \begin{align}
        \begin{split}
            p_{\mathrm{fail}}(\lambda)
=1-p_1(\lambda)
&=(1-\beta_\lambda)\bigl(1-2^{-s(\lambda)}\bigr)\\
&=
\frac{2^{2m(\lambda)}}{2^{2m(\lambda)}-1}
(1-\alpha_\lambda)
\bigl(1-2^{-s(\lambda)}\bigr)\\
&\leq
\frac{2^{2m(\lambda)}}{2^{2m(\lambda)}-1}
\bigl(\varepsilon(\lambda)+\delta(\lambda)\bigr)
\bigl(1-2^{-s(\lambda)}\bigr).
        \end{split}
        \end{align}
    \paragraph{\textbf{Computing the error of distillation.}} We now calculate the error of distillation (see Definition~\ref{def: comp distil vidick}) using $\{\Pi_{X,1}\circ\tilde \Gamma^\lambda\}_{\lambda\in\mathbb N}$ acting on $\{\sigma^\lambda_{AB}\}_{\lambda\in\mathbb N}$, where $\Pi_{X,1}$ is the projection operator onto
    the state $|1\rangle_X$. For this, we first calculate the fidelity between the output state $\mu_1^\lambda$ and $\Phi^{\otimes( m(\lambda)-s(\lambda))}$ as follows
\begin{align}
    \begin{split}
        \forall \lambda\in\mathbb N,\quad F(\mu_1^\lambda,\Phi^{\otimes( m(\lambda)-s(\lambda))} ) &=  \frac{1}{p_1(\lambda)}\left[\beta_\lambda + (1-\beta_\lambda)\frac{1}{2^{s(\lambda)}}\frac{1}{2^{2m(\lambda)-2s(\lambda)}}\right]\\
        & \geq \frac{\beta_\lambda}{\beta_\lambda + 2^{-s(\lambda)}(1-\beta_\lambda)} \\
        & = \frac{2^{2m(\lambda)}\alpha_\lambda - 1 }{2^{2m(\lambda)}\alpha_\lambda-1+2^{-s(\lambda)}2^{2m(\lambda)}(1-\alpha_\lambda)} \\
        &=1-\frac{2^{-s(\lambda)}(1-\alpha_\lambda)}{\alpha_\lambda-2^{-2m(\lambda)}+2^{-s(\lambda)}(1-\alpha_\lambda)},
    \end{split}
\end{align}
where, in the first inequality, we have inserted the value of $p_1(\lambda)$ obtained in Equation~\eqref{eq: p_1}. 
From the lower bound obtained above, we deduce the upper bound on the error of distillation, which is given by 
\begin{align}\label{eq: p_err with guessing}
    \begin{split}
        p_{\text{err}}(\Pi_{X,1}\circ\tilde\Gamma^\lambda,\sigma^\lambda)  &= 1-F(\mu_1^\lambda, \Phi^{\otimes(m(\lambda)-s(\lambda))})\\
        &\leq\frac{2^{-s(\lambda)}(1-\alpha_\lambda)}{\alpha_\lambda-2^{-2m(\lambda)}+2^{-s(\lambda)}(1-\alpha_\lambda)} \\
        &\leq\frac{2^{-s(\lambda)}[\varepsilon(\lambda) + \delta(\lambda)]}{(1-2^{-s(\lambda)})[1-\varepsilon(\lambda)-\delta(\lambda)]+ 2^{-s(\lambda)}-2^{-2m(\lambda)}}.
    \end{split}
\end{align}
To obtain the second inequality, we used the result of the following derivation for the coefficient $\alpha_\lambda$:
\begin{align}\label{eq: alpha bound}
     \begin{split}
         \alpha_\lambda = \operatorname{Tr}(\Phi^{\otimes m(\lambda)}\gamma_\sigma^\lambda) &= F(\Phi^{\otimes m(\lambda)},\hat\Gamma^\lambda(\sigma^\lambda_{AB}))
         \\
         &\geq F(\Phi^{\otimes m(\lambda)},\hat\Gamma^\lambda(\rho^\lambda_{AB})) - \frac{1}{2}\left\| \hat\Gamma^\lambda(\rho^\lambda_{AB}) - \hat\Gamma^\lambda(\sigma^\lambda_{AB}) \right\|_1
         \\
         &\geq F(\Phi^{\otimes m(\lambda)},\hat\Gamma^\lambda(\rho^\lambda_{AB})) - \frac{1}{2}\left\| \rho^\lambda_{AB} - \sigma^\lambda_{AB} \right\|_1
         \\
         &\geq 1-\varepsilon(\lambda) - \delta(\lambda),
     \end{split}
     \end{align}
where the first inequality comes from Proposition~\ref{prop: close states fid} and the remaining inequalities are straightforward consequences of the theorem's assumptions.

Finally, for $\{\sigma^\lambda\}_{\lambda\in\mathbb N}$ to be efficiently distillable to $m(\lambda)-s(\lambda)$ EPR pairs, we require the r.h.s. of Equation \eqref{eq: p_err with guessing} to be less than or equal to $\varepsilon(\lambda)$ for all values of $\lambda$. That is, for all $\lambda\in\mathbb N$, we require the following
\begin{equation}\label{eq: bound on eps}
   \frac{2^{-s(\lambda)}[\varepsilon(\lambda) + \delta(\lambda)]}{(1-2^{-s(\lambda)})[1-\varepsilon(\lambda)-\delta(\lambda)]+ 2^{-s(\lambda)}-2^{-2m(\lambda)}}\leq\varepsilon(\lambda).
\end{equation}

Hence, we have 
\begin{equation}\label{eq: bound on s}
    \forall\lambda\in\mathbb N,\quad s(\lambda) \geq \log_2\left(1 + \frac{\delta(\lambda)+2^{-2m(\lambda)}\varepsilon(\lambda)}{\varepsilon(\lambda)[1-\varepsilon(\lambda)-\delta(\lambda)-2^{-2m(\lambda)}]} \right).
\end{equation}
Therefore, to satisfy the distillability condition (Equation~\eqref{eq: bound on eps}), it is sufficient to take the smallest integer greater than or equal to the r.h.s. of Equation \eqref{eq: bound on s} as the value of $s_{\delta,\varepsilon,m}(\lambda)$. That is, 
\begin{equation}\label{eq: s lambda}
    \forall\lambda\in\mathbb N,\quad s_{\delta,\varepsilon,m}(\lambda) := \left\lceil \log_2\left(1 + \frac{\delta(\lambda)+2^{-2m(\lambda)}\varepsilon(\lambda)}{\varepsilon(\lambda)[1-\varepsilon(\lambda)-\delta(\lambda)-2^{-2m(\lambda)}]}\right)\right\rceil.
\end{equation}
Finally, if $s_{\delta,\varepsilon,m}(\lambda)\geq m(\lambda)$ for a given $\lambda\in\mathbb N$, we instead use the trivial deterministic LOCC protocol that discards the input state and outputs zero EPR pairs. This concludes the proof of the theorem.
\end{proof}

We now discuss an operational interpretation of the result in Theorem~\ref{thm: val-cont E_d} through a quantum-network lens. Suppose that a network link is intended to generate a family of states \(\{\rho^\lambda_{AB}\}_\lambda\), for which the network nodes possess an efficient protocol that distills \(m(\lambda)\) EPR pairs with error \(\epsilon(\lambda)\). Due to imperfect local gates, transmission noise, or finite quantum-memory lifetimes, the link may instead generate a family \(\{\sigma^\lambda_{AB}\}_\lambda\) that is \(\delta(\lambda)\)-close to the intended family. Directly applying the original distillation protocol may then increase the error beyond the target value. Theorem~\ref{thm: val-cont E_d} shows that the nodes can supplement this protocol with efficient twirling and local measurements to recover the original error \(\epsilon(\lambda)\), conditioned on success. The price consists in sacrificing \(s_{\delta,\epsilon,m}(\lambda)\) EPR pairs and a failure probability bounded by \(P_{\mathrm{fail}}(\lambda)\). Thus, in a network scenario, the theorem translates an error in the state generated by a link into explicit guarantees on the number, fidelity, and delivery probability of the resulting EPR pairs. These quantities constitute central performance parameters for entanglement delivery and purification in quantum networks~\cite{PDT22,VTS23,XLX24}.

To analyze the dependence of the performance of the proposed protocol (see Equation~\eqref{eq: Lambda constr}) on the form of the function $\delta(\lambda)$, we plot the number of $\varepsilon(\lambda)$-approximate distilled EPR pairs for different types of $\delta(\lambda)$: $\tilde m(\lambda) := m(\lambda)-\min\{s_{\delta,\varepsilon,m}(\lambda),m(\lambda)\}$. For this, we fix the error of distillation $\varepsilon(\lambda)$ to be $1/\lambda$ and the original number of distilled EPR pairs $m(\lambda)$ to be a linear function of $\lambda$. We then take  the function $\delta(\lambda)$ to be either inverse logarithmic ($1/\log_2(\lambda)$), hyperbolic ($1/\lambda$), or exponentially decaying ($2^{-\lambda}$). We present the results in Figure~\ref{fig:tilde m}. Note that $\lambda$ here starts at three to satisfy the condition in the assumption of Theorem~\ref{thm: val-cont E_d} for the chosen parameters $\varepsilon$ and $m$. 

\begin{figure}[h!]
    \centering
    \includegraphics[width=0.7\linewidth]{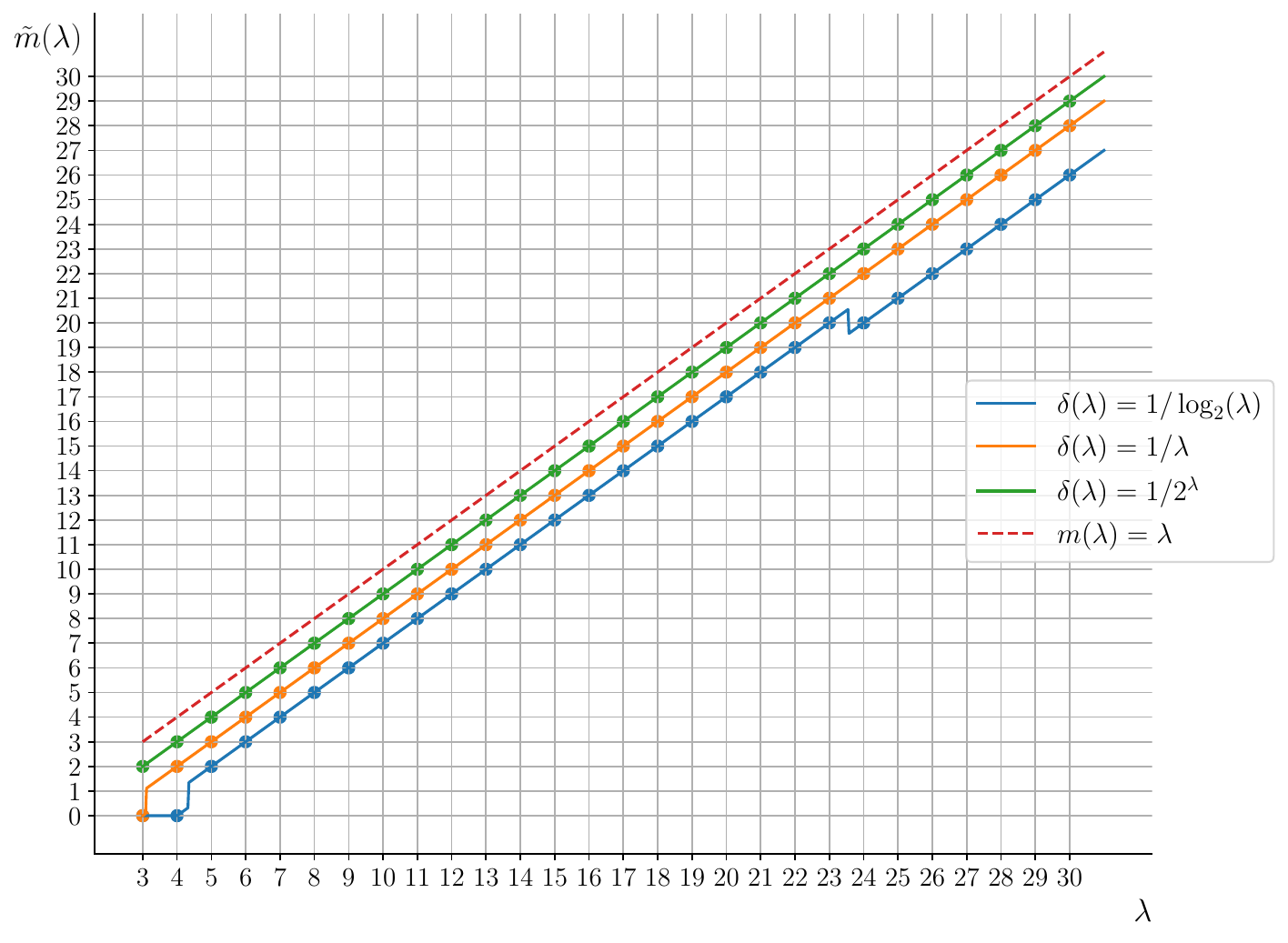}
    \caption{Dependence of the resulting number of efficiently distillable EPR pairs ($\tilde m$) on the form of the function governing the error in state preparation ($\delta$). Error of distillation $\varepsilon(\lambda)=1/\lambda$. Number of initial efficiently distilled EPR pairs $m(\lambda)=\lambda$.}
    \label{fig:tilde m}
\end{figure}

From Figure \ref{fig:tilde m} we observe the overall robustness trend of the efficiently extractable entanglement. That is, the faster the function $\delta$ decays relative to $\varepsilon$, the closer the value of the computational one-shot distillable entanglement of the ``distorted family'' is to the original one. We note that the result of Theorem  \ref{thm: val-cont E_d} naturally extends to the uniform case, as demonstrated by the previous results of this paper. 

\subsection{Putting it All Together: Second Law of Efficient Entanglement Manipulation}\label{sec: application}
In this subsection, we combine the results of this paper to obtain the ``second law'' of non-asymptotic entanglement manipulation in the computationally restricted setting. Namely, we establish that for any given family of bipartite states, there exists a quantitative relation between valid lower and upper bounds on computational one-shot distillable entanglement and cost. Similar to \cite{W21}, we refer to this result as the ``Second Law of Efficient Entanglement Manipulation''.
\begin{theorem}[Second Law of Efficient Entanglement Manipulation]\label{th: second law}
    Let $\lambda\in\mathbb N$, $n:\mathbb N\rightarrow\mathbb N$ be an arbitrary polynomially bounded function. Let $\{\rho^\lambda_{AB}\}_{\lambda\in\mathbb N}$ be a family of $2n(\lambda)$-bipartite mixed states. Assume that $\hat{E}^{\varepsilon_2}_D(\{\rho^\lambda_{AB}\}) \geq d(\lambda)$ and $\hat{E}^{\varepsilon_1}_C(\{\rho^\lambda_{AB}\}) \leq c(\lambda)$ where $\varepsilon_1,\varepsilon_2:\mathbb N\rightarrow (0,1)$ such that $\forall\lambda\in\mathbb N,\,\varepsilon_2(\lambda)+\sqrt{\varepsilon_1(\lambda)}< 1-2^{-2d(\lambda)}$. Denote $P_{\emph{fail}}(\lambda):=\frac{2^{2d(\lambda)}}{2^{2d(\lambda)}-1}(\varepsilon_2(\lambda)+\sqrt{\varepsilon_1(\lambda)})(1-2^{-s_{\sqrt{\varepsilon_1},\varepsilon_2,d}(\lambda)})$. Then the following relation holds:
    \begin{equation*}
        \forall \lambda\in\mathbb N,\quad c(\lambda)  \geq d(\lambda)-t_{\sqrt{\varepsilon_1},\varepsilon_2,d}(\lambda),
    \end{equation*}
    where $t_{\sqrt{\varepsilon_1},\varepsilon_2,d}:\mathbb N\rightarrow\mathbb R$ is defined as follows
    \begin{equation*}
        t_{\sqrt{\varepsilon_1},\varepsilon_2,d}(\lambda):=\min\{s_{\sqrt{\varepsilon_1},\varepsilon_2,d}(\lambda),d(\lambda)\}+\log_2\left(\frac{1}{1-\varepsilon_2(\lambda)}\right) + \log_2\left(\frac{1}{1-P_{\emph{fail}}(\lambda)}\right),
    \end{equation*}
    with  the function $s_{\sqrt{\varepsilon_1},\varepsilon_2,d}:\mathbb N\rightarrow\mathbb N$ given by Theorem~\ref{thm: val-cont E_d}.
\end{theorem}
\begin{proof}
    First, for every value of $\lambda\in\mathbb N$, we consider  the following chain of transformations
    \begin{equation}\label{eq: chain of transformations}
        \Phi^{\otimes c(\lambda)}_{AB} \xrightarrow[\varepsilon_1(\lambda)]{\hat\Lambda^\lambda} \rho^\lambda_{AB} \xrightarrow[\varepsilon_2(\lambda)]{\hat\Gamma^\lambda} \Phi^{\otimes d(\lambda)-f(\lambda)}_{AB},
    \end{equation}
    where $\hat\Lambda^\lambda\in\{\hat\Lambda^\lambda\}_{\lambda\in\mathbb N}$ and $\hat\Gamma^\lambda\in\{\hat\Gamma^\lambda\}_{\lambda\in\mathbb N}$ are  dilution and distillation LOCC channels of $\{\rho_{AB}^\lambda\}_{\lambda\in\mathbb N}$, respectively. 
    The function $f:\mathbb  N_+\rightarrow\mathbb N$ will be obtained in the following as a consequence of Theorem~\ref{thm: val-cont E_d}.

    Now, consider the resulting family of states after the first transformation in Equation~\eqref{eq: chain of transformations}: $\{\hat \Lambda^\lambda(\Phi^{\otimes c(\lambda)}_{AB})\}_{\lambda\in\mathbb N}$. By applying Fuchs-van de Graaf inequalities (Proposition~\ref{prop: fidelity and one norm}) to the condition on the error of dilution in Definition~\ref{def: comp entang cost vidick},  we get that
    \begin{equation}\label{eq: condition on image of Phi}
        \forall \lambda\in\mathbb N,\quad \frac{1}{2}\left\| \hat\Lambda^\lambda\left(\Phi^{\otimes c(\lambda)}_{AB}\right)- \rho^\lambda_{AB}\right\|_1 \leq \sqrt{\varepsilon_1(\lambda)}.
    \end{equation}

    Since Equation~\eqref{eq: condition on image of Phi} holds, we apply Theorem~\ref{thm: val-cont E_d} to the family of states $\{\rho^\lambda_{AB}\}_{\lambda\in\mathbb N}$ and  $\{\hat \Lambda^\lambda(\Phi^{\otimes c(\lambda)}_{AB})\}_{\lambda\in\mathbb N}$. As a result, we obtain that, with failure probability $p_{\text{fail}}(\lambda)$ upper-bounded by $P_{\text{fail}}(\lambda)$ (given by Theorem~\ref{thm: val-cont E_d}), the following holds
    \begin{equation}\label{eq: d-f}
            \hat E_D^{\varepsilon_2}\left(\left\{\hat\Lambda^\lambda\left(\Phi^{\otimes c(\lambda)}_{AB}\right)\right\}\right) \geq d(\lambda) - \min\{s_{\sqrt{\varepsilon_1},\varepsilon_2,d}(\lambda),d(\lambda)\},
    \end{equation}
    where $s_{\sqrt{\varepsilon_1},\varepsilon_2,d}(\lambda)$ is defined in Equation~\eqref{eq: s lambda}.

    The next step of the proof consists of utilizing LOCC monotonicity of one-shot computational distillable entanglement, which was established previously in this work (see Proposition~\ref{thm: locc not monotone comp dist ent}). For this, we first account for the failure probability of the underlying distillation protocols used to obtain Equation~\eqref{eq: d-f}. That is, we use Remark~\ref{remark: SLOCC} to obtain that there exists a family of distillation protocols with zero failure probability that distill the same amount of EPR pairs as in  Equation~\eqref{eq: d-f} with the modified  error of distillation $\bar\varepsilon_2(\lambda):=\varepsilon_2(\lambda) + (1-\varepsilon_2(\lambda))P_{\text{fail}}(\lambda)<1$. Here, $P_{\text{fail}}(\lambda)$ denotes the upper bound on the failure probability given by Theorem~\ref{thm: val-cont E_d}: $P_{\text{fail}}(\lambda):= \frac{2^{2d(\lambda)}}{2^{2d(\lambda)}-1}(\varepsilon_2(\lambda)+\sqrt{\varepsilon_1(\lambda)})(1-2^{-s_{\sqrt{\varepsilon_1},\varepsilon_2,d}(\lambda)})$. Next, since $\{\hat\Lambda^\lambda\}_{\lambda\in\mathbb N}$ is a family of efficient LOCC channels by Definition~\ref{def: definition efficient LOCC}, we apply Proposition~\ref{thm: locc not monotone comp dist ent} to the deterministic family of distillation protocols obtained above. This gives the following:
    \begin{equation}\label{eq: locc monotonicity applic}
         \hat E_D^{\bar\varepsilon_2}\left(\left\{\Phi^{\otimes c(\lambda)}_{AB}\right\}\right) \geq d(\lambda)- \min\{s_{\sqrt{\varepsilon_1},\varepsilon_2,d}(\lambda),d(\lambda)\}.
    \end{equation}

     We now obtain an upper bound on the one-shot computational distillable entanglement $\hat E^{\bar\varepsilon_2}_D\left(\left\{\Phi^{\otimes c(\lambda)}_{AB}\right\}\right)$. For this, consider the following chain of inequalities
    \begin{equation}\label{eq: triv ineq}
        \forall\lambda\in\mathbb N,\quad \hat E^{\bar\varepsilon_2}_D\left(\left\{\Phi^{\otimes c(\lambda)}_{AB}\right\}\right)\leq E^{\bar\varepsilon_2}_D\left(\Phi^{\otimes c(\lambda)}_{AB}\right)\leq R_{H}^{\bar\varepsilon_2}\left(\Phi^{\otimes c(\lambda)}_{AB}\right),
    \end{equation}
    where $R_{H}^{\bar\varepsilon_2}$ is the $\bar\varepsilon_2$-hypothesis testing Rains relative entropy~\cite{khatri2024principles}. 
     The first inequality holds due to an inclusion of the class of efficient LOCC channels in the class of general LOCC channels. The second inequality holds by~\cite[Theorem 13.7]{khatri2024principles}.

    Now, we utilize the result of~\cite[Proposition 1]{W21}. That is, for $\varepsilon:\mathbb N\rightarrow[0,1)$ and for all values of $\lambda\in\mathbb N$, the $\varepsilon(\lambda)$-hypothesis testing Rains relative entropy of the maximally entangled state $\Phi^{\otimes c(\lambda)}_{AB}$ obeys the following expression
\begin{equation}\label{eq: bound on Ed of EPR}
    R_H^{\varepsilon}\left(\Phi^{\otimes c(\lambda)}_{AB}\right)=c(\lambda)+\log _2\left(\frac{1}{1-\varepsilon(\lambda)}\right) .
\end{equation}
    
    Combining Equation~\eqref{eq: triv ineq} and Equation~\eqref{eq: bound on Ed of EPR} gives us the following
    \begin{equation}\label{eq: upper bound}
        \hat E^{\bar\varepsilon_2}_D\left(\left\{\Phi^{\otimes c(\lambda)}_{AB}\right\}\right) \leq c(\lambda)+\log_2\left(\frac{1}{1-\bar\varepsilon_2(\lambda)}\right).
    \end{equation}

    Finally, combining Equations~\eqref{eq: locc monotonicity applic} and~\eqref{eq: upper bound}, and using the definition of \(\bar\varepsilon_2(\lambda)\), we obtain the desired result
    \begin{align}
        \begin{split}
            \forall \lambda\in\mathbb N,\quad c(\lambda) &\geq d(\lambda) -  \min\{s_{\sqrt{\varepsilon_1},\varepsilon_2,d}(\lambda),d(\lambda)\} -\log_2\left(\frac{1}{1-\bar\varepsilon_2(\lambda)}\right)\\
            &= d(\lambda) -  \min\{s_{\sqrt{\varepsilon_1},\varepsilon_2,d}(\lambda),d(\lambda)\} -\log_2\left(\frac{1}{1-\varepsilon_2(\lambda)}\right)-\log_2\left(\frac{1}{1-P_{\text{fail}}(\lambda)}\right),
        \end{split}
    \end{align}
which concludes the proof of the theorem.
    \end{proof}

We now discuss the result of Theorem~\ref{th: second law}. Our result establishes a quantitative relation between non-asymptotic computational entanglement measures, thereby generalizing the result of \cite{W21}. Informally, we showed that a valid upper bound on the computational one-shot entanglement cost of a given family is at least the computational one-shot distillable entanglement minus an additive term. As it was covered in the discussion after Theorem~\ref{thm: val-cont E_d}, the behavior of the additive term depends on the relation between the functions $\varepsilon$ and $\delta$. Here, this implies that the relation between computational one-shot distillable entanglement and cost depends on the relation between $\sqrt{\varepsilon_1}$ and $\varepsilon_2$. We note that the infimum of the additive term in Theorem~\ref{th: second law} is equal to 1, and $t_{\sqrt{\varepsilon_1},\varepsilon_2,d}$ approaches this value for a wide range of errors of distillation and dilution. 

Moreover, we now compare the result of Theorem~\ref{th: second law} to a related result in~\cite[Theorem 1]{W21}. For this, we first note that the relation in~\cite[Theorem 1]{W21} can be trivially extended to the case of computational measures, since, with a slight abuse of notation, $\hat E^\varepsilon_D\leq E^\varepsilon_D$ and $E^\varepsilon_C\leq\hat E^\varepsilon_C$~\cite{arnon2023computational}. However, the generalized inequality of~\cite[Theorem 1]{W21} can be tight only for the families of states for which $\hat E^\varepsilon_D= E^\varepsilon_D$ and $E^\varepsilon_C=\hat E^\varepsilon_C$.
For all other families, the gap between non-asymptotic computational and information-theoretic versions of these entanglement measures can be arbitrarily large, as shown in this work for computational one-shot distillable entanglement (see Theorem~\ref{lemma: bound on E_D so}). This can make the generalized inequality of~\cite[Theorem 1]{W21} less informative compared to the original information-theoretic version, i.e., it may say little about how close the actual computational cost and distillation bounds are to each other. We make this observation formal in the following proposition.
\begin{proposition}[Extended Theorem 1~\cite{W21}]\label{prop: extended W21}
    Let $\lambda\in\mathbb N$, $n_A,n_B:\mathbb N\rightarrow\mathbb N$ be polynomially bounded functions. Let $\{\rho^\lambda_{AB}\}_{\lambda\in\mathbb N}$ be a family of $n_A(\lambda)+n_B(\lambda)$-bipartite quantum states. Assume that $\hat E^{\varepsilon_1}_C(\{\rho^\lambda_{AB}\}_\lambda)\leq c(\lambda)$ and $\hat E^{\varepsilon_2}_D(\{\rho^\lambda_{AB}\}_\lambda)\geq d(\lambda)$, where $\varepsilon_1,\varepsilon_2:\mathbb N\rightarrow[0,1)$ are the errors of dilution and distillation, respectively, such that $[\sqrt{\varepsilon_1(\lambda)}+\sqrt{\varepsilon_2(\lambda)}]^2<1,\,\forall\lambda\in\mathbb N$. Then the following holds
    \begin{equation*}
        \forall\lambda\in\mathbb N,\quad d(\lambda)\leq c(\lambda) + g_{\varepsilon_1,\varepsilon_2}(\lambda),
    \end{equation*}
    where $g_{\varepsilon_1,\varepsilon_2}:\mathbb N\rightarrow[0,+\infty)$ is defined as follows
    \begin{equation*}
        g_{\varepsilon_1,\varepsilon_2}(\lambda):= \operatorname{log}_2\left(\frac{1}{1-[\sqrt{\varepsilon_1(\lambda)}+\sqrt{\varepsilon_2(\lambda)}]^2}\right).
    \end{equation*}
    However, at every $\lambda\in\mathbb N$ satisfying $d(\lambda)<E^{\varepsilon_2}_D(\rho^\lambda_{AB})$ or $c(\lambda)>E^{\varepsilon_1}_C(\rho^\lambda_{AB})$, the following inequality holds: 
    \begin{equation*}
        d(\lambda)\leq c(\lambda) + g_{\varepsilon_1,\varepsilon_2}(\lambda)-1.
    \end{equation*}
    Moreover, for $\varepsilon_1,\varepsilon_2$ vanishing and $d(\lambda)<E^{\varepsilon_2}_D(\rho^\lambda_{AB})$ or $c(\lambda)>E^{\varepsilon_1}_C(\rho^\lambda_{AB})$ holding for all sufficiently large $\lambda\in\mathbb N$, it follows that:
    \begin{equation*}
        c(\lambda)+g_{\varepsilon_1,\varepsilon_2}(\lambda)-d(\lambda)\geq g_{\varepsilon_1,\varepsilon_2}(\lambda)+1.
    \end{equation*}
\end{proposition}
\begin{proof}
    The first part of the proof follows directly from applying inequalities $\hat E^\varepsilon_D\leq E^\varepsilon_D$ and $E^\varepsilon_C\leq\hat E^\varepsilon_C$ to the result of~\cite[Theorem 1]{W21}. 

    To show the second part of the proof, define the following non-negative quantities
    \begin{equation}\label{eq: Deltas}
        \Delta_D(\lambda):= E^{\varepsilon_2}_D(\rho^\lambda_{AB})-d(\lambda) \quad\text{and}\quad \Delta_C(\lambda):= c(\lambda) - \ E^{\varepsilon_1}_C(\rho^\lambda_{AB}),
    \end{equation}
    where $d,c:\mathbb N\rightarrow\mathbb N$ are the valid lower and upper bounds on $\hat E^{\varepsilon_2}_D$ and $\hat E^{\varepsilon_1}_C$, respectively. 
    By the assumption of the proposition, either $d(\lambda)< E^{\varepsilon_2}_D(\rho^\lambda_{AB})$ or $E^{\varepsilon_1}_C(\rho^\lambda_{AB}) <c(\lambda)$ for some values of $\lambda$. Therefore, $\Delta_C(\lambda) + \Delta_D(\lambda)\geq 1$, for each $\lambda\in\mathbb N$ at which the separation occurs. 
    
    Additionally, let $\Delta_W(\lambda)\geq 0$ be a slack in~\cite[Theorem 1]{W21} for the state $\rho^\lambda_{AB}$ defined as follows:
    \begin{equation}\label{eq: slack W}
        \Delta_W(\lambda):=E_C^{\varepsilon_1}(\rho^\lambda_{AB})+g_{\varepsilon_1,\varepsilon_2}(\lambda)-E_D^{\varepsilon_2}(\rho^\lambda_{AB})\geq0,
    \end{equation}

    Now, we derive the slack in the computational extension of~\cite[Theorem 1]{W21}
    \begin{align}\label{eq: slack}
        \begin{split}
            c(\lambda)
+g_{\varepsilon_1,\varepsilon_2}(\lambda)
-d(\lambda)
&=
\left(c(\lambda)-E_C^{\varepsilon_1}(\rho^\lambda_{AB})\right)
+
\left(E_C^{\varepsilon_1}(\rho^\lambda_{AB})
+g_{\varepsilon_1,\varepsilon_2}(\lambda)
-E_D^{\varepsilon_2}(\rho^\lambda_{AB})\right)\\
&+
\left(E_D^{\varepsilon_2}(\rho^\lambda_{AB})
-d(\lambda)\right)\\
&=
\Delta_C(\lambda)+\Delta_W(\lambda)+\Delta_D(\lambda) \\
&\geq1,
        \end{split}
    \end{align}
    where the inequality comes from the established inequality $\Delta_C(\lambda)+\Delta_D(\lambda)\geq 1$. 
    
    Moreover, we note that, in the case of $\Delta_C(\lambda) + \Delta_D(\lambda)\geq 1$ for sufficiently large $\lambda$ with vanishing $\varepsilon_1(\lambda)$ and $\varepsilon_2(\lambda)$, the bound in Equation~\eqref{eq: slack} can be improved. To see this, first note that the function $g_{\varepsilon_1,\varepsilon_2}(\lambda)<1$ for all sufficiently large $\lambda$. Then, the slack $\Delta_W(\lambda)$  can be lower bounded as follows: 
    \begin{equation}\label{eq: bound on slack W improved}
         \Delta_W(\lambda)\geq g_{\varepsilon_1,\varepsilon_2}(\lambda),
    \end{equation}
    for all sufficiently large $\lambda$. This holds since, for values of $\lambda$ at which $g_{\varepsilon_1,\varepsilon_2}<1$, the result of~\cite[Theorem 1]{W21} yields that
    \begin{equation}\label{eq: simple ineq}
        E_C^{\varepsilon_1}(\rho^\lambda_{AB})\geq \lceil E_D^{\varepsilon_2}(\rho^\lambda_{AB})-g_{\varepsilon_1,\varepsilon_2}(\lambda)\rceil\geq E^{\varepsilon_2}_D(\rho^\lambda_{AB}),
    \end{equation}
    where the inequalities hold because $E_C^{\varepsilon_1}(\rho^\lambda_{AB}),E^{\varepsilon_2}_D(\rho^\lambda_{AB})$ are both integers by definition (Definition~\ref{def: inf-th dist ent} and Definition~\ref{def: inf-th cost ent}).
    Substituting Equation~\eqref{eq: simple ineq} into Equation~\eqref{eq: slack W} yields the desired inequality (Equation~\eqref{eq: bound on slack W improved}). 

    Finally, substituting Equation~\eqref{eq: bound on slack W improved} into the third line of Equation~\eqref{eq: slack} yields the following: 
    \begin{equation}\label{eq: add 1}
    c(\lambda) + g_{\varepsilon_1,\varepsilon_2}(\lambda)-d(\lambda)\geq          g_{\varepsilon_1,\varepsilon_2}(\lambda) +1,
    \end{equation}
    holding for all sufficiently large $\lambda\in\mathbb N$.
    This establishes the claim of the proposition. 
\end{proof}

With this, we can draw a comparison between the result in Theorem~\ref{th: second law} and the computational analog of~\cite[Theorem 1]{W21} (Proposition~\ref{prop: extended W21}). First, we address the analytical aspect of the comparison. That is, we identify the regime, i.e., a combination of errors of distillation and dilution, in which the additive term in Theorem~\ref{th: second law} ($t_{\sqrt{\varepsilon_1},\varepsilon_2,d}$) is strictly smaller than the counterpart in Proposition~\ref{prop: extended W21} ($g_{\varepsilon_1,\varepsilon_2}$). For this, we compare the admissible ranges of distillation ($\varepsilon_2$) and dilution ($\varepsilon_1$) errors. Recall that Theorem~\ref{th: second law} admits $\varepsilon_1$ and $\varepsilon_2$ which obey the following relation
\begin{equation}\label{eq: comp 1}
    \forall \lambda\in\mathbb N,\quad \varepsilon_2(\lambda) + \sqrt{\varepsilon_1(\lambda)} < 1-2^{-2d(\lambda)},
\end{equation}
where $d(\lambda)$ is the valid lower bound on computational one-shot distillable entanglement of the given family. On the other hand, the corresponding condition in Proposition~\ref{prop: extended W21} takes the form
\begin{equation}\label{eq: comp 2}
     \forall \lambda\in\mathbb N,\quad \left[\sqrt{\varepsilon_1(\lambda)} +\sqrt{\varepsilon_2(\lambda)}  \right]^2 < 1.
\end{equation}
Comparing Equation~\eqref{eq: comp 1} to Equation~\eqref{eq: comp 2}, we see that, in the case of $2^{-2d(\lambda)}<\sqrt{\varepsilon_1(\lambda)}(1-\sqrt{\varepsilon_1(\lambda)})$ Theorem~\ref{th: second law} admits a wider range of error values than Proposition~\ref{prop: extended W21}, which follows from $\sqrt{\varepsilon_2(\lambda)}>\varepsilon_2(\lambda)$ for all $\lambda\in\mathbb N$. The latter holds for all $\varepsilon_2(\lambda)\in(0,1)$. 

To demonstrate such an advantage, we pick the error of distillation $\varepsilon_2(\lambda):=(1-\sqrt{\varepsilon_1(\lambda)})^2$ and the error of dilution $\varepsilon_1(\lambda):=1/\lambda$. Notice that such an error choice is not compatible with Equation~\eqref{eq: comp 2}, and therefore cannot be used in Proposition~\ref{prop: extended W21} as it makes $g_{\varepsilon_1,\varepsilon_2}$ diverge. However, for the valid lower bound $d(\lambda):=\lambda$, Theorem~\ref{th: second law} provides a finite bound on the gap between computational one-shot distillable entanglement and cost. We show this in Figure~\ref{fig: comparison}.

\begin{figure}[h!]
    \centering
    \includegraphics[width=0.7\linewidth]{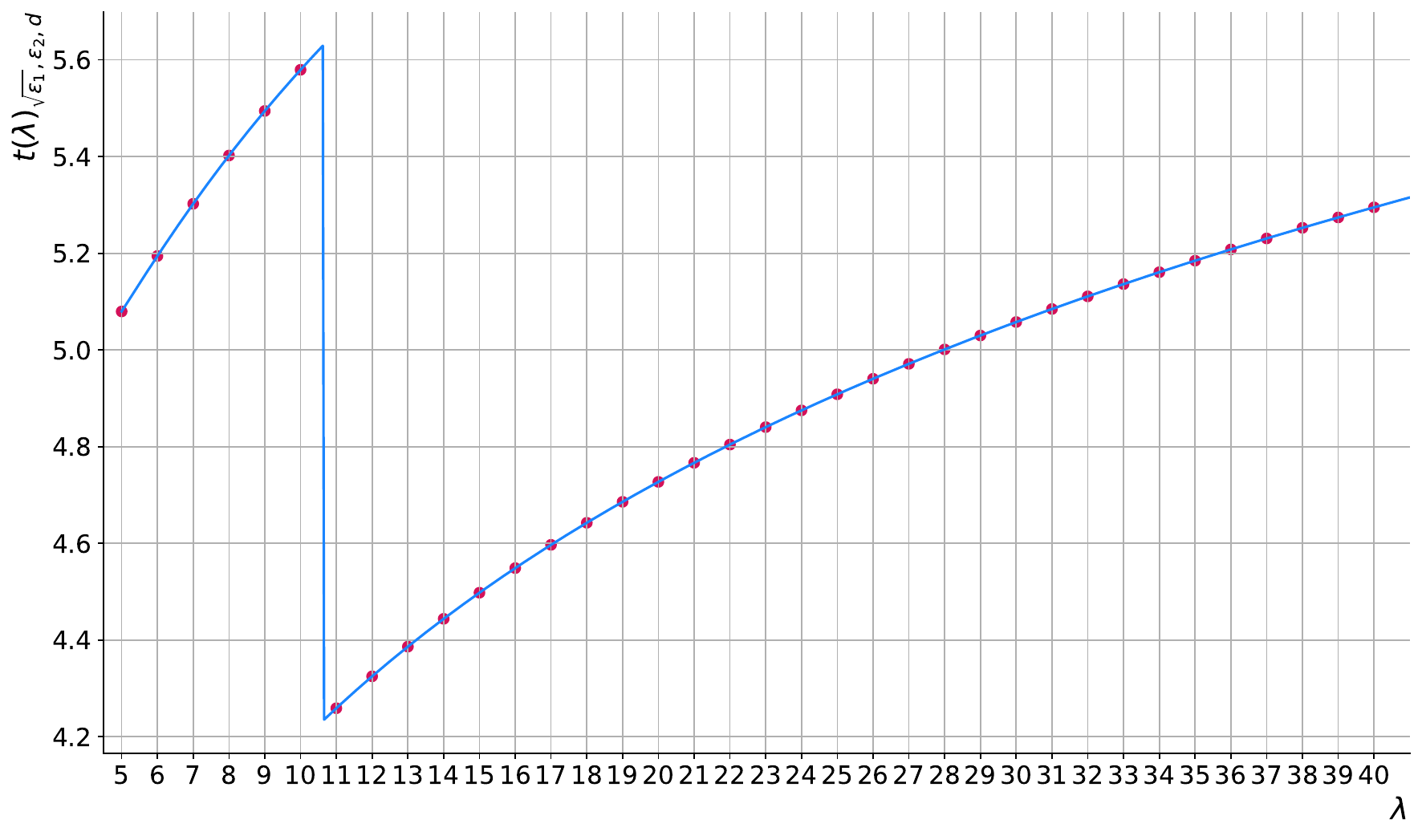}
    \caption{The additive term $t_{\sqrt{\varepsilon_1},\varepsilon_2,d}$. Error of dilution $\varepsilon_1(\lambda)=1/\lambda$. Error of distillation $\varepsilon_2(\lambda)=(1-\sqrt{\varepsilon_1(\lambda)})^2$. Valid lower bound on efficiently distilled EPR pairs $d(\lambda)=\lambda$.}
    \label{fig: comparison}
\end{figure}

Next, we investigate whether the correction term of Theorem~\ref{th: second law} can be strictly smaller than the one in Proposition~\ref{prop: extended W21} in the regime when both Equation~\eqref{eq: comp 1} and Equation~\eqref{eq: comp 2} are satisfied. For this, we solve the inequality $t_{\sqrt{\varepsilon_1},\varepsilon_2,d}<g_{\varepsilon_1,\varepsilon_2}$ for some fixed parameters $\varepsilon_1,\varepsilon_2,d$. Performing algebraic transformations gives the following
\begin{equation}\label{eq: comp 3}
    2^{-\min\{s_{\sqrt{\varepsilon_1},\varepsilon_2,d},d\}}(1-\varepsilon_2)(1-P_{\text{fail}}) > 1- \left[\sqrt{\varepsilon_1} + \sqrt{\varepsilon_2}\right]^2,
\end{equation}
where $P_{\text{fail}}$ is the upper bound on the failure probability of distillation protocols in Theorem~\ref{th: second law}. As an example, we solve this inequality numerically for fixed $\varepsilon_1=0.03$ and $d=7$ using Python. For our choice of parameters, the interval of $\varepsilon_2$ over which the inequality in Equation~\eqref{eq: comp 3}  holds is as follows 
\begin{equation}
    \varepsilon_2\in(0.6569,0.6835).
\end{equation}
Fixing $(\varepsilon_1,\varepsilon_2,d)=(0.03,0.67,7)$ and computing $t_{\sqrt{\varepsilon_1},\varepsilon_2,d}$ and $g_{\varepsilon_1,\varepsilon_2}$ yields the following values
\begin{equation}
    t_{\sqrt{\varepsilon_1},\varepsilon_2,d} \approx 5.0434\quad\text{and}\quad g_{\varepsilon_1,\varepsilon_2} \approx 5.9257,
\end{equation}
which gives a numerical advantage of the bound obtained in Theorem~\ref{th: second law} over that of~\cite{W21} in this example. In general, for large enough $d$, picking $\varepsilon_2\in([1-\sqrt{\varepsilon_1}]^2-\eta,[1-\sqrt{\varepsilon_1}]^2)$ for sufficiently small $\eta>0$ will yield $t_{\sqrt{\varepsilon_1},\varepsilon_2,d}  <g_{\varepsilon_1,\varepsilon_2} $ due to the term $g_{\varepsilon_1,\varepsilon_2}$ diverging at $\varepsilon_2=[1-\sqrt{\varepsilon_1}]^2$ and the term $t_{\sqrt{\varepsilon_1},\varepsilon_2,d}$ remaining bounded.

Furthermore, we demonstrate that in the regime of vanishing errors of distillation and dilution, the obtained bound behaves similarly to that of~\cite{W21}, with the latter showing a numerical advantage. For this, we fix the error of distillation $\varepsilon_2=1/\lambda$ and the error of dilution $\varepsilon_1=1/\lambda^4$. Then, we observe the behavior as shown in Figure~\ref{fig: comparison 2}. Note that due to the operational nature of the proof of Theorem~\ref{th: second law}, the additive term $t_{\sqrt{\varepsilon_1},\varepsilon_2,d}$ approaches one instead of zero. However, as discussed after Theorem~\ref{th: second law}, the extensions of~\cite[Theorem 1]{W21} are expected to be loose in general due to the gap between values of computational and information-theoretic measures. Therefore, in the case in which either of the two measures is separated from their information-theoretic counterparts in the regime of sufficiently large  $\lambda$, the lower bound on the slack in Equation~\eqref{eq: add 1}  approaches one. To compare the decay rate of both terms, we artificially add one to the term $g_{\varepsilon_1,\varepsilon_2}$ and plot the resulting curves in  Figure~\ref{fig: comparison 2}. We observe that the term $g_{\varepsilon_1,\varepsilon_2}$ has a numerical advantage, i.e., converges to its infimum faster. 

\begin{figure}[h!]
    \centering
    \includegraphics[width=0.7\linewidth]{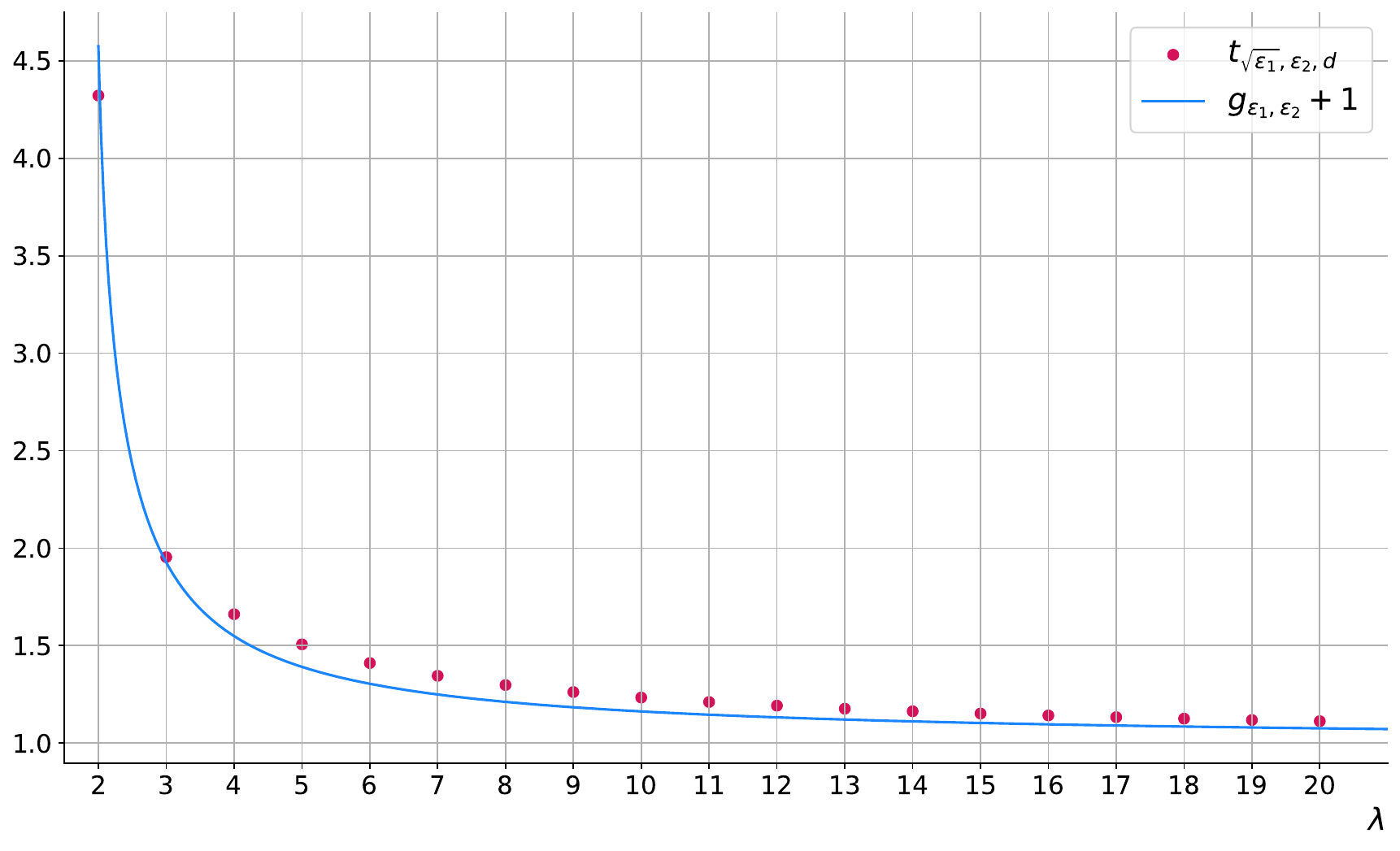}
    \caption{Comparison of additive terms $t_{\sqrt{\varepsilon_1},\varepsilon_2,d}$ and $g_{\varepsilon_1,\varepsilon_2}$ with added minimal value of slack (1). Error of dilution $\varepsilon_1(\lambda)=1/\lambda^4$. Error of distillation $\varepsilon_2(\lambda)=1/\lambda$. Number of distilled EPR pairs $d(\lambda)=\lambda$.}
    \label{fig: comparison 2}
\end{figure}

We note that, if both $\sqrt{\varepsilon_1}/\varepsilon_2$ and $d$ diverge to infinity, but $\varepsilon_1$ and $\varepsilon_2$ are vanishing functions, then the result of~\cite{W21} outperforms the result established in Theorem~\ref{th: second law}. That is, the additive term $g_{\varepsilon_1,\varepsilon_2}$ converges to its infimum, while $t_{\sqrt{\varepsilon_1},\varepsilon_2,d}$ does not. This is due to the fact that the asymptotic behavior of the term $t_{\sqrt{\varepsilon_1},\varepsilon_2,d}$ is dictated by the form of  $\sqrt{\varepsilon_1}/\varepsilon_2$ (see discussion after Theorem~\ref{th: second law}). In the described setting, since $\sqrt{\varepsilon_1}/\varepsilon_2$ diverges to infinity, the corresponding term  $t_{\sqrt{\varepsilon_1},\varepsilon_2,d}$  follows suit. 

Finally, we address the operational aspect of the comparison. That is, we now highlight the operational nature of the proof of Theorem~\ref{th: second law}. We recall that in this proof an efficient family of distillation protocols was constructed that shows the bound in Theorem~\ref{th: second law}. In contrast, the information-theoretic result in~\cite[Theorem~1]{W21} optimizes over unrestricted LOCC protocols, which are not required to be efficiently implementable. Consequently, as previously discussed, this can introduce substitution slack into the inequality  in Proposition~\ref{prop: extended W21} whenever the relevant information-theoretic transformations cannot be implemented efficiently.  With this in mind, we conclude that the result obtained in this work (Theorem~\ref{th: second law}) complements the one in~\cite[Theorem 1]{W21} and provides, to the best of our knowledge, the first efficiency-preserving, protocol-level way to establish the relation between computational one-shot distillable entanglement and cost. 

\section{Conclusion}

In this work, we have investigated how entanglement behaves under a wide range of non-asymptotic manipulations when the computational efficiency of the underlying protocols is taken into account. For this, we have systematically analyzed the mathematical structure of computational one-shot distillable entanglement and entanglement cost, which serve as central objects in the recently established computational entanglement framework~\cite{arnon2023computational}. To accommodate these measures, which are defined through functional lower and upper bounds, we developed a framework of efficient entanglement manipulation by extending familiar properties of conventional scalar entanglement measures.

With this framework in place, we established basic properties of efficiently accessible entanglement under state-mixing and tensor-product operations. Namely, we showed that computational one-shot distillable entanglement is lower-bound convex and lower-bound superadditive, while computational one-shot entanglement cost is upper-bound concave and upper-bound subadditive.

We then revisited the problem of establishing unconditional existential separations between information-theoretic and computational one-shot distillable entanglement, which was previously considered in~\cite{arnon2023computational}. Building on the counting argument used there, we developed a quantitative connection between the geometrical ``density'' of specially constructed families of states and their efficiently distillable entanglement, which we used to demonstrate an unconditional separation between computational and information-theoretic one-shot distillable entanglement. Using this separation, we further showed that computational entanglement measures are invariant under efficient families of local unitaries and monotone under efficient families of LOCC channels, while the corresponding properties need not hold for arbitrary families of transformations. For computational one-shot entanglement cost, this establishes the corresponding local-unitary behavior under explicit error conditions and extends the phenomenon suggested by the construction in \cite[Lemma 4.1]{arnon2023computational}.

We also investigated the robustness of computationally distillable entanglement under perturbations of the original family of states. The continuity bound in Theorem~\ref{thm: val-cont E_d} gives an explicit lower bound on the number of EPR pairs that remain efficiently distillable when the original family is perturbed in trace distance. To establish this result, we constructed an efficient family of protocols that preserves the target error conditioned on success after sacrificing an explicitly determined number of EPR pairs, with a corresponding bound on the failure probability. We then numerically illustrated how the performance of these protocols depends on the magnitude of the perturbation.

Finally, as an application of the obtained ``dictionary'' of properties, we established a quantitative relation between computational one-shot distillable entanglement and cost, which we call the ``Second Law of Efficient Entanglement Manipulation.'' This result extends the type of finite-error relation obtained in~\cite{W21} to the computational setting and provides a direct relation between computational dilution and distillation bounds even when the corresponding computational and information-theoretic quantities do not coincide.

From the perspective of large-scale quantum technologies, our results distinguish the entanglement contained in a quantum state from the entanglement that computationally constrained devices can efficiently access and transform into a standard resource. The results on mixing and tensor products describe probabilistic and parallel combinations of entanglement resources, while the results on local-unitary invariance and LOCC monotonicity identify the transformations under which computational entanglement retains familiar resource-theoretic properties. The continuity bound and the computational second law further quantify the robustness and conversion of efficiently accessible entanglement. Taken together, these results provide a resource-theoretic foundation for studying entanglement manipulation under the computational restrictions arising in large quantum processors, fault-tolerant architectures, and quantum networks.

At the same time, the framework developed in this work does not constitute a complete model of a particular quantum device or architecture. In particular, it does not explicitly account for hardware connectivity, gate-dependent noise, memory decoherence, compilation overhead, scheduling latency, or hardware-dependent circuit depth. Developing hardware-aware variants of the computational measures, in which polynomial circuit-size restrictions are supplemented by explicit constraints on these quantities, therefore constitutes a natural direction for future work. In the setting of quantum networks, this may include combining computational entanglement bounds with multi-hop routing and purification scheduling~\cite{LLLC21,CERW20,XLX24}. More generally, incorporating experimentally relevant control and resource constraints may provide a connection between computational entanglement theory and the analysis of scalable quantum devices.

One natural fundamental direction concerns the extension of the computational entanglement framework to the asymptotic i.i.d. setting. Leone et al.~\cite{leone2025entanglement} developed $k$-copy and asymptotic computational measures for bipartite pure states manipulated by efficient stochastic LOCC protocols. Extending a comparable operational characterization to mixed-state sources and broader classes of transformations remains an important direction. Such a framework would make it possible to investigate the asymptotic and thermodynamic behavior of computational entanglement.

Finally, another property not considered in this work is the monogamy of computational entanglement. Given the relevance of entanglement monogamy to cryptography~\cite{TF13}, it is natural to investigate whether analogous relations hold for computational one-shot entanglement measures, whose definitions incorporate both operational and computational restrictions. 

\pagebreak
\bibliographystyle{alpha}
\bibliography{main}

\pagebreak
\appendix
\section{Construction of $S^\lambda$ from $\eta$-nets on Projective Unitary and Special Orthogonal Groups}\label{app: cons s lambda}
In this appendix, we formalize the iterative construction of a special set of states $S^\lambda$, motivated by the family introduced in~\cite{arnon2023computational}. We do so for the setting in which $S^\lambda$ is generated by a quotient of the unitary group $\operatorname{P\mathbb U}(2^{m(\lambda)}):=\mathbb U(2^{m(\lambda)})/\mathbb U(1)$. We then extend this algorithm to the case in which $S^\lambda$ is generated by a quotient of the special orthogonal group $\operatorname{PSO}(2^{m(\lambda)}):=\operatorname{SO}(2^{m(\lambda)})/\{\mathbb I^\lambda,-\mathbb I^\lambda\}$, where $\mathbb I^\lambda$ is the identity operator of dimension $2^{m(\lambda)}$. In the following, we denote the resulting families by $S^\lambda_{\operatorname{P\mathbb{U}}}$ and $S^\lambda_{\operatorname{PSO}}$. 

We first formalize the construction of $S^\lambda_{\operatorname{P\mathbb{U}}}$. The process starts with $S^\lambda_0=\emptyset$. For every $i>0$, let $S^{\lambda}_i$ denote the set of states obtained after the $i^{\text{th}}$ ($i>0$) iteration:
\begin{equation*}
    S^\lambda_i := \left\{|\Phi_{[U]}^{\lambda}\rangle\,|\,
 [U] \in\mathcal U^{(i)} \right\},
 \end{equation*}
 where $|\Phi_{[U]}^{\lambda}\rangle:=\mathbb I\otimes U|{\phi^{\otimes m(\lambda)}}\rangle$ for a chosen representative $U\in[U]$ and $\mathcal U^{(i)}\subseteq\operatorname{P\mathbb U}(2^{m(\lambda)})$, in turn, is defined as
\begin{equation}\label{eq: u_i construction}
    \mathcal U^{(i)}:= \mathcal U^{(i-1)}\cup\{[V]\} \quad\text{where}\quad [V]\in\operatorname{P\mathbb U}(2^{m(\lambda)}):\sqrt{F\left(|\psi^{\lambda}\rangle, |\Phi_{[V]}^{\lambda}\rangle\right)} \leq 1-\eta,\, \forall  |\psi^{\lambda}\rangle\in S^\lambda_{i-1},
\end{equation}
with $\mathcal U^{(0)}:=\emptyset$ and pure-state fidelity defined as $F\left(|\psi\rangle, |\phi\rangle\right)=|\langle\psi|\phi\rangle|^2$.

\noindent
 Furthermore, define the sets $S^\lambda$ and $\mathcal U^{\lambda}$ as
\begin{equation*}
    S^\lambda_{\operatorname{P\mathbb U}} := S^\lambda_{i_{\eta}},
\quad \mathcal{U}^{\lambda}:=\mathcal U^{(i_{\eta})}.
\end{equation*}
In the following, we argue that the iterative construction of $S^{\lambda}$ is finite ($i_{\eta}<\infty$).
First, recall that, at each nonterminal iteration of the construction of $S^{\lambda}$, there exists $[V]\in\operatorname{P\mathbb U}(2^{m(\lambda)})$ such that
\begin{equation}\label{eq: equivalence between net and states}
\sqrt{F(|\psi^{\lambda}\rangle,|\Phi_{[V]}^{\lambda}\rangle)}\leq1-\eta,\,\forall|\psi^{\lambda}\rangle\in S^\lambda_{i-1}\iff \frac{1}{2^{m(\lambda)}}\Delta_2([U],[V])^2 \geq 2\eta,\,\forall [U]\in\mathcal U^{(i-1)},
\end{equation}
where the distance $\Delta_2([U],[V])$ is defined as follows
\begin{equation}\label{eq: distance measure. app}
    \forall[U],[V]\in\operatorname{P\mathbb U}(2^{m(\lambda)}),\quad \Delta_2([U],[V]):= \min_{\theta\in\mathbb R}\|U-e^{i\theta}V\|_2.
\end{equation}
The equivalence in Equation~\eqref{eq: equivalence between net and states} is obtained by observing that, for all $\lambda\in\mathbb N$, the following holds
\begin{equation}\label{eq: overlap}
\frac{1}{2^{m(\lambda)}}\min_{\theta\in\mathbb R}\|U-e^{i\theta}V\|_2^2 = 2\left(1-\frac{|\operatorname{Tr}(U^\dagger V)|}{2^{m(\lambda)}}\right),
\end{equation}
and that $|\langle\Phi_{[U]}|\Phi_{[V]}\rangle| = 2^{-m(\lambda)}|\operatorname{Tr}(U^\dagger V)|$.

Equation~\eqref{eq: equivalence between net and states} shows that, at every iteration, the elements of $\mathcal U^{(i)}$ have pairwise normalized distance at least $\eta^\prime:=\sqrt{2\eta}$. Here, we use the convention that a set is $\eta^\prime$-separated if the normalized distance between every pair of distinct elements is at least $\eta^\prime$. Since ($\operatorname{P\mathbb U}(d),\Delta_2$) is a compact metric space~\cite{szarek1997metric}, it is totally bounded. Consequently, every $\eta^\prime$-separated subset is finite for every fixed $\eta^\prime>0$. It follows that the iterative construction terminates after finitely many steps, i.e., $i_\eta<\infty$. If the construction did not terminate, then $\bigcup_{i\geq 0}\mathcal U^{(i)}$ would be an infinite \(\eta'\)-separated subset of \(\operatorname{P\mathbb U}(d)\), contradicting total boundedness.

At the terminal iteration, it holds that for any $[V] \in \operatorname{P\mathbb U}(2^{m(\lambda)})$, the following applies:
\begin{equation*}
\exists [U]\in\mathcal U^\lambda,\quad \frac{1}{2^{m(\lambda)/2}}\Delta_2([U],[V])<
\eta'.
\end{equation*}
Indeed, otherwise $[V]$ could be added according to Equation~\eqref{eq: u_i construction}, contradicting terminality. Since the strict inequality above implies the non-strict inequality required in Definition~\ref{def: eta-nets on PU/PSO}, the set $\mathcal U^\lambda$ is an $\eta'$-net in $\operatorname{P\mathbb U}(2^{m(\lambda)})$.

The algorithm for constructing $S^\lambda_{\operatorname{PSO}}$ consists of the same steps. The main difference lies in replacing the quotient of the unitary group $\operatorname{P\mathbb U}(2^{m(\lambda)})$ in Equation~\eqref{eq: u_i construction} with a quotient of the special orthogonal group by $\{\mathbb I^\lambda,-\mathbb I^\lambda\}$. Moreover, since $\operatorname{PSO}(2^{m(\lambda)})$ consists of equivalence classes of matrices with real entries, the resulting $S^\lambda_{\operatorname{PSO}}$ consists of real unit vectors. Restricting the distance metric introduced in Equation~\eqref{eq: distance measure. app} to $\operatorname{PSO}$ results in the following analogue of  Equation~\eqref{eq: overlap}:
\begin{equation*}
    \frac{1}{2^{m(\lambda)}}\min\{\|U-V\|_2^2, \|U+V\|_2^2\} = 2\left(1 -\frac{|\operatorname{Tr}(U^TV)|}{2^{m(\lambda)}}\right).
\end{equation*}
\end{document}